\documentclass[conference]{IEEEtran}
\usepackage{amsmath,amsfonts,amssymb,amsthm}
\usepackage{multicol,array}
\usepackage{graphicx}
\usepackage{subcaption}
\renewcommand{\thesection}{\arabic{section}}
\theoremstyle{definition}
\newtheorem{defn}{\textit{Definition}}
\newtheorem{cons}{\textit{Construction}}

\newtheorem{algo}{\textit{Algorithm}}
\newtheorem{thm}{\textit{Theorem}}

\newtheorem{ex}{\textit{Example}}
\newtheorem{exmp}{\textit{Example}}
\newtheorem{lm}{\textit{Lemma}}

\newtheorem*{clm}{\textit{Claim}}
\theoremstyle{remark}
\newtheorem{remark}{Remark}
\newtheorem{cond}{\textit{Condition}}
\newtheorem*{v1}{Verification of c$ 1 $ and c$ 2 $}

\newtheorem{part}{Part}

\newtheorem*{case}{\textit{Case}}

\title{On Index Codes for Interlinked Cycle Structured Side-Information Graphs}
\author{K. Vikas Bharadwaj and B. Sundar Rajan, {\it Fellow, IEEE}}
	
\begin{document}
		
\maketitle
\begin{abstract}
	In connection with the index code construction and the decoding algorithm for interlinked cycle (IC) structures proposed by Thapa, Ong and Johnson in \cite{TOJ} ("Interlinked Cycles for Index Coding: Generalizing Cycles and Cliques", IEEE Trans. Inf. Theory, vol. 63, no. 6, Jun. 2017), it is shown in \cite{VaR} ("Optimal Index Codes For A New Class of Interlinked Cycle Structure", in \textit{IEEE Communication Letters,} available as early access article in \textit{IEEE Xplore}: DOI-10.1109/LCOMM.2018.2799202) that the decoding algorithm does not work for all IC structures.  In this work,  a set of necessary and sufficient conditions on the IC structures is presented for the decoding algorithm to work for the code construction given in \cite{TOJ}. These conditions are shown to be satisfied for the IC structures without any cycles consisting of only non-inner vertices. \footnote{The authors are with the Department of Electrical Communication Engineering, Indian Institute of Science, Bangalore-560012, India. Email:bsrajan@iisc.ac.in} 
\end{abstract}
\section{Introduction}
The problem of index coding was introduced by Birk and Kol in \cite{BiK}. 
The index coding problem consists of a single sender with a set of \textit{M} independent messages
			\begin{displaymath}
			\mathcal{X}=\lbrace x_1,x_2,\dots,x_M\rbrace,
			\end{displaymath} and a set of $ N $ users
			\begin{displaymath}
				\mathcal{D}=\lbrace D_1,D_2,\dots,D_N \rbrace,
			\end{displaymath}connected to the sender by a single shared error-free link, with the $ k^{th} $ user $ D_k $ identified as 
			\begin{displaymath}
				D_k=(\mathcal{X}_k,\mathcal{A}_k),
			\end{displaymath}where $ \mathcal{X}_k \subseteq \mathcal{X}$ is the set of messages desired by $ D_k $, the set $ \mathcal{A}_k \subset \mathcal{X}$ is comprised of the messages available to user $ D_k $ as side-information. The set of side-informations $ \mathcal{A}_k $ satisfies $ \mathcal{X}_k \cap \mathcal{A}_k=\phi$, i.e., a user does not desire a message that is already available to it.

An $ (\mathcal{S},n,\mathcal{R}) $ index coding scheme \cite{MCJ} corresponds to the choice of a finite alphabet $ \mathcal{S} $ of cardinality $ |\mathcal{S}| > 1 $, a coding function, $ f $, and a decoding function $ g_{k,i} $, for each desired message $ x_i $ at each user $ D_k $. The coding function maps all the messages to the sequence of transmitted symbols
			\begin{displaymath}
				f(x_1,x_2,\dots,x_M)=S^n
			\end{displaymath}where $ S^n \in \mathcal{S}^n $ is the sequence of symbols transmitted over $ n $ channel uses. Here $ \forall m \in \lbrace1,2,\dots,M\rbrace $, message $ x_m $ is a random variable uniformly distributed over the set
			\begin{displaymath}
				x_m \in \lbrace1,2,\dots,|\mathcal{S}|^{nR_m}\rbrace,
			\end{displaymath}and $ \mathcal{R} \in \mathbb{R}^M_+ $ is simply a rate vector
			\begin{displaymath}
				\mathcal{R}=(R_1,R_2,\dots,R_M)
			\end{displaymath}that satisfies the condition that $ |\mathcal{S}|^{nR_m} $ is an integer for every $ m \in \lbrace1,2,\dots,M\rbrace $. At each user, $ D_k $, there is a decoding function for each desired message
			\begin{displaymath}
				g_{k,i}(S^n,\mathcal{A}_k)=x_i,
			\end{displaymath}for all $ i $ such that $ x_i\in \mathcal{X}_k $.\\An index coding scheme is said to be a linear index coding scheme if the coding and the decoding functions are linear and the alphabet $ \mathcal{S} $ is a finite field. An index coding scheme is said to be a scalar index coding scheme if 
			\begin{displaymath}
				\mathcal{R}=\left(\frac{1}{n},\frac{1}{n},\dots,\frac{1}{n}\right).
			\end{displaymath}In other words, in a scalar index coding scheme, the sender sends one symbol for each message over $ n $ channel uses. $ n $ is referred to as the length of the index code.
			
			An index coding problem is said to be unicast \cite{OnH} if $ \mathcal{X}_k \cap \mathcal{X}_j = \phi $ for $ k\not= j $ and $ k,j\in \lbrace1,2,\dots,N\rbrace $, i.e., no message is desired by more than one user. The problem is said to be single unicast if the problem is unicast and $ |\mathcal{X}_k|=1$ for all $ k\in \lbrace1,2,\dots,N\rbrace $. A unicast index coding problem can be reduced into single unicast index coding problem, by splitting the user demanding more than one message into several users, each demanding one message and with the same side-information as the original user. For example, let there are $ 5 $ messages at the sender, $ \lbrace x_1,x_2,x_3,x_4,x_5\rbrace $. A user demanding three messages $ x_1 $, $ x_2 $ and $ x_3 $ and with side-information $ x_4 $ and $ x_5 $ is split into three users, each with side-information $ x_4 $ and $ x_5 $ and demanding one message $ x_1 $, $ x_2 $ and $ x_3 $ respectively.
			
Single unicast index coding problems can be described by a directed graph called a side-information graph \cite{BBJK}, in which the vertices in the graph represent the indices of messages $ \lbrace x_1,x_2,\dots,x_M\rbrace $ and there is a directed edge from vertex $ i $ to vertex $ j $ if and only if the user requesting $ x_i $ has $ x_j $ as side-information. 

The set of vertices in a directed graph $ \mathcal{G} $ is denoted by $ V(\mathcal{G}) $ and the set of vertices in the out-neighbourhood of a vertex $q$ in $\mathcal{G}$ is denoted by $N^{+}_{\mathcal{G}}(q)$.

Interlinked Cycle Cover (ICC) scheme is proposed as a scalar linear index coding scheme to solve unicast index coding problems by Thapa et al. \cite{TOJ}, by defining a graph structure called an Interlinked Cycle (IC) structure.
\begin{defn}[\textbf{IC Structure} \cite{TOJ}]
A side-information graph $\mathcal{G}$ is called a $K$-IC structure with inner vertex set $V_{I}$ $\subseteq V(\mathcal{G})$, such that $|V_{I}| = K$ if $ \mathcal{G} $ satisfies the following three conditions.
\begin{enumerate}
\item There is no I-cycle in $\mathcal{G}$, where an I-cycle is defined as a cycle which contains only one inner vertex.
\item There is a unique I-path between any two different inner vertices in $\mathcal{G}$, where an I-path is defined as a path from one inner vertex to another inner vertex without passing through any other inner vertex (as a result, $K$ rooted trees can be drawn where each rooted tree is rooted at an inner vertex and has the remaining inner vertices as the leaves).
\item $\mathcal{G}$ is the union of the $K$ rooted trees.
\end{enumerate}
The set of the vertices $ V(\mathcal{G})\backslash V_I $ is called the set of non-inner vertices, denoted by $ V_{NI} $.
Let the $K$-IC structure, $\mathcal{G}$, have inner vertex set $V_I=\lbrace 1,2,\dots,K\rbrace$ and non-inner vertices $ V_{NI}=\lbrace K+1, K+2,\dots,N \rbrace $. Let $T_i$ be the rooted tree corresponding to the inner vertex $i$ where $i \in \lbrace 1,2,\dots,K\rbrace$. Let $V_{NI}(i)$ be the set of non-inner vertices in $\mathcal{G}$ which appear in the rooted tree $T_i$ of an inner vertex $ i $.
\end{defn}

The ICC scheme finds disjoint IC structures in a given side-information graph and then constructs an index code for each IC structure using the following construction proposed by Thapa et al. \cite{TOJ} (stated below as \textit{Construction} $1$).	 
\begin{cons}[An index code construction for IC structures]
	Let the $K$-IC structure be denoted by $\mathcal{G}$ and let $|V(\mathcal{G})|=N$. Let $ V(\mathcal{G})=\lbrace1,2,\dots,N\rbrace$, $V_I=\lbrace1,2,\dots,K)$ be the set of the $K$ inner vertices and hence $ V_{NI}=\lbrace K+1, K+2,\dots,N\rbrace $. Let $x_n \in \mathbb{F}_q$ be the message corresponding to the vertex $n \in V(\mathcal{G})$ and where $ \mathbb{F}_q $ is a finite field with characteristic $ 2 $ and to which the all the $ N $ messages at the sender belong to (note that in single unicast setting, the number of messages will be equal to the number of users).
\end{cons}
\begin{enumerate}
	\item An index code symbol $W_I$ obtained by XOR of messages corresponding to inner vertices is transmitted, where
	\begin{equation}
	W_I=\underset{i=1}{\overset{K}{\bigoplus}} x_i.
	\end{equation}
	\item An index code symbol corresponding to each non-inner vertex, obtained by XOR of message corresponding to the non-inner vertex with the messages corresponding to the vertices in the out-neighbourhood of the non-inner vertex is transmitted, i.e., for $j \in V_{NI}$, $W_j$ is transmitted, where
	\begin{equation}
	W_j=x_j \underset{q \in N^{+}_{\mathcal{G}}(j)}{\bigoplus} x_q,
	\end{equation}
\end{enumerate}where $ \oplus $ denotes modulo addition over $ \mathbb{F}_q $.

	\begin{algo} It is the algorithm proposed in \cite{TOJ} to decode an index code obtained by using \textit{Construction} $ 1 $ on an IC structure, $ \mathcal{G} $.
	\begin{itemize}
		\item The message $x_j$ corresponding to a non-inner vertex $j$ is decoded directly using the transmission $W_j$ and
		\item the message $x_i$ corresponding to an inner vertex $i$ is decoded using
		\begin{displaymath}
			Z_i=W_I \underset{q \in V_{NI}(i)}{\bigoplus}W_q
			\implies
			Z_i=x_i \underset{k:k \in N^{+}_{T_{i}}(i)}{\bigoplus}x_k.
		\end{displaymath}

		\end{itemize}
	\end{algo}

Recently, in \cite{VaR} it has been shown that the index codes obtained from \textit{Construction} $1$ are not necessarily decodable using \textit{Algorithm} $1$ for some IC structures.
 
The contributions of this paper are listed as follows.
\begin{itemize} 
\item The cases where the index code obtained from \textit{Construction} $ 1 $ on the given IC structure is decodable using \textit{Algorithm} $ 1 $ are identified and are presented in \textit{Theorem} $ 1 $.
\item It is shown in \textit{Theorem} $2$ that an IC structure which has no cycles containing only non-inner vertices satisfies the conditions presented in \textit{Theorem} $ 1 $. Thus the proof of optimality of IC structures of Case $ 1 $ of Theorem $ 3 $ in \cite{TOJ} holds.
\item Examples of IC structures for which index code given by \textit{Construction} $ 1 $ is decodable using some other decoding algorithm are presented in the following section.
\item An example of an IC structure for which index code given by \textit{Construction} $ 1 $ is not decodable using any decoding algorithm employing only linear combinations of the index code symbols is presented  (Example \ref{exam8}).
\end{itemize}

The rest of the paper is organized as follows. Section $ 2 $ presents the examples that motivate the results of this paper. Section $ 3 $ discusses the main results along with some illustrating examples. Section $ 4 $ provides the conclusion and the problems that are opened by the results obtained in this paper.
\section{Motivating Examples}

In \cite{VaR} through an example it is shown that index codes obtained from \textit{Construction} $1$ are not necessarily decodable using \textit{Algorithm} $1$ for some IC structures. For a class of such structures the code construction is modified and a decoding algorithm is presented. In this section we present two more examples to show that the codes from \textit{Construction 1} are not decodable using \textit{Algorithm} 1. However, in the following section, after presenting the main results, these two codes are revisited and are shown to be decodable with some other algorithm employing only linear combinations of the index code symbols.
\begin{figure}[!t]
\includegraphics[height=\columnwidth,width=\columnwidth,angle=0]{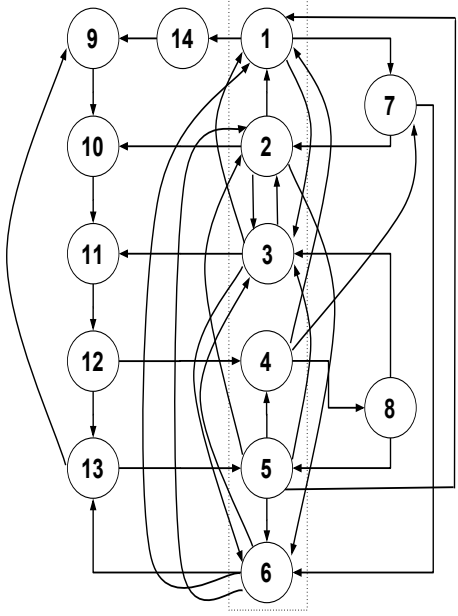}
\caption{$6$-IC structure $\mathcal{G}_1$ with inner vertex set, $V_I=\lbrace1,2,3,4,5,6\rbrace$.}
\label{f1}
\end{figure}
\begin{figure*}[!t]
\centering
\begin{subfigure}{.31\textwidth}
  \hspace{-6mm}
  \includegraphics[width=15pc]{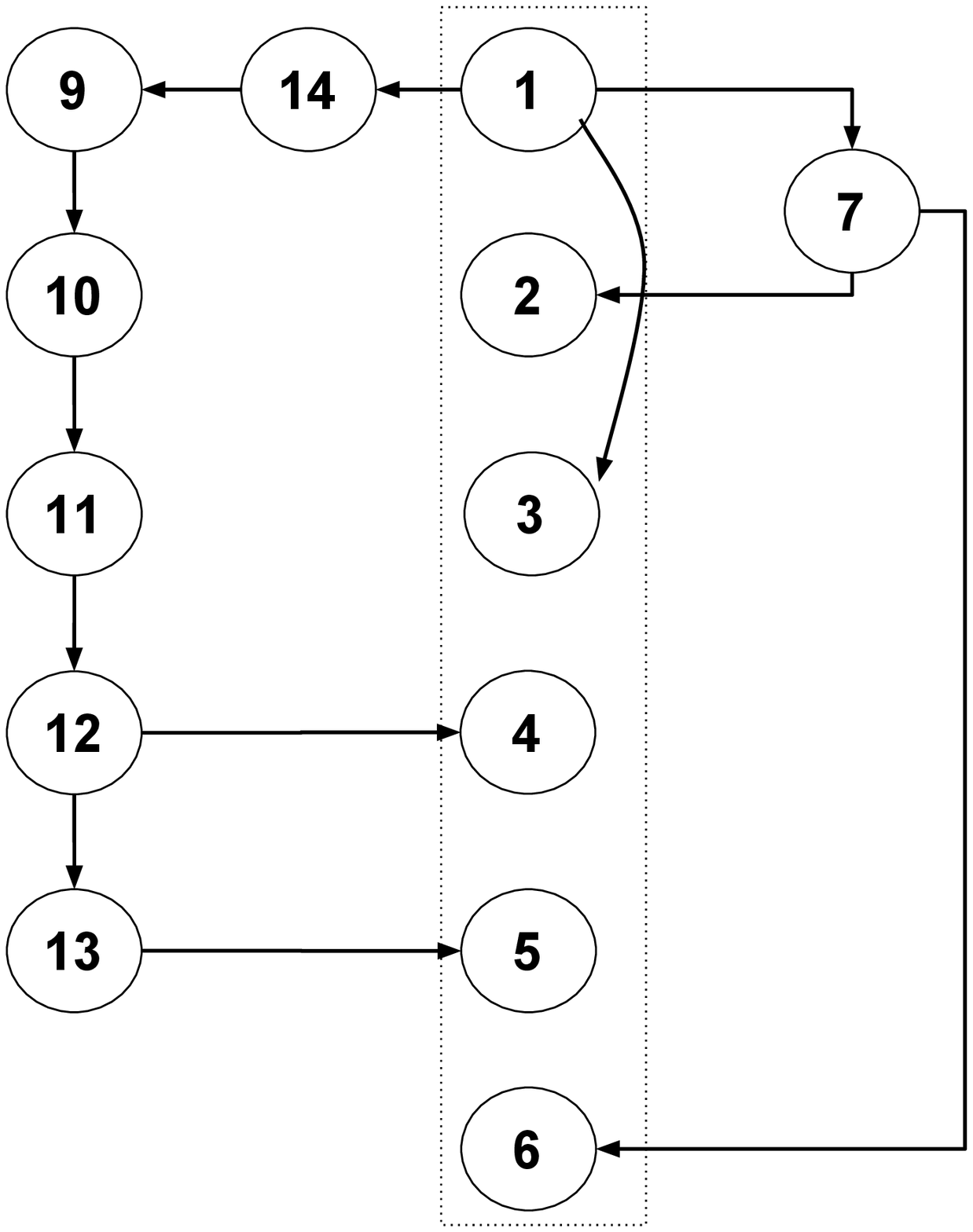}
  \caption{}
  \label{rt11}
\end{subfigure}%
\begin{subfigure}{.31\textwidth}
  \hspace{-6mm}
  \includegraphics[width=15pc]{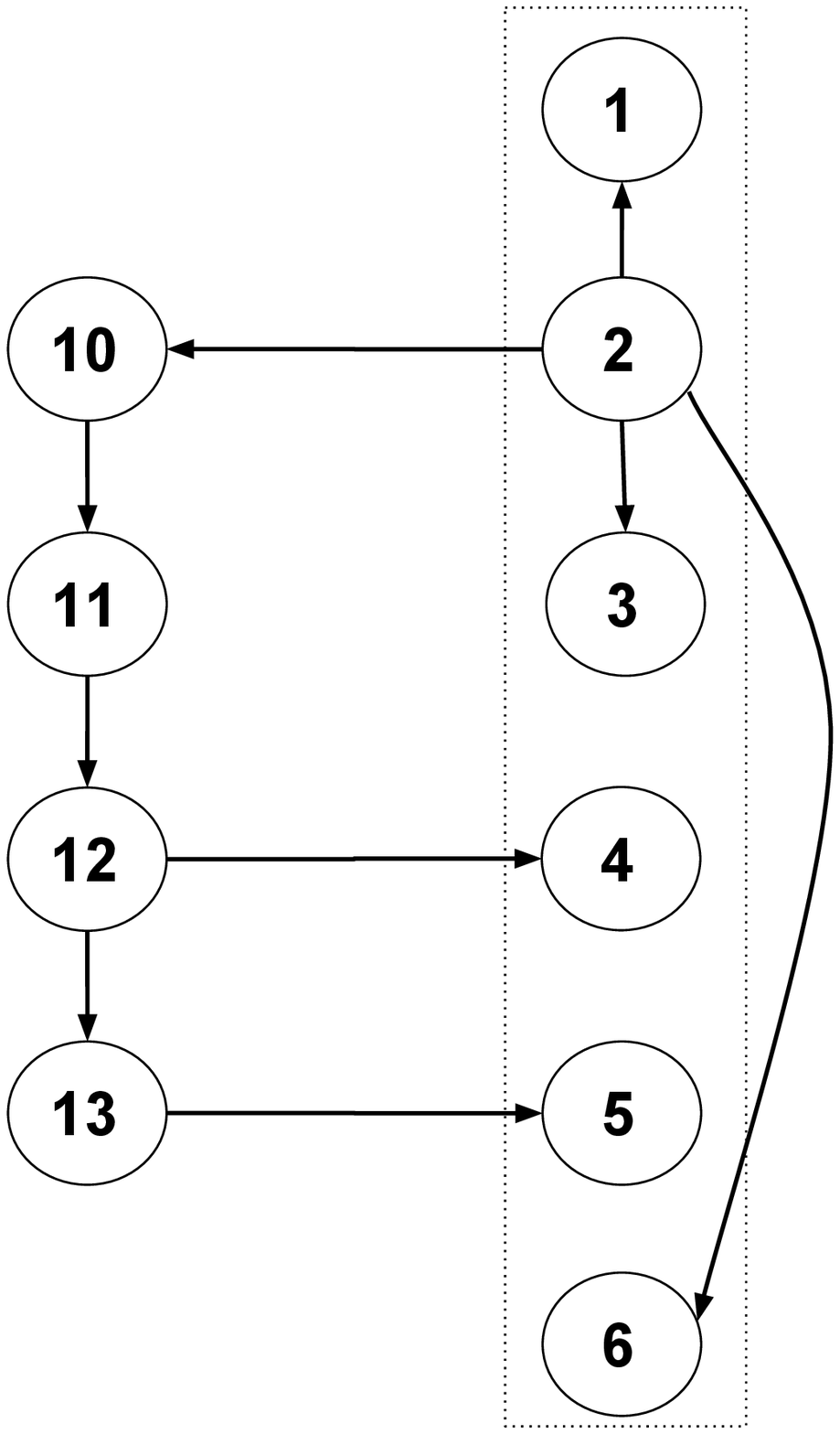}
  \caption{}
  \label{rt12}
\end{subfigure}
\begin{subfigure}{.31\textwidth}
	\hspace{-6mm}
  \includegraphics[width=15pc]{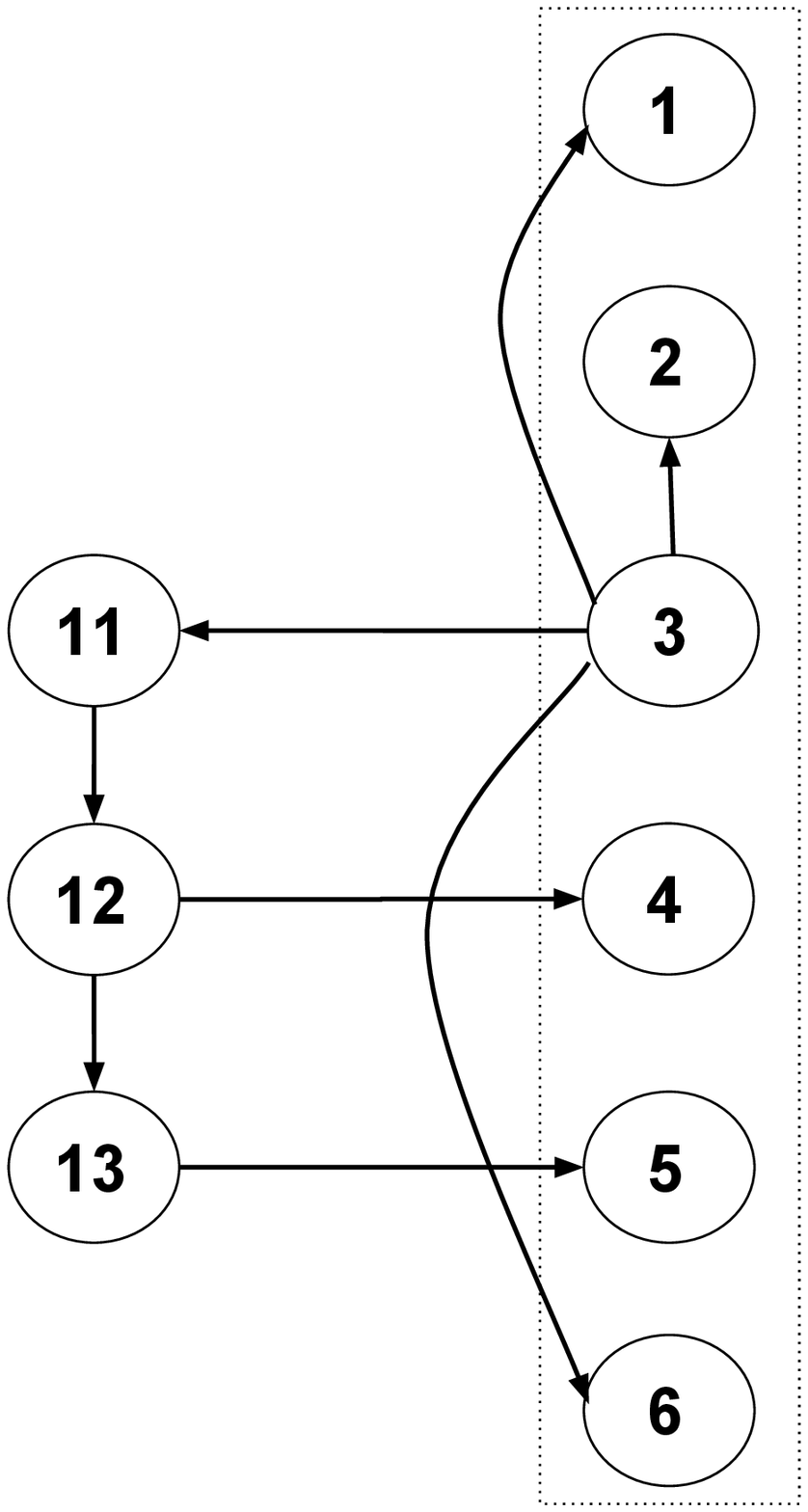}
  \caption{}
  \label{rt13}
\end{subfigure}
\centering
\begin{subfigure}{.31\textwidth}
	\hspace{-6mm}
  \includegraphics[width=15pc]{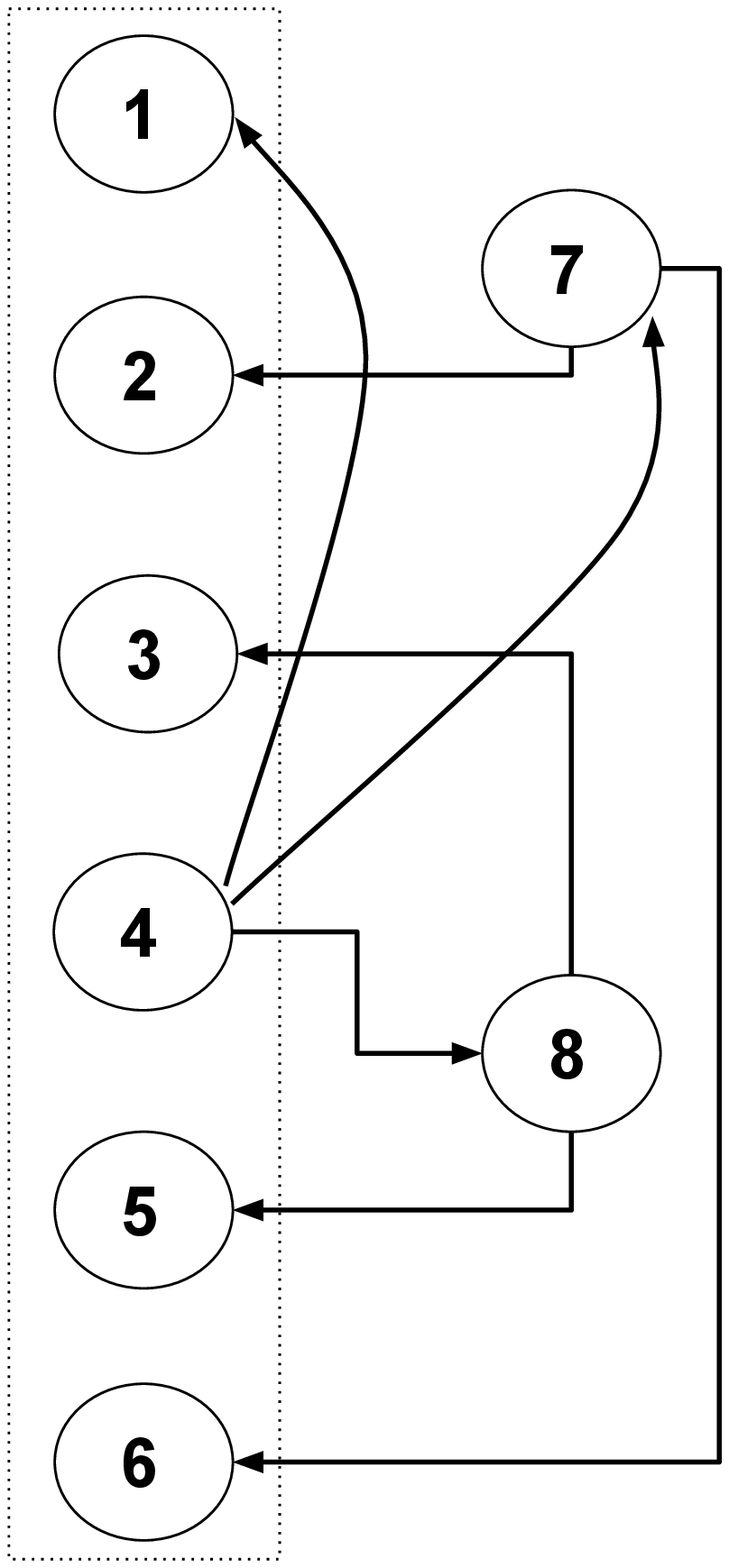}
  \caption{}
  \label{rt14}
\end{subfigure}%
\begin{subfigure}{.31\textwidth}
	\hspace{-6mm}
  \includegraphics[width=15pc]{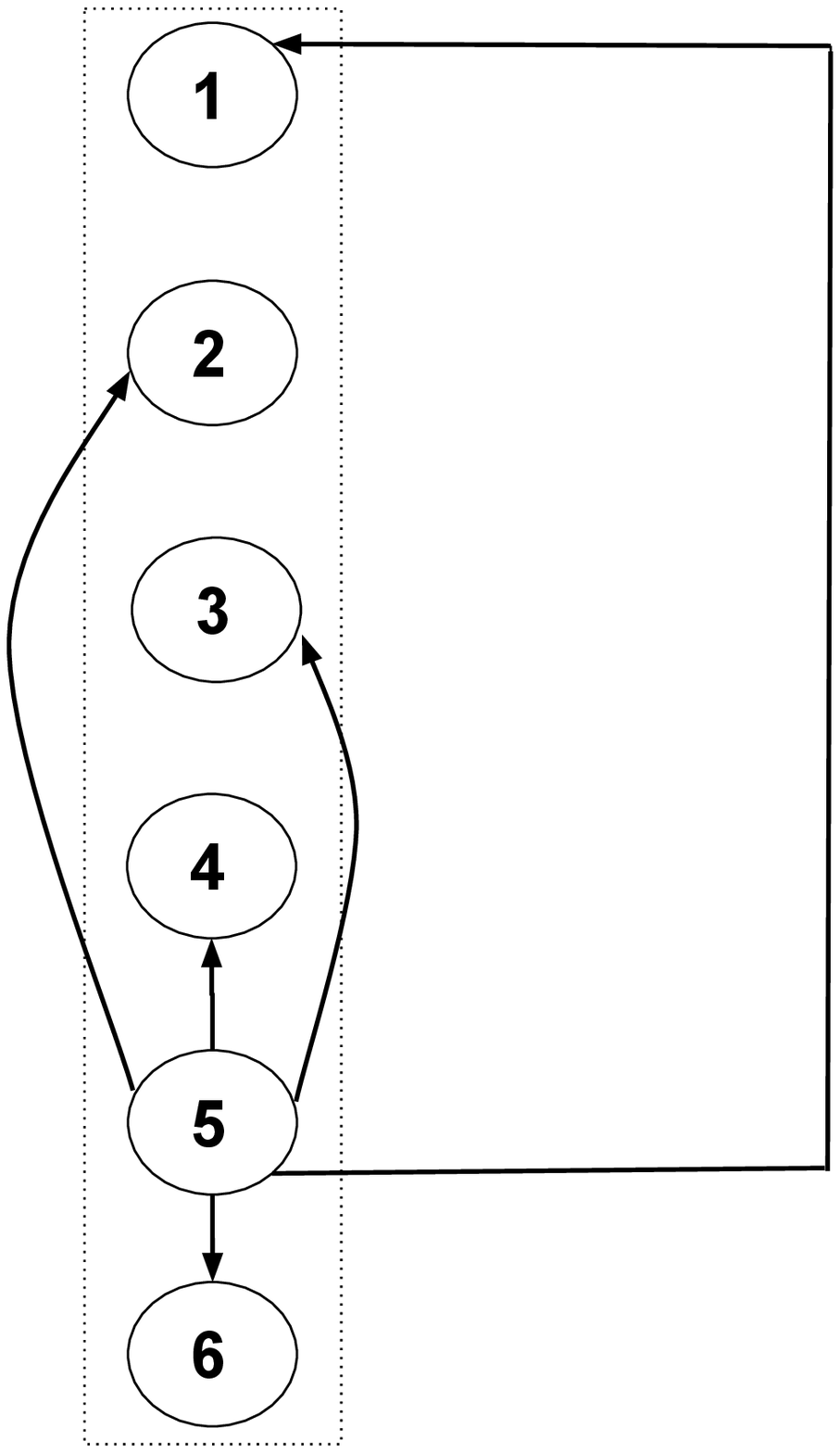}
  \caption{}
  \label{rt15}
\end{subfigure}
\begin{subfigure}{.31\textwidth}
	\hspace{-6mm}
  \includegraphics[width=15pc]{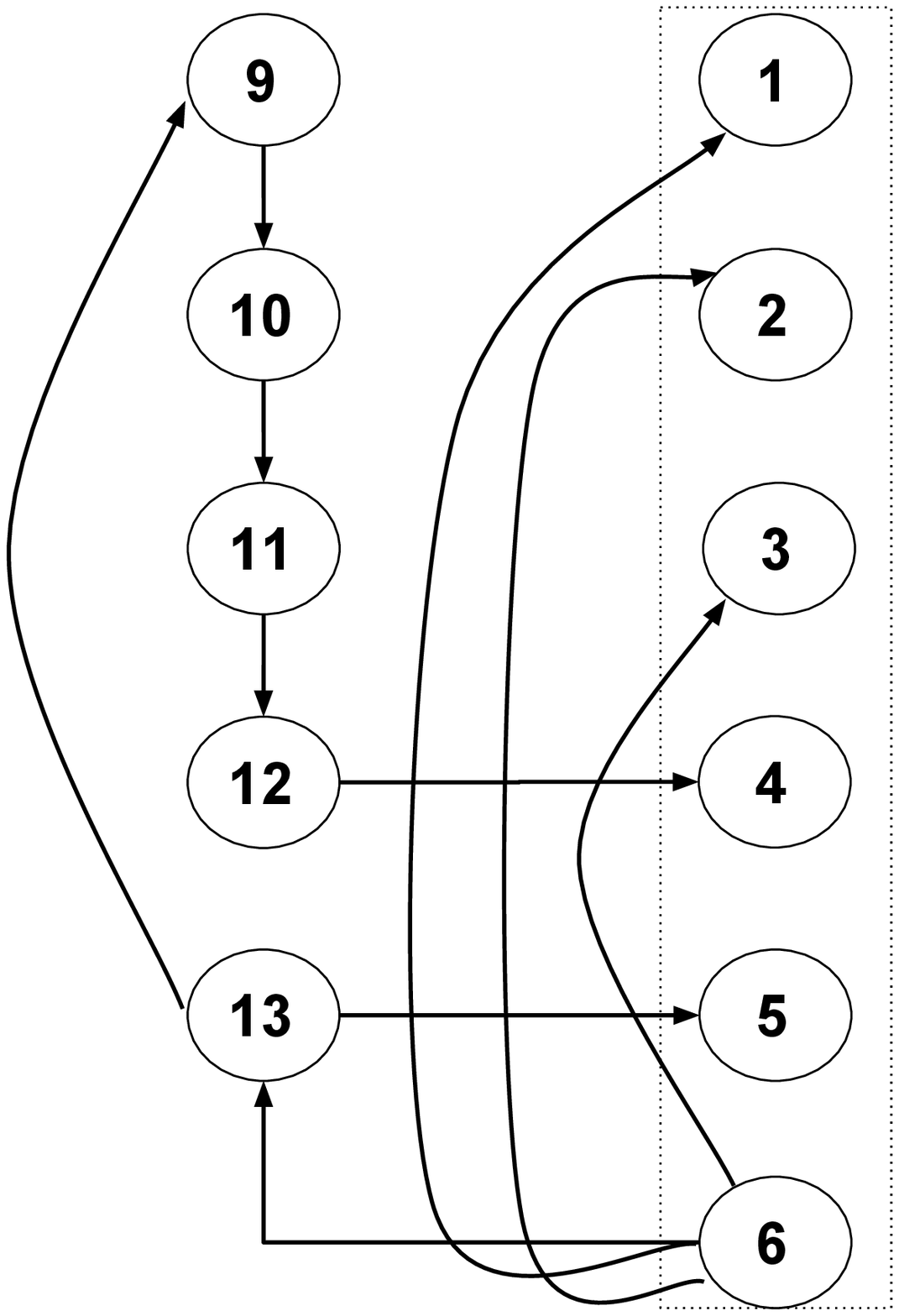}
  \caption{}
  \label{rt16}
\end{subfigure}
\caption{Figures showing rooted trees of inner vertices $ 1,2,3,4,5,6 $ of $ \mathcal{G}_1 $, respectively.}
\end{figure*}
\begin{ex}
\label{exam1}
Consider $\mathcal{G}_1$, a side-information graph which is a $ 6$- IC structure with inner vertex set $ V_I = \lbrace1,2,3,4,5,6\rbrace$ given in Fig. \ref{f1}.

It can be easily verified that
\begin{enumerate}
        \item there are no cycles containing only one vertex from the set $\lbrace1,2,3,4,5,6\rbrace$ in $\mathcal{G}_1$ (i.e., no I-cycles),
        \item using the rooted trees for each vertex in the set $\lbrace1,2,3,4,5,6\rbrace$, which are given in Fig. \ref{rt11}, \ref{rt12}, \ref{rt13}, \ref{rt14}, \ref{rt15} and \ref{rt16} respectively, there exists a unique path between any two different vertices in $ V_I $ in $ \mathcal{G}_1 $ and does not contain any other vertex in $ V_I $ (i.e., unique I-path between any pair of inner vertices),
        \item $\mathcal{G}_1$ is the union of all the $6$ rooted trees.
\end{enumerate}
Using \textit{Construction} $1$, the transmitted index code symbols are
 \begin{eqnarray*}
W_I &=&x_1\oplus x_2\oplus x_3\oplus x_4\oplus x_5\oplus x_6\\ 
W_7&=&x_7 \oplus x_2 \oplus x_6\\
W_8&=&x_8 \oplus x_3 \oplus x_5\\
W_9&=&x_9\oplus x_{10}\\
W_{10}&=&x_{10}\oplus x_{11}\\
W_{11}&=&x_{11}\oplus x_{12}\\
W_{12}&=&x_{12}\oplus x_{4} \oplus x_{13}\\
W_{13}&=&x_{13}\oplus x_{5} \oplus x_{9}\\
W_{14}&=&x_{14}\oplus x_{9}. 
 \end{eqnarray*}
Now consider the rooted tree of the inner vertex $2$  shown in  Fig. \ref{rt12}. 
Applying \textit{Algorithm} $1$ to decode $x_2$, we get
\begin{displaymath}
Z_2 =W_I \oplus W_{10} \oplus W_{11} \oplus W_{12} \oplus W_{13}  
\end{displaymath}
which results in 
\begin{displaymath}
Z_2 =x_1 \oplus x_2 \oplus x_3 \oplus x_6 \oplus x_{10} \oplus x_9
\end{displaymath}
using which message $x_2$ cannot be decoded by the user requesting it since $x_9$ is not available at that user as side-information. 
\end{ex}
\begin{ex}
\label{exam2}
Consider $\mathcal{G}_2$, a side-information graph which is a $5$-IC structure with inner vertex set $ V_I=\lbrace1,2,3,4,5\rbrace$, given in Fig. \ref{f2}.
\begin{figure}[!t]\centering
\includegraphics[height=\columnwidth, width=\columnwidth,angle=0]{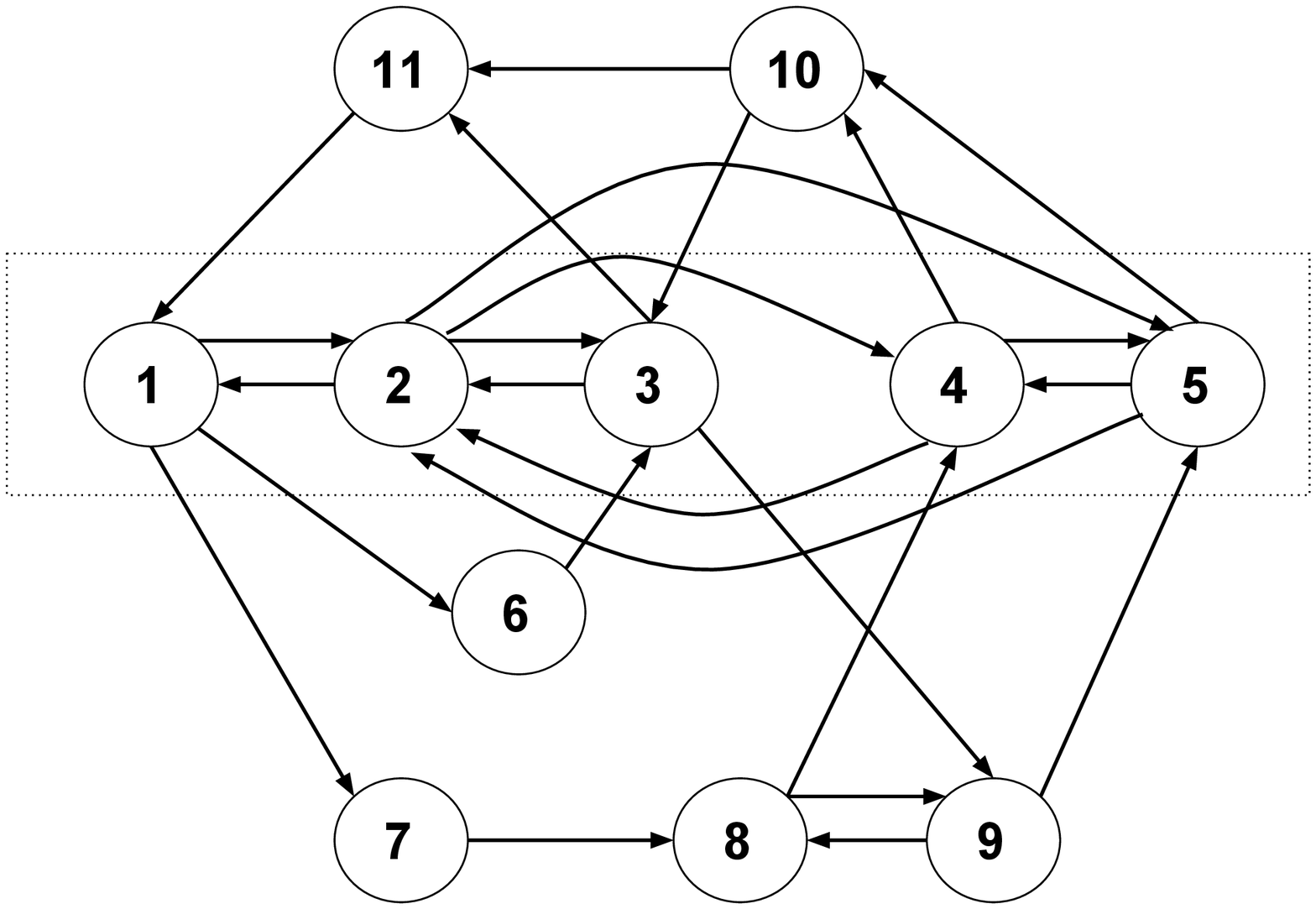}
\caption{$5$-IC structure, $\mathcal{G}_2$ with $V_I=\lbrace1,2,3,4,5\rbrace$}
\label{f2}
\end{figure}
It can be easily verified that
\begin{enumerate}
\item there are no cycles containing only one vertex from the set $\lbrace1,2,3,4,5\rbrace$ in $\mathcal{G}_2$ (i.e., no I-cycles),
\item using the rooted trees for each vertex in the set $\lbrace1,2,3,4,5\rbrace$, which are given in Fig. \ref{rt21}, \ref{rt22}, \ref{rt23}, \ref{rt24} and \ref{rt25}, respectively,there exists a unique path between any two different vertices in $ V_I $ in $ \mathcal{G}_2 $ and does not contain any other vertex in $ V_I $ (i.e, unique I-path between any pair of inner vertices),
\item $\mathcal{G}_2$ is the union of all the $5$ rooted trees.
\end{enumerate}
\begin{figure*}[!t]
\centering
\begin{subfigure}{.31\textwidth}
  \hspace{-8mm}
  \includegraphics[width=15pc]{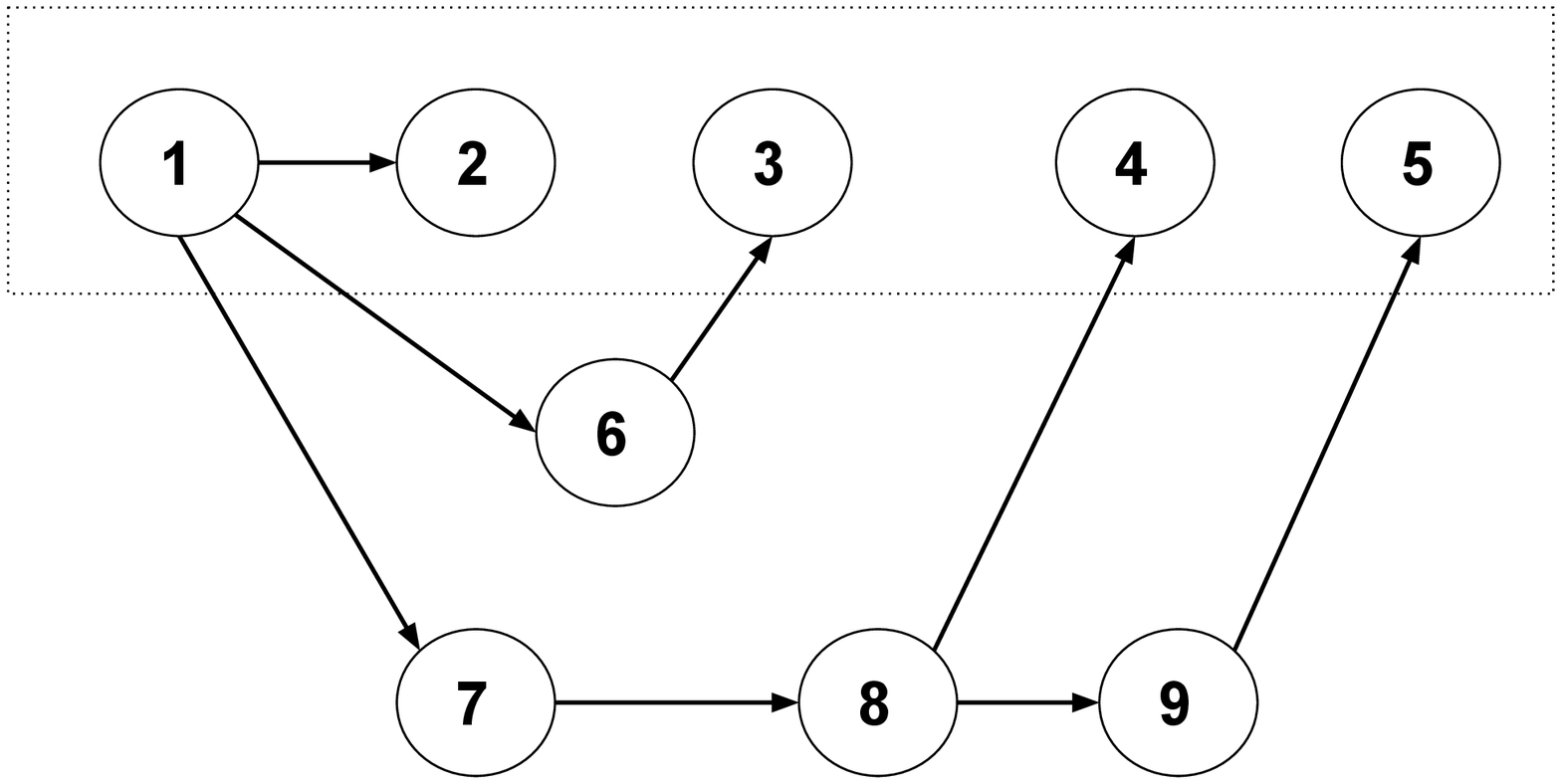}
  \caption{}
  \label{rt21}
\end{subfigure}%
\begin{subfigure}{.31\textwidth}
  \hspace{-8mm}
  \includegraphics[width=15pc]{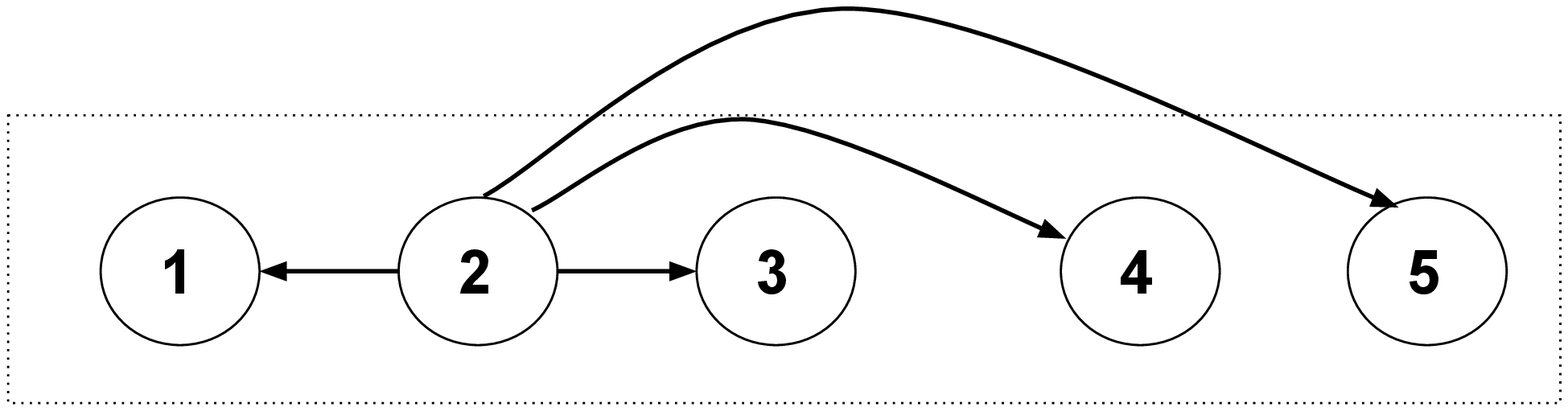}
  \caption{}
  \label{rt22}
\end{subfigure}
\begin{subfigure}{.31\textwidth}
  \hspace{-8mm}
  \includegraphics[width=15pc]{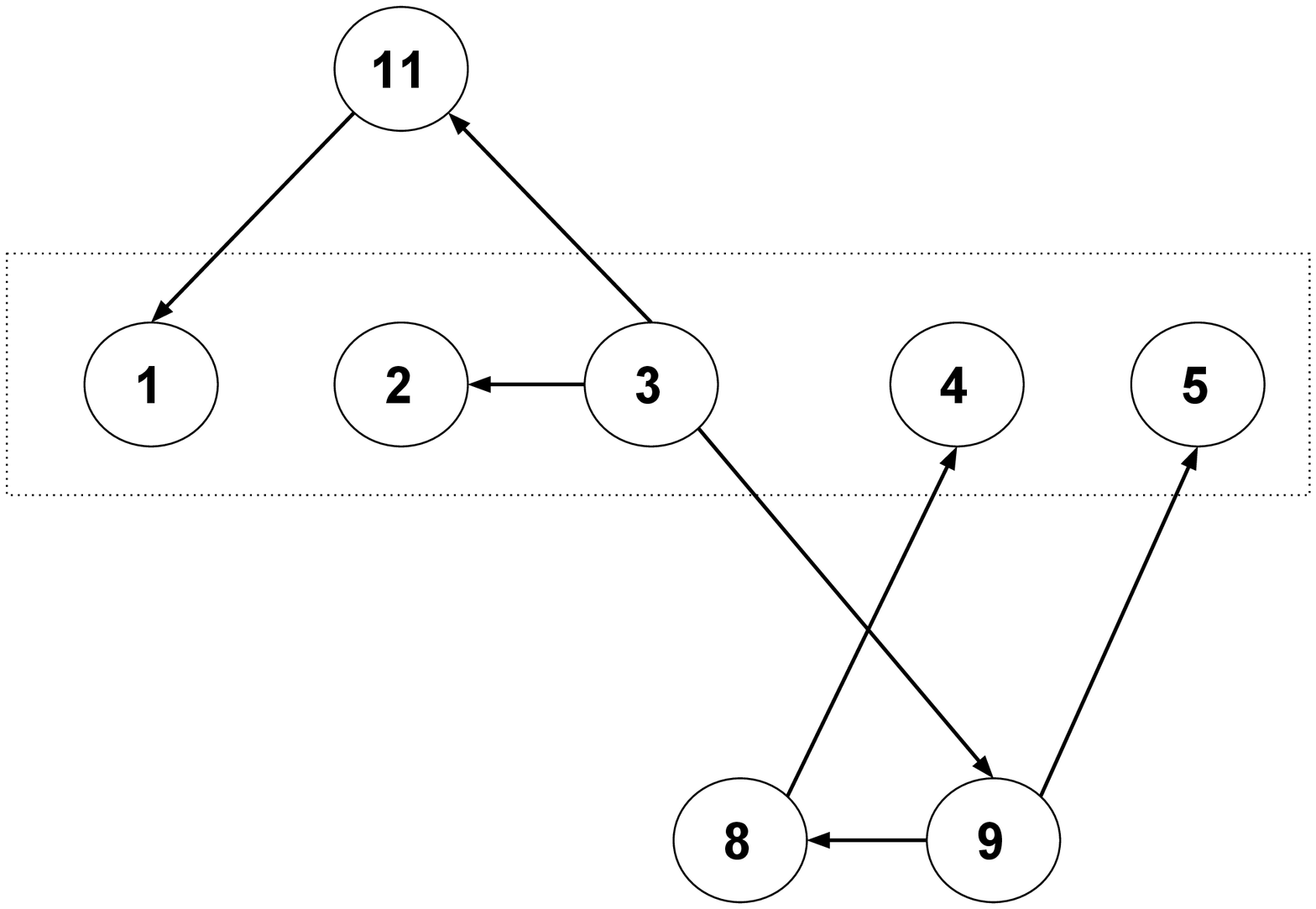}
  \caption{}
  \label{rt23}
\end{subfigure}
\centering
\begin{subfigure}{.31\textwidth}
  \hspace{-8mm}
  \includegraphics[width=15pc]{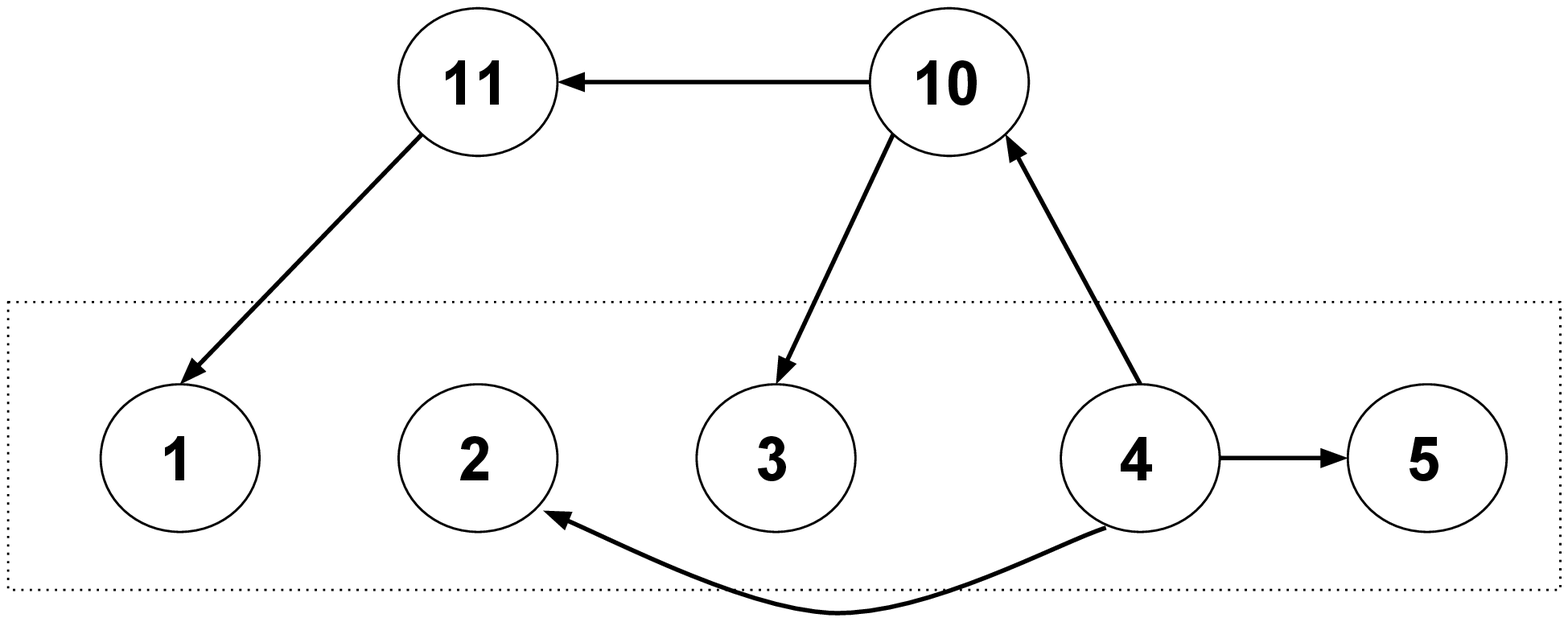}
  \caption{}
  \label{rt24}
\end{subfigure}%
\begin{subfigure}{.31\textwidth}
  \hspace{-8mm}
  \includegraphics[width=15pc]{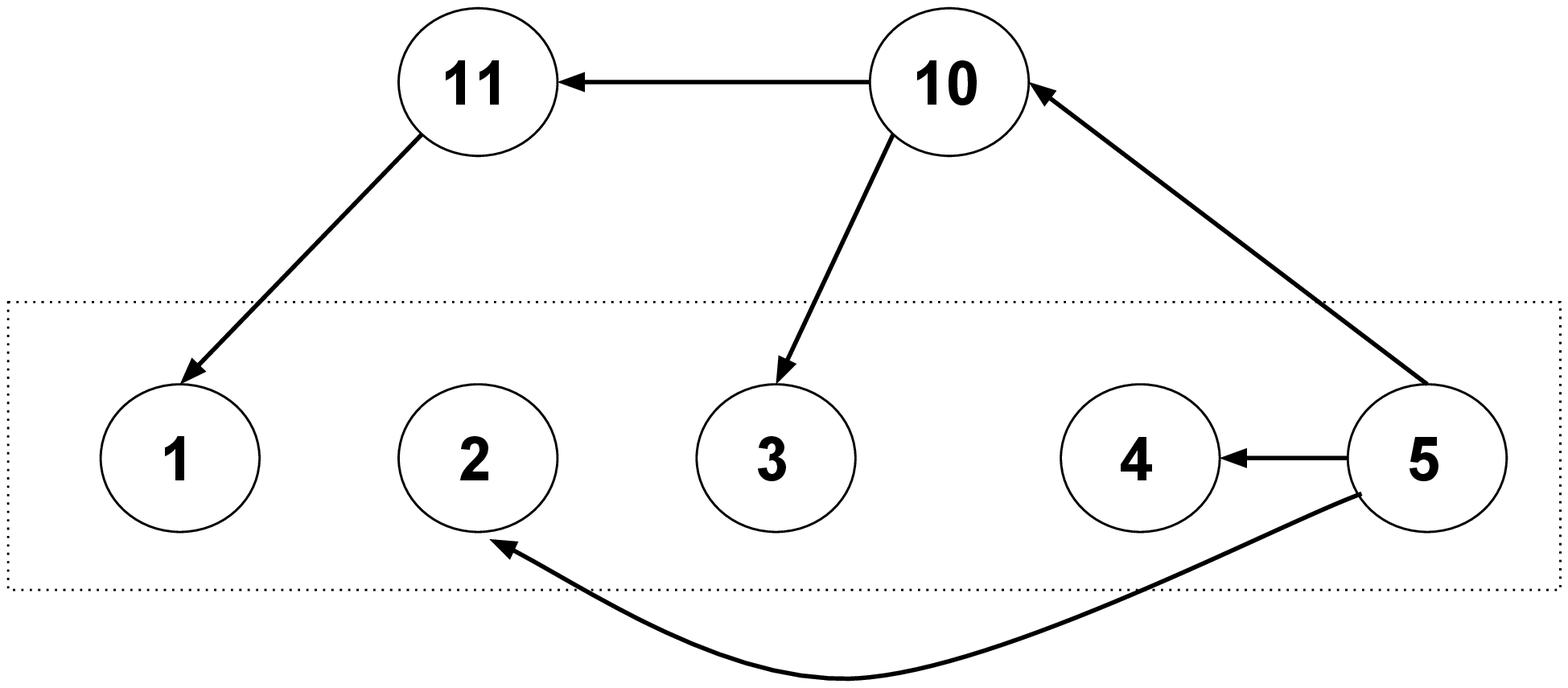}
  \caption{}
  \label{rt25}
\end{subfigure}
\caption{Figures showing rooted trees of inner vertices $ 1,2,3,4,5 $ of $ \mathcal{G}_2 $, respectively.}
\end{figure*}
Using \textit{Construction} $1$, the index code symbols are 
	\begin{eqnarray*}
		W_I &=& x_1 \oplus x_2 \oplus x_3 \oplus x_4 \oplus x_5 \\
		W_6 &=& x_6 \oplus x_3 \\
		W_7 &=& x_7 \oplus x_8 \\
		W_8 &=& x_8 \oplus x_4 \oplus x_9\\
		W_9 &=& x_9 \oplus x_5 \oplus x_8\\
		W_{10} &=& x_{10} \oplus x_3 \oplus x_{11}\\
		W_{11} &=& x_{11} \oplus x_1.
	\end{eqnarray*}
Now, consider the rooted tree for the inner vertex $1$, shown in Fig. \ref{rt21}. By applying \textit{Algorithm} $1$ to decode $x_1$, we get
\begin{displaymath}
	Z_1=W_I \oplus W_6 \oplus W_7 \oplus W_8\oplus W_9
\end{displaymath}
which results in 
\begin{displaymath}
	Z_1=x_1 \oplus x_2 \oplus x_6 \oplus x_7 \oplus x_8,
\end{displaymath} using which, message $x_1$ is not decodable by the user requesting it because $ x_8 $ is not available at that user as side-information.
\end{ex}
\section{Main Results}

The two examples  in the previous section motivate \textit{Theorem} $1$ which imposes a set of necessary and sufficient conditions on a given IC structure for an index code obtained by using \textit{Construction} $1$ for that IC structure to be decodable using \textit{Algorithm} $1$.

Recall that for the $K$-IC structure, $\mathcal{G}$, having inner vertex set $V_I=\lbrace 1,2,\dots,K\rbrace$ and non-inner vertices $ V_{NI}=\lbrace K+1, K+2,\dots,N \rbrace $   $T_i$ is the rooted tree corresponding to the inner vertex $i$ where $i \in \lbrace 1,2,\dots,K\rbrace$  and $V_{NI}(i)$ is the set of non-inner vertices in $\mathcal{G}$ which appear in the rooted tree $T_i$ of an inner vertex $ i $. 

For each $ i \in \lbrace1,2,\dots,K\rbrace $ and for a non-inner vertex $ j $ which is at a depth $ \geq2 $ in the rooted tree $ T_i $, define $ a_{i,j} $ as the number of vertices in $V_{NI}(i)$ for which $ j $ is in out-neighbourhood in $ \mathcal{G} $, i.e., for each $ i \in \lbrace1,2,\dots,K\rbrace $ and for $j \in {V_{NI}(i) \backslash N^+_{T_i}(i)} $,
\begin{displaymath}
a_{i,j} \triangleq |\lbrace v:v\in V_{NI}(i), j\in N^{+}_{\mathcal{G}}(v)\rbrace|.
\end{displaymath}
Also, for each $ i \in \lbrace1,2,\dots,K\rbrace $ and for a non-inner vertex $ j $ not in the rooted tree $ T_i $, define $ b_{i,j} $ as the number of vertices in $ V_{NI}(i) $ for which $ j $ is in out-neighbourhood in $ \mathcal{G}$, i.e., for each $ i \in \lbrace1,2,\dots,K\rbrace $ and $ j \in V(\mathcal{G}) \backslash V(T_i) $,
\begin{displaymath}
b_{i,j} \triangleq |\lbrace v:v\in V_{NI}(i),~j\in N^{+}_{\mathcal{G}}(v)\rbrace|.
\end{displaymath}
First, the following \textit{Lemma} is proved.
\begin{lm} 
\label{lem1}
Given an IC structure $ \mathcal{G} $ with inner vertex set $ V_I=\lbrace1,2,\dots,K \rbrace $, $ b_{i,j} \in \lbrace 0,1\rbrace $ for each $ i \in \lbrace1,2,\dots,K\rbrace $ and $ j \in V(\mathcal{G}) \backslash V(T_i) $.
\end{lm}
\begin{proof}
	Suppose, for an $ i \in \lbrace 1,2,\dots,K \rbrace $ and a $ j \in V(\mathcal{G})\backslash V(T_i)$, let $ b_{i,j}=a $, for some integer $ a \geq2 $. Let $ p $ and $ q $  be any two different vertices in the set $ \lbrace v:v\in V_{NI}(i), j\in N^{+}_{\mathcal{G}}(v)\rbrace $. Then $ j \in N^+_{\mathcal{G}}(p) $ and $ j \in N^+_{\mathcal{G}}(q) $. In $ T_i $, $ p $ and $ q $ can be predecessors of a single inner vertex or two different inner vertices. 
\begin{case}[i]
Consider the case where $ p $ and $ q $ are predecessors of a single inner vertex, $ n \in V_I \backslash \lbrace i \rbrace $, i.e., non-inner vertices $ p $ and $ q $ are on the I-path from the inner vertex $ i $ to the inner vertex $ n $. Also let $ q $ be reached from $ p $, i.e., in the I-path from $ i $ to $ n $, the vertex $ p $ is reached first and $ q $ is reached from $ p $. Since every non-inner vertex has to be a predecessor of at least one inner vertex (by definition of IC structure), the non-inner vertex $ j $ must also be a predecessor of an inner vertex. Let $ j $ be a predecessor of an inner vertex $ m \in V_I\backslash \lbrace i,n \rbrace $. Since the arc from $ p $ to $ j $ does not exist in $ T_i $ but exists in $ \mathcal{G} $, there exists an I-path from an inner vertex $ s\in V_I\backslash\lbrace i,n \rbrace $ to the inner vertex $ m $, passing through the non-inner vertex $ p $ and then through $ j $. But now, there exist two I-paths in $ \mathcal{G} $ from $ s $ to $ m $, one that has a direct arc between $ p $ and $ j $,
	\begin{displaymath}
	s \rightarrow \dots \rightarrow  p \rightarrow j \dots \rightarrow m
	\end{displaymath} and one in which $ j $ is reached from $ p $ through $ q $.
	\begin{displaymath}
	s \rightarrow \dots \rightarrow p \dots \rightarrow q\rightarrow j \dots \rightarrow m,
	\end{displaymath}
which is not allowed in an IC structure and hence $ p $ and $ q $ being predecessors of a single inner vertex is not possible.
\end{case}	
\begin{case}[ii]

	The other case is where $ p $ and $ q $ are predecessors of different inner vertices. Let $ p $ be predecessor of an inner vertex $ n \in V_I \backslash \lbrace i \rbrace $ and $ q $ be predecessor of an inner vertex $ m \in V_I \backslash \lbrace i,n\rbrace $. Let the set of inner vertices that are reached from the non-inner vertex $ j $ through some path in $ \mathcal{G} $ be $ V_I(j) $. Define $ S_j $ as the set of vertices that are successors of $ j $ and predecessors of vertices in $ V_I(j) $. We have $ V_I(j) \subset V(T_i) $. The paths from $ i $ to some vertices in $ V_I(j) $ pass through some vertices in $ S_j $ and paths from $ i $ to remaining vertices in $ V_I(j) $ will not pass through any vertex in $ S_j $. The subset of vertices in $ S_j $ which are also successors of the inner vertex $ i $ is denoted by $ S_{i,j} $.
	\begin{figure*}[!t]
	\centering
		\includegraphics[height=\columnwidth,width=\columnwidth,angle=0]{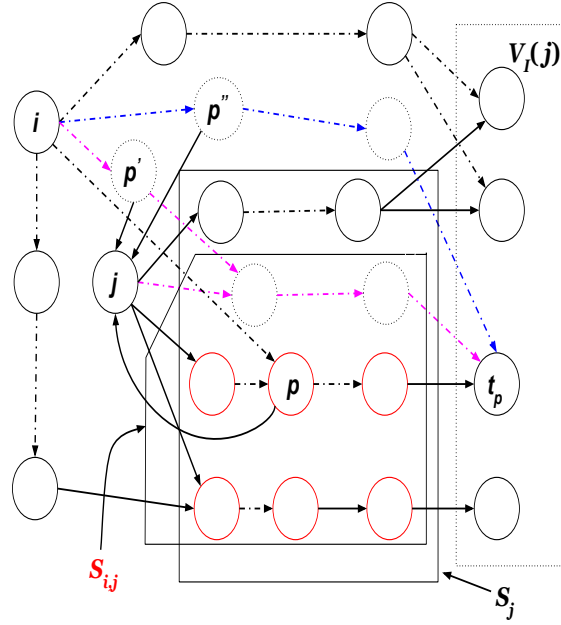}
		\caption{Illustration of $ S_j $ and $ S_{i,j}$.}
		\label{pf2}
	\end{figure*}Fig. \ref{pf2} illustrates $ S_j $ and $ S_{i,j} $. Dotted arrows indicate presence of some vertices in the path. Now 
	\begin{clm}
		$ p,q\in S_{i,j} $.
	\end{clm} 
	\begin{proof}
		Since there is an edge from $ p $ to $ j $ in $ \mathcal{G} $, $ V_I(j) \subseteq V_I(p) $. Let $ t_p $ be a vertex in $ V_I(j) $. Then $ p $ will be a predecessor of $ t_p $. 
		 Fig. \ref{pf2} also illustrates all the possible positions of $ p $. In the cases when $ p \not\in S_{i,j} $ the corresponding vertices are shown to be dotted and $ p $ is marked as $ p' $ and $ p'' $. It is seen that there exist two I-paths from $ i $ and $ t_p $ in the two cases when $ p \not \in S_{i,j} $ which is not allowed in an IC structure. So $ p \in S_{i,j} $ which means that there exists a path from $ j $ to $ p $. Similarly, there exists an inner vertex $ t_q \in V_I(j) $ to which $ q $ is a predecessor. This implies that $ q \in S_{i,j} $. Hence there exists a path from $ j $ to $ q $. 
	\end{proof}Since both $ p$, $ q \in S_{i,j}$, i.e., $ j $ is a predecessor of both $ p $ and $ q $, there exists a path from $ p $ to $ q $ through $ j $ in $ \mathcal{G} $ which implies that there exists an I- path from $ i $ to $ m $ through $ p $ and $ q $, which is a contradiction to the assumption that $ p $ and $ q $ don't lead to the same inner vertex.
\end{case}Thus $ b_{i,j} $ cannot take values more than 1. Hence $ b_{i,j} \in \lbrace0,1\rbrace $.
\end{proof}

After Theorem \ref{thm2} several examples are discussed for some of which $b_{i,j}=0$ and for the remaining ones $b_{i,j}=1.$

\begin{thm}
\label{thm1}
	The index code obtained from \textit{Construction} $1$ on $ \mathcal{G} $ is decodable using \textit{Algorithm} $1$ if and only if the IC structure, $ \mathcal{G} $, satisfies the following two conditions \textit{c}$1$ and \textit{c}$2$.
	\begin{cond}[\textit{c}$1$] $ a_{i,j} $ must be an odd number for each $ i \in \lbrace1,2,\dots,K\rbrace $ and $ j \in V_{NI}(i)\backslash N^+_{T_i}(i)$.  
	\end{cond}
	\begin{cond}[\textit{c}$2$] $ b_{i,j} $ must be zero for each $ i \in \lbrace1,2,\dots,K\rbrace $ and $ j \in V(\mathcal{G}) \backslash  V(T_i)$.
	\end{cond}
\end{thm}
\begin{proof}
	The proof of the \textit{if} part is as follows. Let the inner vertex of interest be $i$ and its rooted tree be $T_i$. Using \textit{Algorithm} $1$ to decode $x_i$, $Z_i$ is computed as 
	\begin{displaymath}
	Z_i=W_I \underset{j\in V_{NI}(i)}{\oplus}W_j.
	\end{displaymath}
In $Z_i$, the messages corresponding to the inner vertices that are not directly connected to $i$ will be cancelled since each such message appears exactly twice, once in $W_I$ and once in the index code symbol corresponding to the non-inner vertex which is the immediate predecessor of the inner vertex in $T_i$. Message $x_j$ that corresponds to a non-inner vertex $j$ at a depth $\geq 2$ in $T_i$ appears exactly even number of times in $Z_i$, once in $W_j$ and in odd number of index code symbols corresponding to the non-inner vertices that are in $T_i$ as it is in out-neighbourhood of odd number of vertices in $ V_{NI}(i) $, in $ \mathcal{G} $, by hypothesis and hence it is also cancelled. Finally, message corresponding to a non-inner vertex that is not in $T_i$ is not present in $Z_i$ as it appears in none of the index code symbols corresponding to the vertices $V_{NI}(i)$, since, by hypothesis, a non-inner vertex not in $T_i$, is in out-neighbourhood of no vertices that are in $ V_{NI}(i) $, in $ \mathcal{G} $. So, $Z_i$ will be of the form 
	\begin{displaymath}
	Z_i=x_i \underset{j \in S \subseteq N^{+}_{T_{i}}(i)}{\oplus}x_j
	\end{displaymath} and hence $x_i$ is decodable by the user requesting $ x_i $ as the user will have messages corresponding to vertices in $ N^+_{T_i}(i) $ as side-information.
	
	The proof of the \textit{only if} part follows. Let the condition \textit{c}$ 1 $ be violated and \textit{c}$2$ be true, and let, in the rooted tree $T_i$ of an inner vertex $i$, a non-inner vertex $j$ at depth $\geq 2$ is in out-neighbourhood of even number of vertices in $ V_{NI}(i) $, in $ \mathcal{G} $ i.e., $ a_{i,j} $ is even. It is evident that $x_j$ is not cancelled in $Z_i$ because $x_j$ appears once in $W_j$ and in even number of index code symbols corresponding to vertices in $ V_{NI}(i)\backslash \lbrace j \rbrace$. As a result, $x_i$ is not decodable by user requesting it since $ x_j $ is not available to the user requesting $ x_i $ as side-information.\\
	Now, let \textit{c}$1$ be true and \textit{c}$2$ be false. Let the non-inner vertex $j$ violate \textit{c}$2$ for a rooted tree $T_i$ corresponding to an inner vertex $i$, i.e., $ b_{i,j} = 1 $. Then $x_j$ is not cancelled in $Z_i$ because it appears only once in the index code symbols corresponding to the vertices that are in $ V_{NI}(i)$ and thus inhibiting decodability of $x_i$ for the user requesting it, since the user does not have $ x_j $ as side-information. 
\end{proof}
\begin{thm}
\label{thm2}
	An IC structure which has no cycles containing only non-inner vertices satisfies both the conditions \textit{c}$ 1 $ and \textit{c}$ 2 $.
\end{thm}
\begin{proof} The proof is in two parts. The condition \textit{c}$ 1 $ is shown to be satisfied in Part $ 1 $ and the condition \textit{c}$ 2 $ in Part $ 2 $.
	\begin{part}
	Let $ \mathcal{G} $ be an IC structure which has no cycles containing only non-inner vertices. It will be shown that $ a_{i,j} \geq 2 $ is not possible for any inner vertex $ i $ and a non-inner vertex $j$ at depth $ \geq2 $ in $ T_i $. Consider for such $ i $ and $ j $, $ a_{i,j} = z $ for some integer $ z\geq2 $. Since $ j $ is at depth $ \geq2 $, there exists a non-inner vertex $ p $ which is a successor of $ i $, predecessor of $ j $ and has $ j $ in its out-neighbourhood. So, $ a_{i,j} \geq1$. Let $ q $ be a non-inner vertex in the set $ \lbrace v:v\in V_{NI}(i), j\in N^{+}_{\mathcal{G}}(v)\rbrace $ and $ q\not = p $. Let $ V_I(j) $ be the set of inner vertices reached from $ j $ in $ T_i $. Let $ S_j $ denotes the set of non-inner vertices that are successors of $ j $ and predecessors of the vertices in $ V_I(j) $.
\begin{clm}
		$ q \in S_j $
	\end{clm}
	\begin{proof}
		Suppose not. Then there exists two I-paths from $ i $ to vertices in $ V_I(j) $ (one passing through $ p $ and the other passing through $ q $, see Fig. \ref{imag}). 
		\begin{figure}[!t]
		\centering
			\includegraphics[height=2in,width=3in,angle=0]{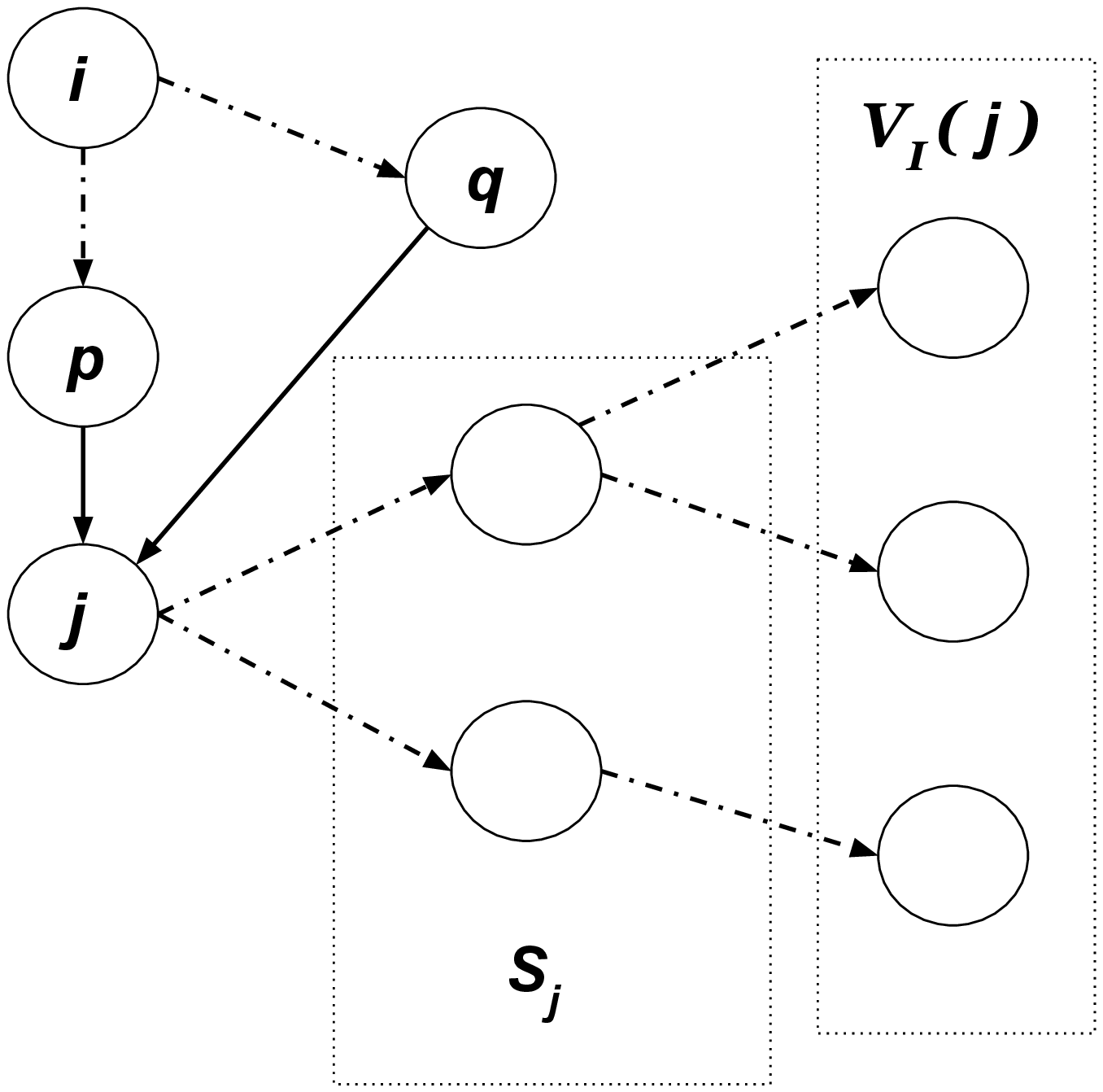}
			\caption{Figure illustrating $ q \not \in S_j $.}
			\label{imag}
		\end{figure}
		Hence it is a contradiction.
\end{proof}
	Now, $ q\in S_j $ means that there is a path from $ j $ to $ q $. As there exists an edge from $ q $ to $ j $, by definition of $ q $, a cycle containing only non-inner vertices which include $ j $ and $ q $ is formed. This is a contradiction to the assumption that $ \mathcal{G} $ doesn't have any cycles containing only non-inner vertices. Hence $ a_{i,j} \leq1 $. As a result, $ a_{i,j} = 1 $ (an odd number) for any inner vertex $ i $ and a non-inner vertex $ j $ which is at a depth $ \geq2 $ in $ T_i $. 
	\end{part}
\begin{part}
	Let $ \mathcal{G} $ be an IC structure which has no cycles containing only non-inner vertices. Suppose $ b_{i,j}=1 $ for an inner vertex $ i $ and a non-inner vertex $ j $ which is not in the rooted tree $ T_i $. This implies that there exists a non-inner vertex $ p $ in $ T_i $ which has $ j $ in its out-neighbourhood. Let $ V_I(j)$ be the set of inner vertices that are successors of non-inner vertex $ j $. Let $ S_j $ be the set of non-inner vertices that are successors of $ j $ and predecessors of vertices in $ V_I(j) $ and let $ S_{i,j} $ be the subset of non-inner vertices in $ S_j $ which are successors of $ i $ (see Fig. \ref{pf2}).
\begin{clm}
		$ p \in S_{i,j} $.
	\end{clm}
	\begin{proof}
		Suppose not. Then it leads to existence of two I-paths between $ i $ and $ t_p $ which is not allowed in an IC structure. See Fig. \ref{pf2}.
\end{proof}
	Since $p \in S_{i,j}$, there exists a path from $ j $ to $ p $. Since an edge exists from $ p $ to $ j $, by definition of $ p $, a cycle containing only non-inner vertices which include $ p $ and $ j $ exists in $ \mathcal{G} $. This is a contradiction.
\end{part}Hence $ b_{i,j} \not =1 $ for any inner vertex $ i $ and any non-inner vertex $ j $ not in the rooted tree $ T_i $. As $ b_{i,j} \in \lbrace0,1\rbrace $ (as proved in \textit{Lemma} $ 1 $), $ b_{i,j} = 0 $ for any inner vertex $ i $ and any non-inner vertex $ j $ which is not present in rooted tree $ T_i $. 
\end{proof}

The conditions of Theorem \ref{thm1} are now illustrated for Example  \ref{exam1} and Example \ref{exam2} discussed in the previous section. Also, it is shown that though the constructed codes for these two examples are not decodable using \textit{Algorithm} 1 they are decodable using only linear combinations of the index code symbols.   
\begin{exmp}[continued]
Table \ref{table1} shows that $ \mathcal{G}_1 $  violates \textit{c}$ 1 $ and Table \ref{table2} shows that \textit{c}$ 2 $ is also violated by $ \mathcal{G}_1 $. Since the conditions are violated in the rooted trees of inner vertices $ 1 $, $ 2 $ and $ 3 $, messages $ x_1 $, $ x_2 $ and $ x_3 $ are not decodable by using \textit{Algorithm} $ 1 $ whereas the messages $ x_4 $, $ x_5 $ and $ x_6 $ are decodable using \textit{Algorithm} $ 1 $.

\begin{remark}
However, the messages $ x_1 $, $ x_2 $ and $ x_3 $ are decodable by using other linear combinations of the index code symbols, as shown.
\begin{itemize}
\item $ x_1 $ is decoded using $ Z'_1= W_I \oplus W_7 \oplus W_9 \oplus W_{10} \oplus W_{11} \oplus W_{12} 	\oplus W_{13} $ which results in $ Z'_1= x_1 \oplus x_3 \oplus x_7 $.
\item $ x_2 $ and $ x_3 $ are decoded using	$ Z'_2 = W_I \oplus W_9 \oplus W_{10} \oplus W_{11} \oplus W_{12} \oplus W_{13} $ which results in $ Z'_2= x_1 \oplus x_2 \oplus x_3 \oplus x_6 $.
\end{itemize}
\end{remark}
\begin{table}[!t]
\centering
\begin{tabular}{|c|c|c|c|}
\hline
$ T_i $ & $ V_{NI}(i) $ & $ j \in V_{NI}(T_i)\backslash N^+_{T_i}(i) $ & $ a_{i,j} $ \\
\hline\hline
$ T_1 $ & $ \lbrace7,9,10,11,12,13,14\rbrace $ & $ \lbrace9,10,11,12,13\rbrace $ & $ 2 $,$ 1 $,$ 1 $,$ 1 $,$ 1 $ \\ 	
\hline
$ T_2 $ & $ \lbrace10,11,12,13\rbrace $ & $ \lbrace11,12,13\rbrace $ & $ 1 $,$ 1 $,$ 1 $ \\
\hline
$ T_3 $ & $ \lbrace11,12,13\rbrace $ & $ \lbrace12,13\rbrace $ & $ 1 $,$ 1 $ \\ 	 
\hline
$ T_4 $ & $ \lbrace7,8\rbrace $ & $ \phi $ & $ - $ \\ 		
\hline
$ T_5 $ & $ \phi $ & $ \phi $ & $ - $ \\
\hline
$ T_6 $ & $9,10,11,12,13$ & $ 9,10,11,12$ & $ 1 $,$ 1 $,$ 1 $,$ 1 $ \\
\hline
\end{tabular}\\
\caption{Table that verifies \textit{c}$ 1 $ for $ \mathcal{G}_1 $}
\label{table1}
\end{table}
\begin{table*}[!t]
\centering
	\begin{tabular}{|c|c|c|c|}
		\hline
		$ T_i $ & $ V_{NI}(i) $ & $ j \in V(\mathcal{G}_1) \backslash V(T_i) $ & $ b_{i,j} $ \\ 
		\hline\hline
		$ T_1 $ & $ \lbrace7,9,10,11,12,13,14\rbrace $ & $ \lbrace8\rbrace $ & $ 0 $ \\ 	
		\hline
		$ T_2 $ & $ \lbrace10,11,12,13\rbrace $ & $ \lbrace7,8,9,14\rbrace $ & $ 0 $,$ 0 $,$ 1 $,$ 0 $ \\
		\hline
		$ T_3 $ & $ \lbrace11,12,13\rbrace $ & $ \lbrace7,8,9,10,14\rbrace $ & $ 0 $,$ 0 $,$ 1 $,$ 0 $ \\ 	 
		\hline
		$ T_4 $ & $ \lbrace7,8\rbrace $ & $ \lbrace9,10,11,12,13,14\rbrace $ & $ 0 $,$ 0 $,$ 0 $,$ 0 $,$ 0 $,$ 0 $\\ 		
		\hline
		$ T_5 $ & $ \phi $ & $ \lbrace7,8,9,10,11,12,13,14\rbrace $ & $ - $ \\
		\hline
		$ T_6 $ & $9,10,11,12,13$ & $\lbrace7,8,14\rbrace$ & $ 0 $,$ 0 $,$ 0 $\\
		\hline
	\end{tabular}\\
	\caption{Table that verifies \textit{c}$ 2 $ for $ \mathcal{G}_1 $}
	\label{table2}
\end{table*}
\end{exmp}
\begin{exmp}[continued]
	Table \ref{table3} shows that $ \mathcal{G}_2 $  violates \textit{c}$ 1 $ and Table \ref{table4} shows that \textit{c}$ 2 $ is satisfied by $ \mathcal{G}_2 $. Since \textit{c}$ 1 $ is violated in the rooted tree of inner vertex $ 1 $, the message $ x_1 $ is not decodable by using \textit{Algorithm} $ 1 $ and the rest of the messages corresponding to inner vertices are decodable by using \textit{Algorithm} $ 1 $.
	\begin{remark}
However, $ x_1 $ is decodable using the the linear combination $ Z''= W_I \oplus W_6 \oplus W_8 \oplus W_9 $ which results in $ Z''= x_1 \oplus x_2 \oplus x_6 $.
	\end{remark}
\begin{table}[!t]\centering
\begin{tabular}{|c|c|c|c|} \hline
			$ T_i $ & $ V_{NI}(i) $ & $ j \in V_{NI}(T_i)\backslash  N^+_{T_i}(i)  $ & $ a_{i,j} $ \\
			\hline\hline
			$ T_1 $ & $ \lbrace6,7,8,9\rbrace $ & $ \lbrace8,9\rbrace $ & $ 2 $,$ 1 $ \\ 	
			\hline
			$ T_2 $ & $ \phi $ & $ \phi $ & $ - $\\
			\hline
			$ T_3 $ & $ \lbrace8,9,11\rbrace $ & $ \lbrace8\rbrace $ & $ 1 $\\ 	 
			\hline
			$ T_4 $ & $ \lbrace10,11\rbrace $ & $ \lbrace11\rbrace $ & $ 1 $ \\ 		
			\hline
			$ T_5 $ & $ \lbrace10,11\rbrace $ & $ \lbrace11\rbrace $ & $ 1 $ \\
			\hline
		\end{tabular}\\
		\caption{Table that verifies \textit{c}$ 1 $ for $ \mathcal{G}_2 $}
		\label{table3}
	\end{table}
	\begin{table}[!t]
		\centering
		\begin{tabular}{|c|c|c|c|}
			\hline
			$ T_i $ & $ V_{NI}(i) $ & $ j \in V(\mathcal{G}_2) \backslash V(T_i) $ & $ b_{i,j} $ \\ 
			\hline\hline
			$ T_1 $ & $ \lbrace7,8,9\rbrace $ & $ \lbrace6,10,11\rbrace $ & $ 0 $,$ 0 $,$ 0 $ \\ 	
			\hline
			$ T_2 $ & $ \phi $ & $ \lbrace6,7,8,9,10,11\rbrace $ & $ - $\\
			\hline
			$ T_3 $ & $ \lbrace8,9,11\rbrace $ & $ \lbrace6,7,10\rbrace $ & $ 0 $,$ 0 $,$ 0 $ \\ 	 
			\hline
			$ T_4 $ & $ \lbrace10,11\rbrace $ & $ \lbrace6,7,8,9\rbrace $ & $ 0 $,$ 0 $,$ 0 $,$ 0 $\\ 		
			\hline
			$ T_5 $ & $ \lbrace10,11\rbrace $ & $ \lbrace6,7,8,9\rbrace $ & $ 0 $,$ 0 $,$ 0 $,$ 0 $ \\
			\hline
		\end{tabular}\\
		\caption{Table that verifies \textit{c}$ 2 $ for $ \mathcal{G}_2 $}
		\label{table4}
	\end{table}
\end{exmp}
Example \ref{exam3} and Example \ref{exam4} discussed below illustrate Theorem \ref{thm2}. 

\begin{ex}
\label{exam3}
Consider $\mathcal{G}_3$, a side-information graph which is a $3$-IC structure shown in Fig. \ref{f4}. Notice that $\mathcal{G}_3$ does not have any cycles consisting of only non-inner vertices and that it is 
\begin{figure}[!t]
	\centering
			\includegraphics[height=\columnwidth,width=\columnwidth,angle=0]{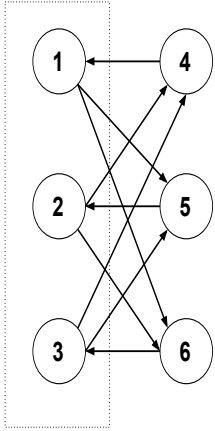}
		\caption{$ 3 $-IC structure $\mathcal{G}_3$ with $ V_I=\lbrace1,2,3\rbrace$.}
		\label{f4}
\end{figure}
indeed  a 3-IC structure with inner vertex set $V_I=\lbrace1,2,3\rbrace$ since 
	\begin{enumerate}
		\item there are no cycles containing only one vertex from the set $\lbrace1,2,3\rbrace $ in $\mathcal{G}_3$ (i.e., no I-cycles),
		\item using the rooted trees for each vertex in the set $\lbrace1,2,3\rbrace$, which are given in Fig. \ref{rt41}, \ref{rt42} and \ref{rt43} respectively, it is verified that there exists a unique path between any two different vertices in $ V_I $ in $ \mathcal{G}_3 $ and does not contain any other vertex in $ V_I $ (i.e., unique I-path between any pair of inner vertices),
		\item $\mathcal{G}_3$ is the union of all the $3$ rooted trees.
	\end{enumerate}
\begin{figure*}[!t]
	\centering
	\begin{subfigure}{.31\textwidth}
		\centering
		\includegraphics[width=15pc]{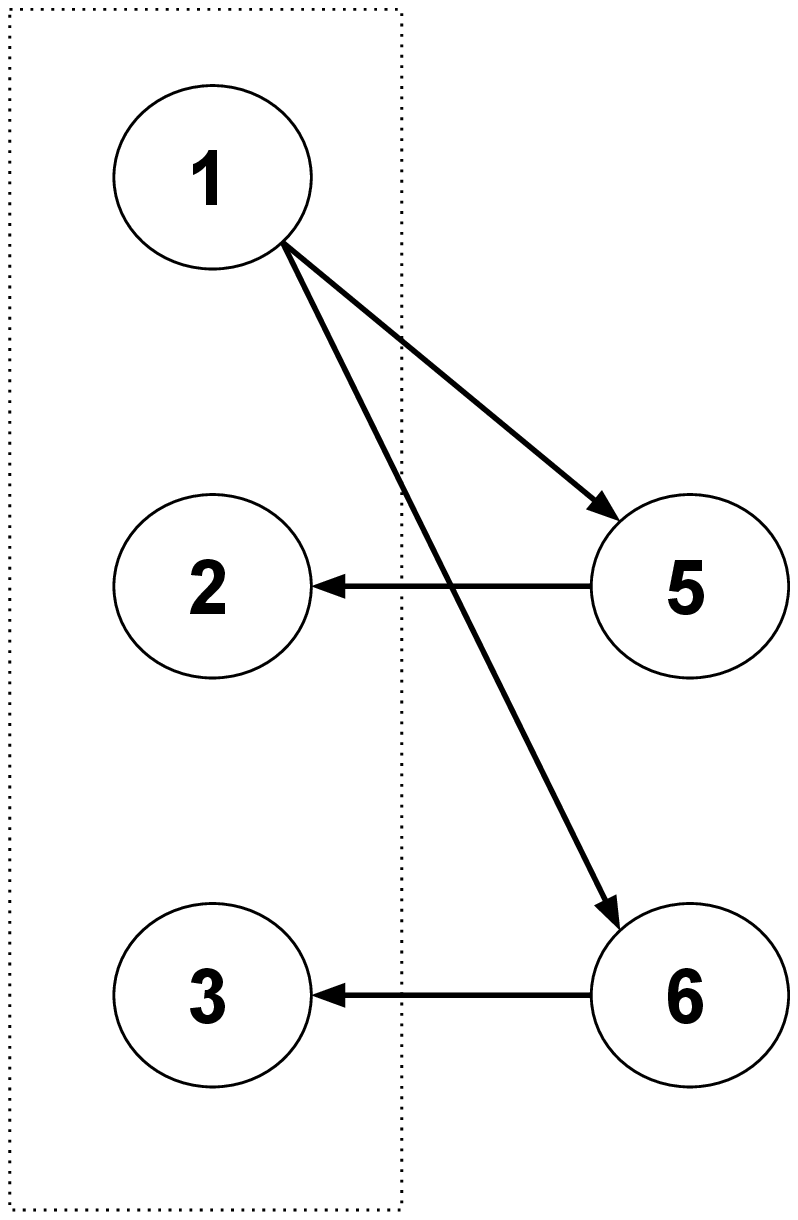}
		\caption{}
		\label{rt41}
	\end{subfigure}%
	\begin{subfigure}{.31\textwidth}
		\centering
		\includegraphics[width=15pc]{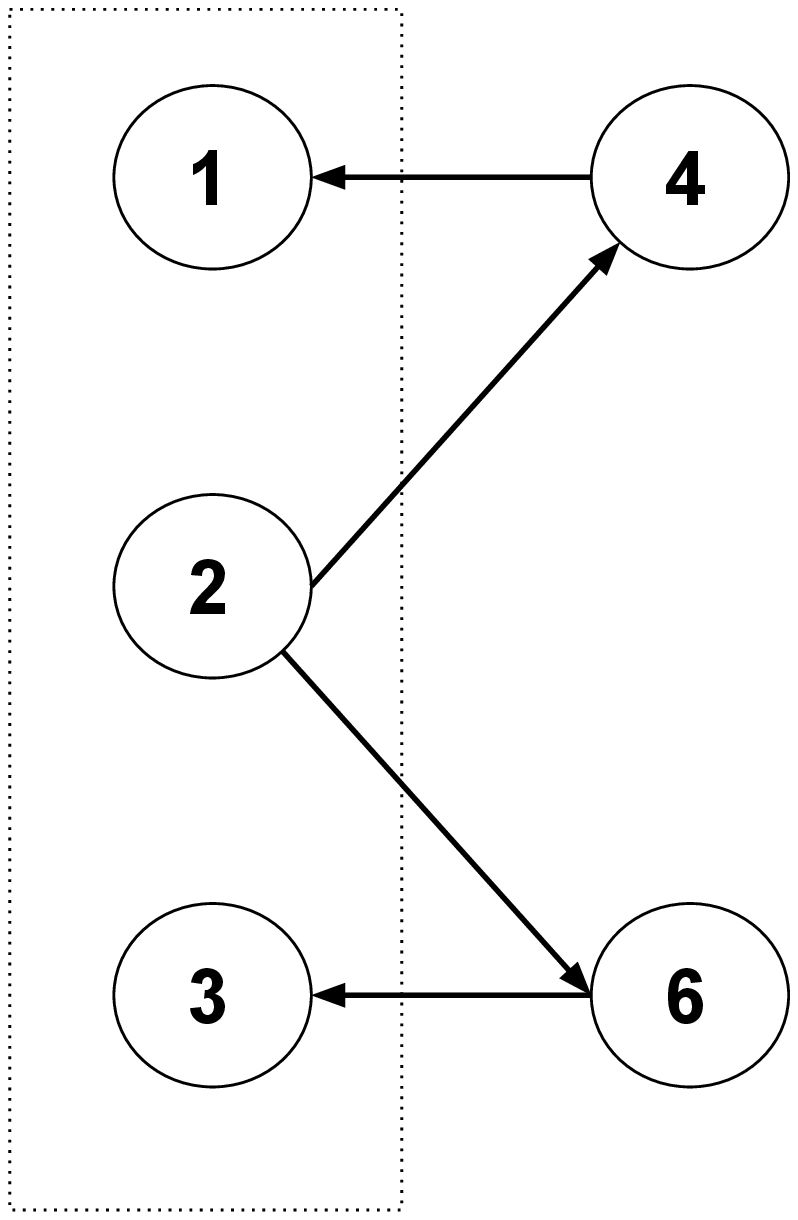}
		\caption{}
		\label{rt42}
	\end{subfigure}
	\begin{subfigure}{.31\textwidth}
		\centering
		\includegraphics[width=15pc]{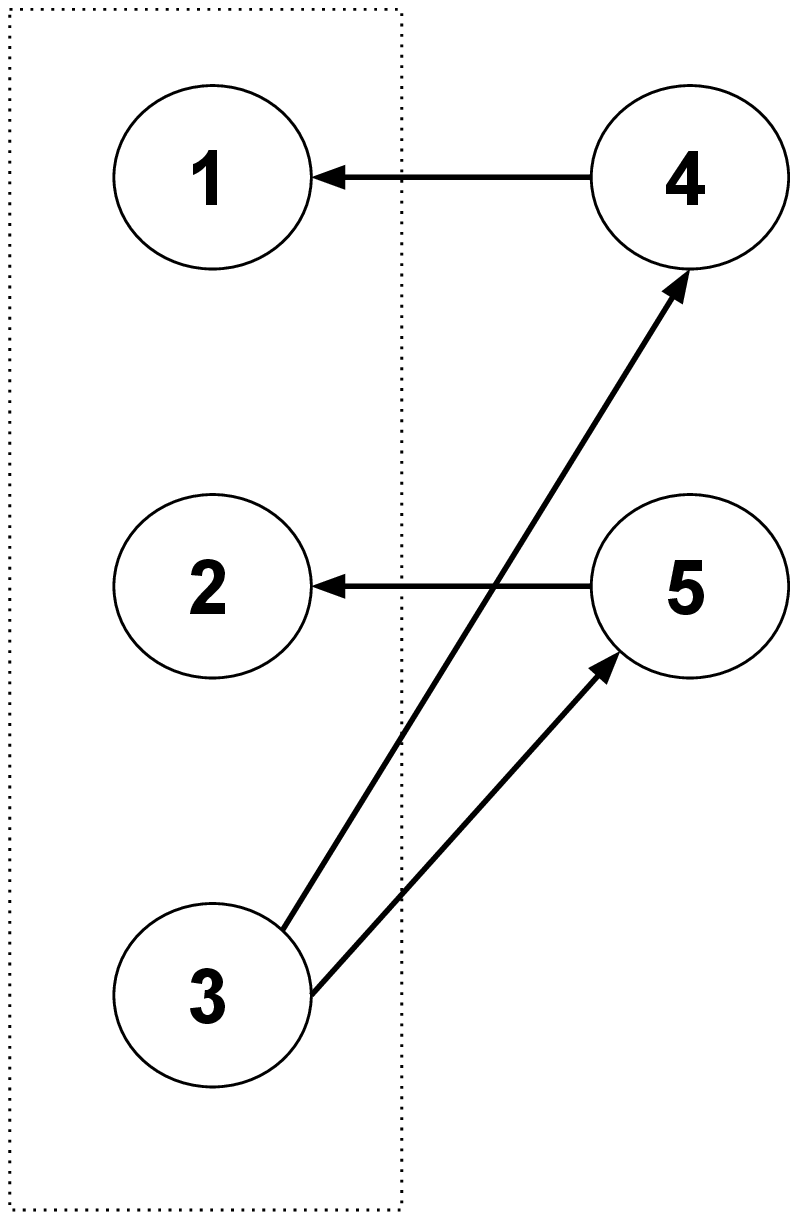}
		\caption{}
		\label{rt43}
	\end{subfigure}
	\caption{Figures showing rooted trees of inner vertices $ 1,2,3$ of $ \mathcal{G}_3 $, respectively.}
\end{figure*}
Conditions \textit{c}$ 1 $ and \textit{c}$ 2 $ are illustrated for $ \mathcal{G}_3 $ as follows.
The rooted trees $T_1$, $ T_2 $ and $ T_3 $ have no non-inner vertices at depth $ \geq2 $ and hence \textit{c}$ 1 $ need not be verified. From Table \ref{table5}, it is clear that $ b_{i,j} = 0$ for each $ i \in \lbrace1,2,3\rbrace $ and $ j \in V(\mathcal{G})\backslash V(T_i) $.
\begin{table}[h!]
\centering
	\hspace{5mm}
\begin{tabular}{|c|c|c|c|} 
	\hline
	$ T_i $ & $ V_{NI}(i) $ & $ j \in V(\mathcal{G}_3) \backslash V(T_i) $ & $ b_{i,j} $ \\ 
	\hline\hline
	$ T_1 $ & $ \lbrace 5,6 \rbrace $ & $ \lbrace4\rbrace $ & $ 0 $ \\ 
	\hline
	$ T_2$ & $ \lbrace 4,6 \rbrace $ & $ \lbrace5\rbrace $ & $ 0 $ \\
	\hline
	$ T_3$ & $ \lbrace 4,5 \rbrace $ & $ \lbrace6\rbrace $ & $ 0 $\\
	\hline 

\end{tabular}\\
\caption{Table that illustrates \textit{c}$ 2 $ for $ \mathcal{G}_3 $.}
\label{table5}	
\end{table}
It is thus verified that \textit{c}$ 1 $ and \textit{c}$ 2 $ are satisfied by $ \mathcal{G}_3 $.
\end{ex}
\begin{ex}
\label{exam4}
Consider $\mathcal{G}_4$, a side-information graph which is a $4$-IC structure shown in Fig. \ref{f5}.
\begin{figure}[h!]
	\includegraphics[height=\columnwidth,width=\columnwidth,angle=0]{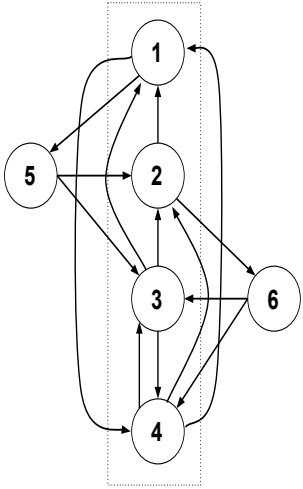}
	\caption{$ 4 $-IC structure $\mathcal{G}_4$ with $ V_I=\lbrace1,2,3,4\rbrace$.}
	\label{f5}
\end{figure}
It is a $ 4 $-IC structure with inner vertex set  $V_I=\lbrace1,2,3,4\rbrace$ since 
\begin{enumerate}
	\item there are no cycles with only one vertex from the set $ \lbrace1,2,3,4\rbrace $ in $\mathcal{G}_4$ (i.e., no I-cycles),
	\item using the rooted trees for each vertex in the set $\lbrace1,2,3,4\rbrace$, which are given in Fig. \ref{rt51}, \ref{rt52}, \ref{rt53} and \ref{rt54} respectively, it is verified that there exists a unique path between any two different vertices in $ V_I $ in $ \mathcal{G}_4 $ and does not contain any other vertex in $ V_I $ (i.e., unique I-path between any pair of inner vertices),
	\item $\mathcal{G}_4$ is the union of all the $4$ rooted trees
\end{enumerate}
and it has no cycles consisting of only non-inner vertices.
\begin{figure*}[!t]
	\centering
	\begin{subfigure}{.31\textwidth}
		\hspace{-4mm}
		\includegraphics[width=15pc]{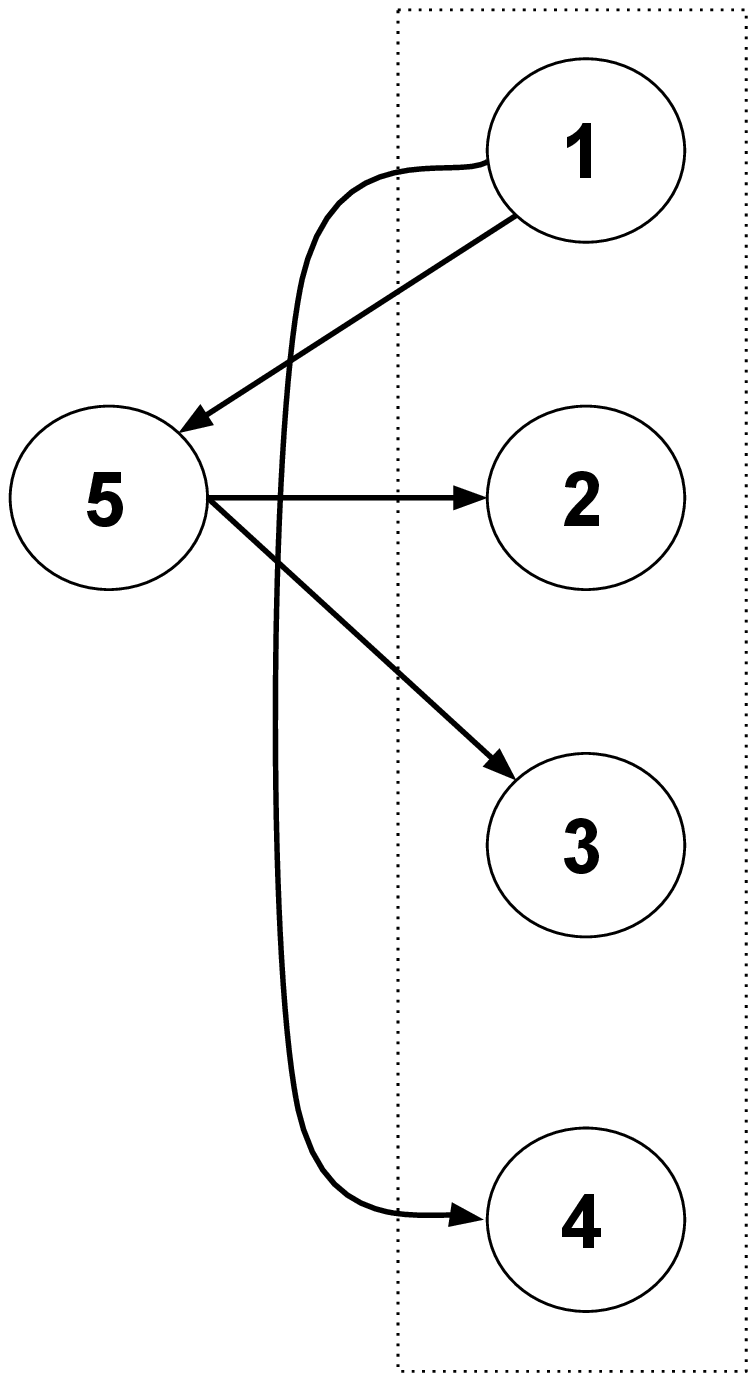}
		\caption{}
		\label{rt51}
	\end{subfigure}%
	\begin{subfigure}{.31\textwidth}
		\hspace{-4mm}
		\includegraphics[width=15pc]{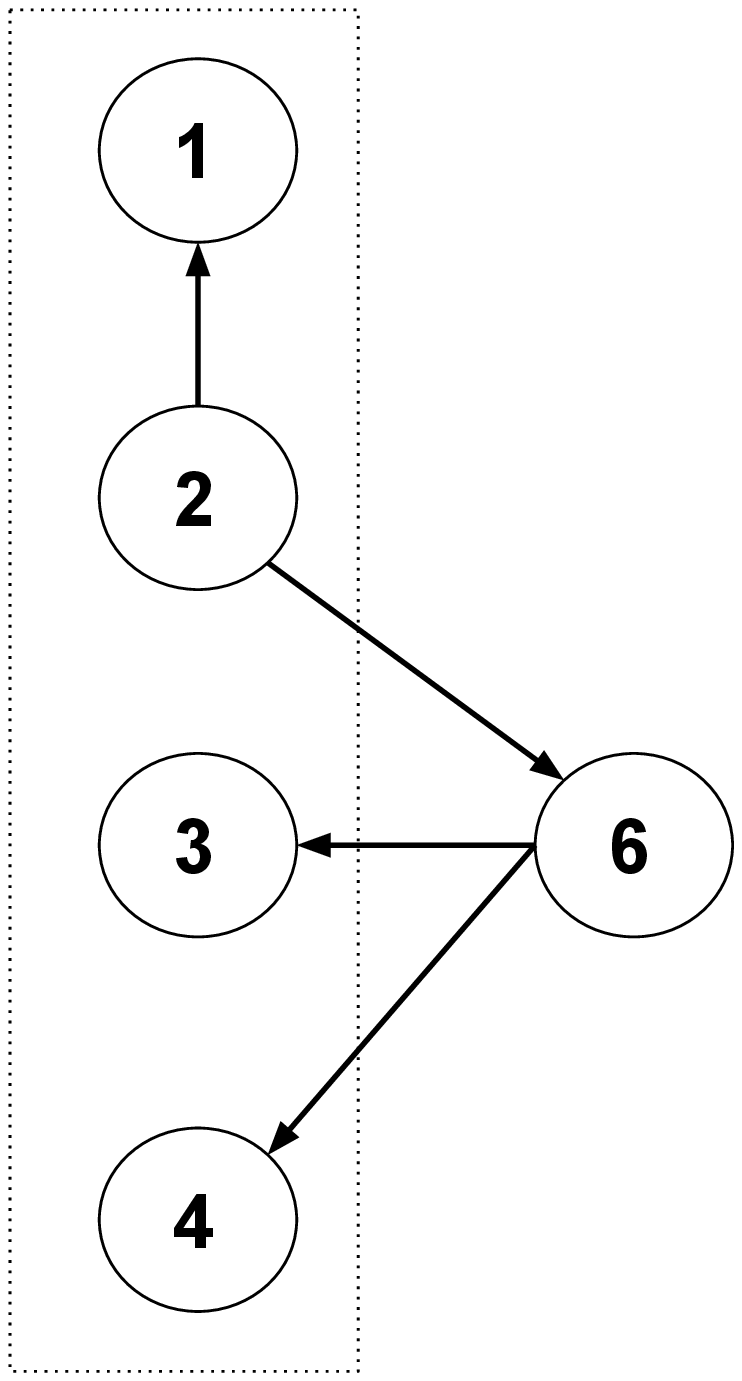}
		\caption{}
		\label{rt52}
	\end{subfigure}
	\begin{subfigure}{.31\textwidth}
		\hspace{-4mm}
		\includegraphics[width=15pc]{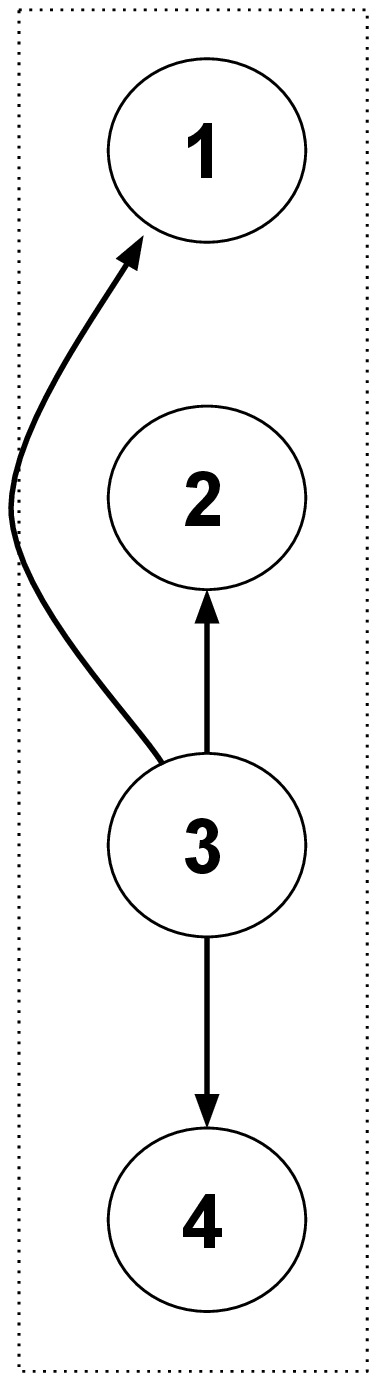}
		\caption{}
		\label{rt53}
	\end{subfigure}
	\begin{subfigure}{.31\textwidth}
		\hspace{-4mm}
		\includegraphics[width=15pc]{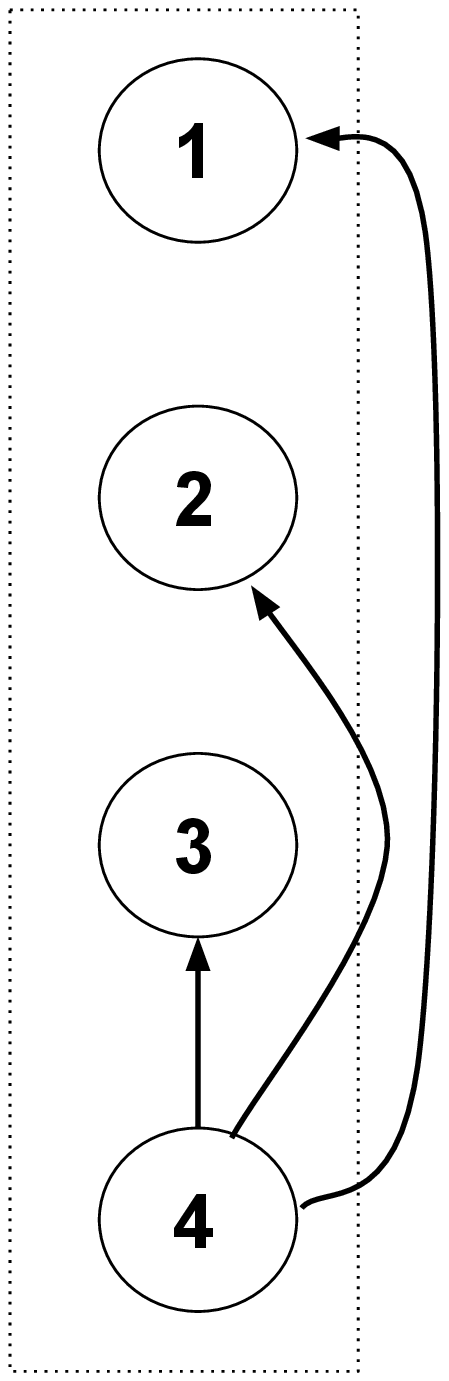}
		\caption{}
		\label{rt54}
	\end{subfigure}%
	\caption{Figures showing rooted trees of inner vertices $ 1,2,3,4 $ of $ \mathcal{G}_4 $, respectively.}
\end{figure*}

Conditions \textit{c}$ 1 $ and \textit{c}$ 2 $ are illustrated for $ \mathcal{G}_4 $ as follows. The rooted trees $T_1$, $ T_2 $, $ T_3 $ and $ T_4 $ have no non-inner vertices at depth $ \geq2 $ and hence \textit{c}$ 1 $ need not be verified. Verification of \textit{c}$ 2 $ is done using Table \ref{table6}.
\begin{table}[h!]
\centering	
\hspace{7mm}
\begin{tabular}{|c|c|c|c|} 
	\hline
	$ T_i $ & $ V_{NI} $ & $ j \in V(\mathcal{G}_4)\backslash V(T_i) $ & $ b_{i,j} $ \\ 
	\hline\hline
	$ T_1 $ & $ \lbrace 5\rbrace $ & $ \lbrace6\rbrace $ & $ 0 $ \\ 
	\hline
	$ T_2$ & $ \lbrace 6 \rbrace $ & $ \lbrace5\rbrace $ & $ 0 $ \\
	\hline
	$ T_3$ & $ \phi $ & $ \lbrace5,6\rbrace $ & $ - $\\
	\hline 
	$ T_4$ & $ \phi $ & $ \lbrace 5,6 \rbrace $ & $ - $\\
	\hline 
\end{tabular}\\
\caption{Table that illustrates \textit{c}$ 2 $ for $ \mathcal{G}_4 $.}
\label{table6}
\end{table}It is thus verified that \textit{c}$ 1 $ and \textit{c}$ 2 $ are satisfied by $ \mathcal{G}_4$.
\end{ex}
The following three examples (Examples \ref{exam5}, \ref{exam6} and \ref{exam7}) illustrate Theorem \ref{thm1} for some IC structures having at least one cycle consisting of only non-inner vertices. 
\begin{ex}
\label{exam5}
Consider $\mathcal{G}_5$, a side-information graph which is a $5$-IC structure, shown in Fig. \ref{f3}.
\begin{figure}[h!]

	\includegraphics[height=\columnwidth,width=\columnwidth,angle=0]{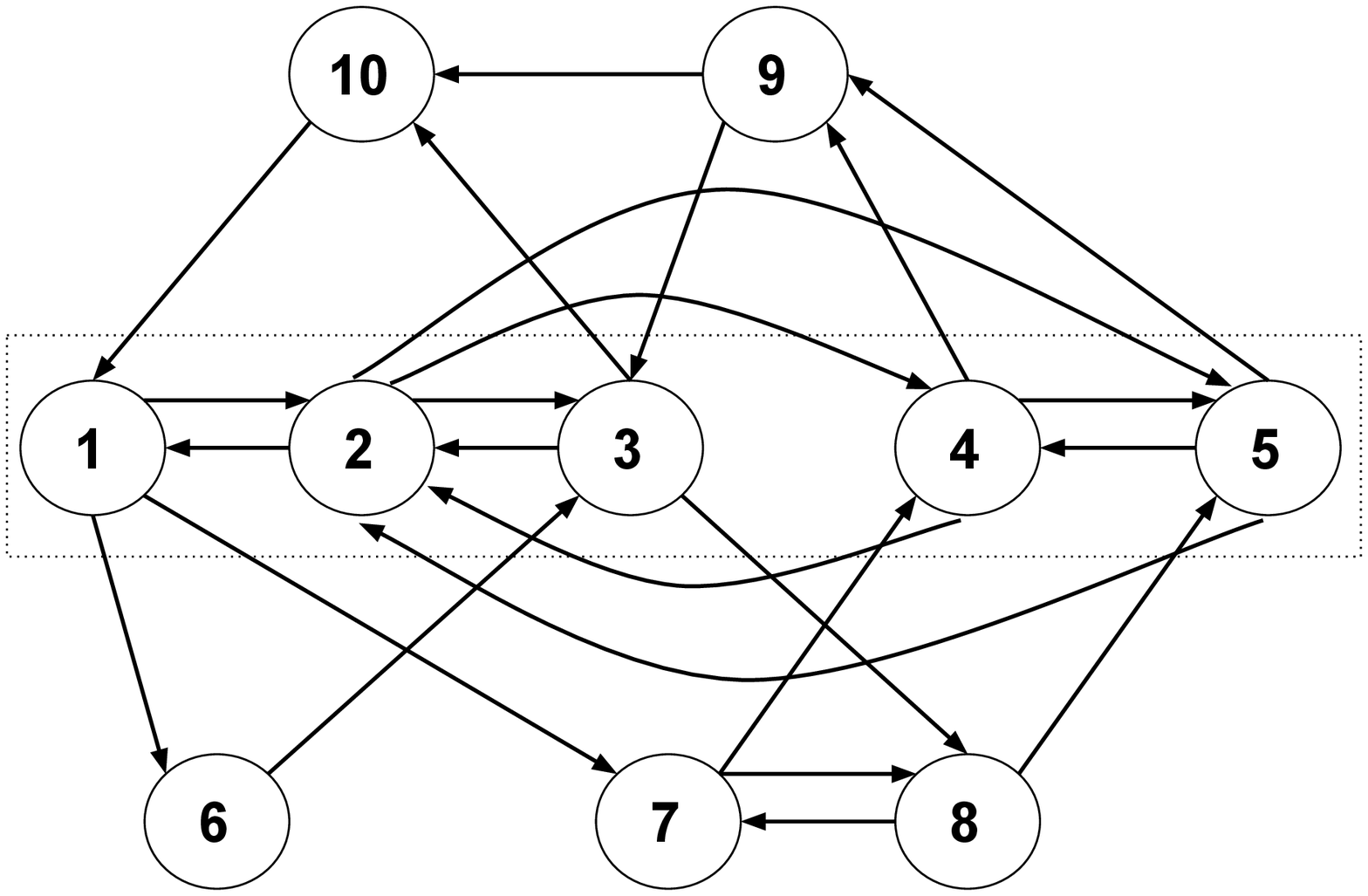}
	\caption{$5$-IC structure $\mathcal{G}_5$ with $V_I=\lbrace1,2,3,4,5\rbrace$.}
		\label{f3}
\end{figure}
It can be easily verified that it is a $ 5 $-IC structure with inner vertex set $V_I=\lbrace1,2,3,4,5\rbrace$ 
because
\begin{enumerate}
\item there are no cycles with only one vertex from the set $\lbrace1,2,3,4,5\rbrace$ in $\mathcal{G}_5$ (i.e., no I-cycles),
\item using the rooted trees for each vertex in the set $\lbrace1,2,3,4,5\rbrace$, which are given in Fig. \ref{rt31}, \ref{rt32}, \ref{rt33}, \ref{rt34} and \ref{rt35} respectively, it is verified that there exists a unique path between any two different vertices in $ V_I $ in $ \mathcal{G}_5 $ and does not contain any other vertex in $ V_I $ (i.e., unique I-path between any pair of inner vertices),
\item $\mathcal{G}_5$ is the union of all the $5$ rooted trees.
\end{enumerate}
Note that the vertices 7 and 8 form a cycle consisting of non-inner vertices. 
\begin{figure*}[!t]
	\centering
	\begin{subfigure}{.31\textwidth}
		\hspace{-6mm}
		\includegraphics[width=15pc]{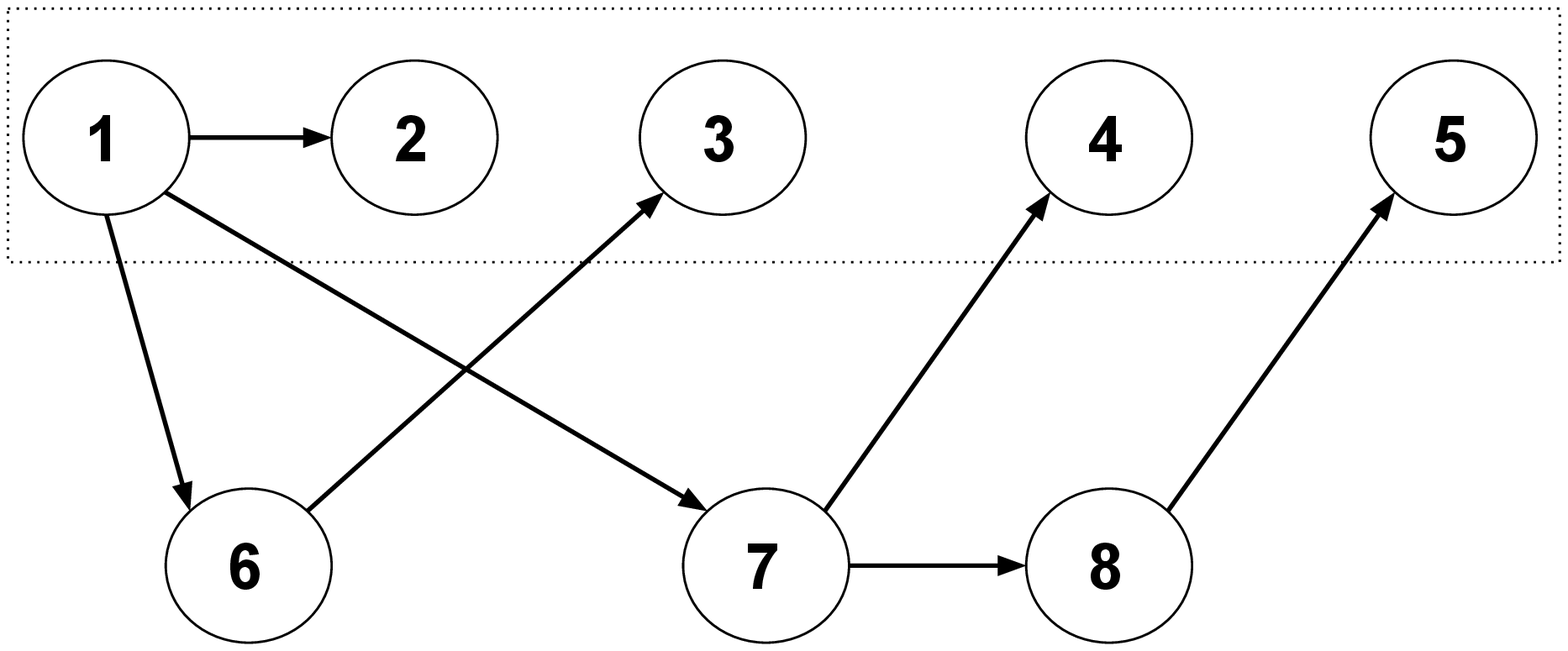}
		\caption{}
		\label{rt31}
	\end{subfigure}%
	\begin{subfigure}{.31\textwidth}
		\hspace{-6mm}
		\includegraphics[width=15pc]{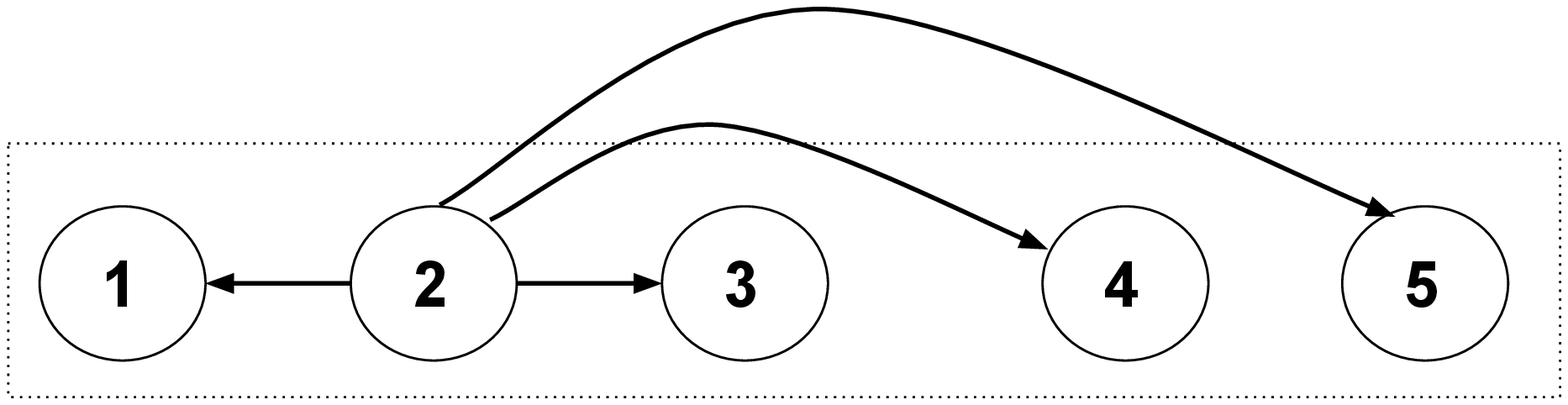}
		\caption{}
		\label{rt32}
	\end{subfigure}
	\begin{subfigure}{.31\textwidth}
		\hspace{-6mm}
		\includegraphics[width=15pc]{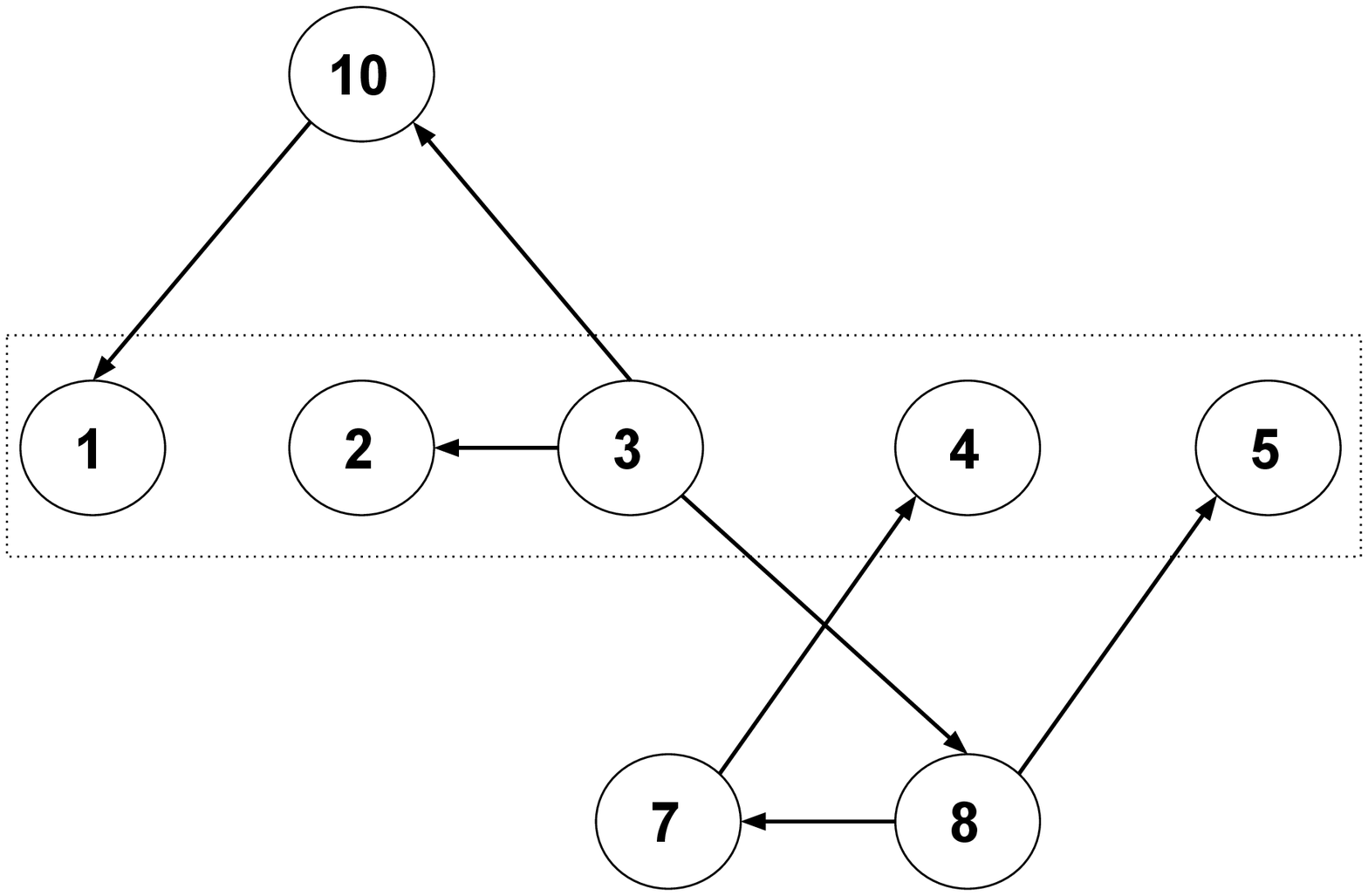}
		\caption{}
		\label{rt33}
	\end{subfigure}
	\centering
	\begin{subfigure}{.31\textwidth}
		\hspace{-6mm}
		\includegraphics[width=15pc]{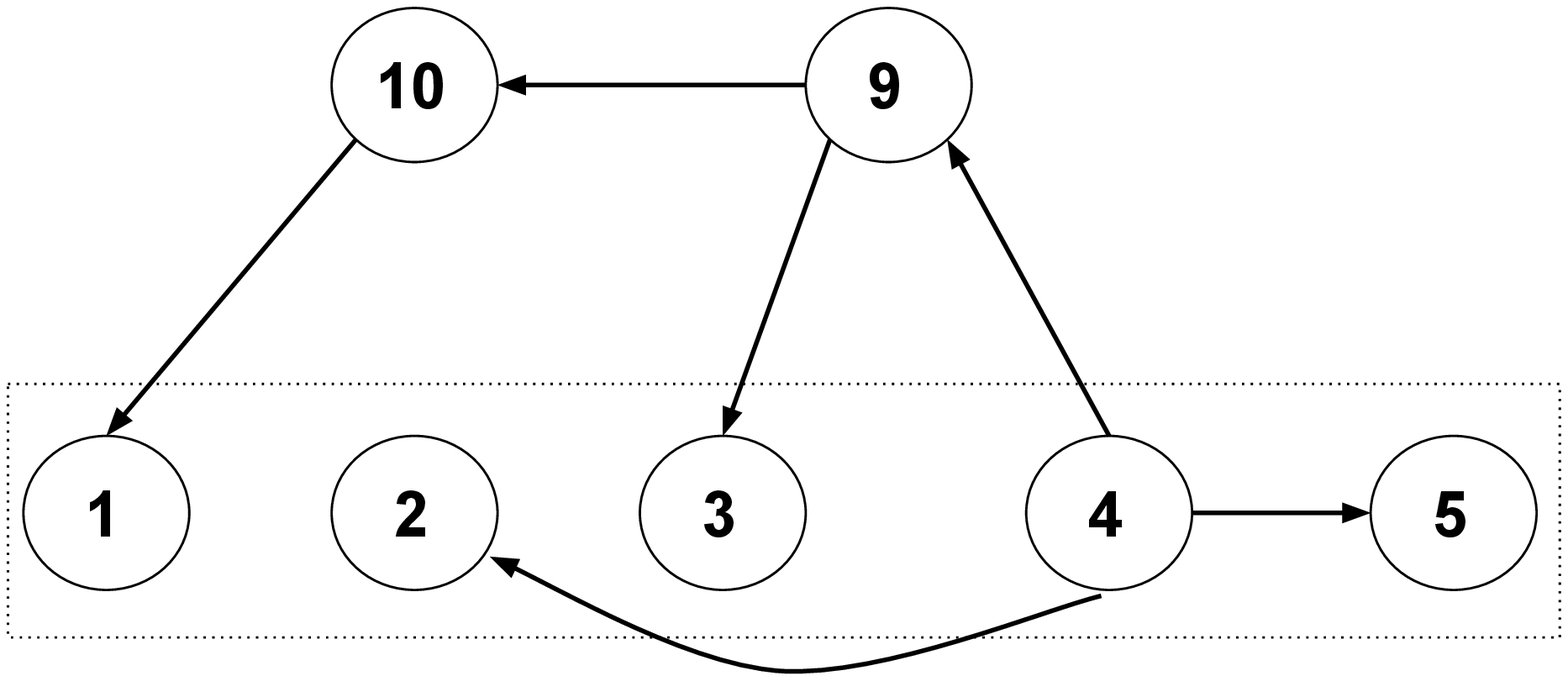}
		\caption{}
		\label{rt34}
	\end{subfigure}%
	\begin{subfigure}{.31\textwidth}
		\hspace{-6mm}
		\includegraphics[width=15pc]{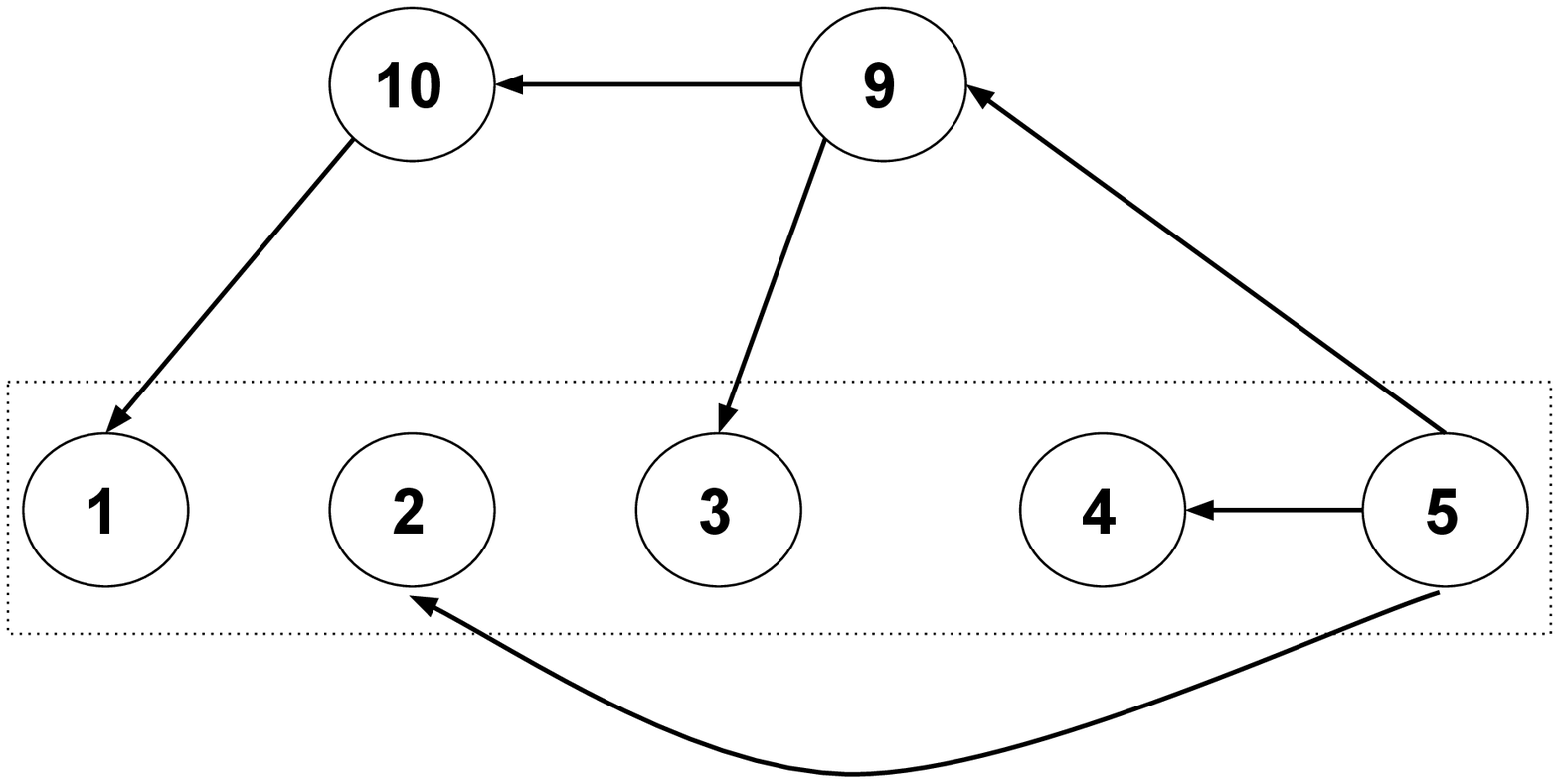}
		\caption{}
		\label{rt35}
	\end{subfigure}
	\caption{Figures showing rooted trees of inner vertices $ 1,2,3,4,5 $ of $ \mathcal{G}_5 $, respectively.}
\end{figure*}

\begin{table}[h!]
\centering
\begin{tabular}{|c|c|c|c|}
	\hline
	$ T_i $ & $ V_{NI}(i) $ & $ j \in V_{NI}(T_i)\backslash  N^+_{T_i}(i) $ & $ a_{i,j} $ \\
	\hline\hline
	$ T_1 $ & $ \lbrace6,7,8\rbrace $ & $ \lbrace8\rbrace $ & $ 1 $ \\ 	
	\hline
	$ T_2 $ & $ \phi $ & $ \phi $ & $ - $ \\
	\hline
	$ T_3 $ & $ \lbrace7,8,10\rbrace $ & $ \lbrace7\rbrace $ & $ 1 $ \\ 	 
	\hline
	$ T_4 $ & $ \lbrace9,10\rbrace $ & $ \lbrace10\rbrace $ & $ 1 $ \\ 		
	\hline
	$ T_5 $ & $ \lbrace9,10\rbrace $ & $ \lbrace10\rbrace $ & $ 1 $ \\
	\hline
\end{tabular}\\
\caption{Table that illustrates \textit{c}$ 1 $ for $ \mathcal{G}_5 $.}
\label{table7}
\end{table}
\begin{table}[h!]
\centering
	\hspace{4mm}
\begin{tabular}{|c|c|c|c|}
\hline
$ T_i $ & $ V_{NI}(i) $ & $ j \in V(\mathcal{G}_5)\backslash V_{T_i}$ & $ b_{i,j} $ \\
\hline\hline
	$ T_1 $ & $ \lbrace6,7,8\rbrace $ & $ \lbrace9,10\rbrace $ & $0$, $ 0 $\\ 	
\hline
	$ T_2 $ & $ \phi $ & $ \lbrace6,7,8,9,10\rbrace $ & $ - $ \\ 	
\hline
	$ T_3$ & $ \lbrace7,8,10\rbrace $ & $ \lbrace6,9\rbrace $ & $ 0 $, $ 0 $\\ 	
\hline
	$ T_4 $ & $ \lbrace9,10\rbrace $ & $ \lbrace6,7,8\rbrace $ & $ 0 $, $ 0 $, $ 0 $\\ 	
\hline
	$ T_5 $ & $ \lbrace9,10\rbrace $ & $ \lbrace6,7,8\rbrace $ & $0$, $ 0 $, $ 0 $\\ 	
\hline
\end{tabular}\\
\caption{Table that illustrates \textit{c}$ 2 $ for $ \mathcal{G}_5 $.}
\label{tab4}
\end{table}
 Conditions \textit{c}$ 1 $ and \textit{c}$ 2 $ are illustrated for $ \mathcal{G}_5 $ in Table \ref{table7} and Table \ref{tab4}, respectively. As a result, \textit{Algorithm} $ 1 $ can be used to decode an index code obtained by using \textit{Construction} $ 1$ on the IC structure $ \mathcal{G}_5 $.
\end{ex}
\begin{ex}
\label{exam6}
	Consider $\mathcal{G}_6$, a side-information graph which is a $6$-IC structure, shown in Fig. \ref{f7}.
	\begin{figure}[h!]

		\includegraphics[height=\columnwidth,width=\columnwidth,angle=0]{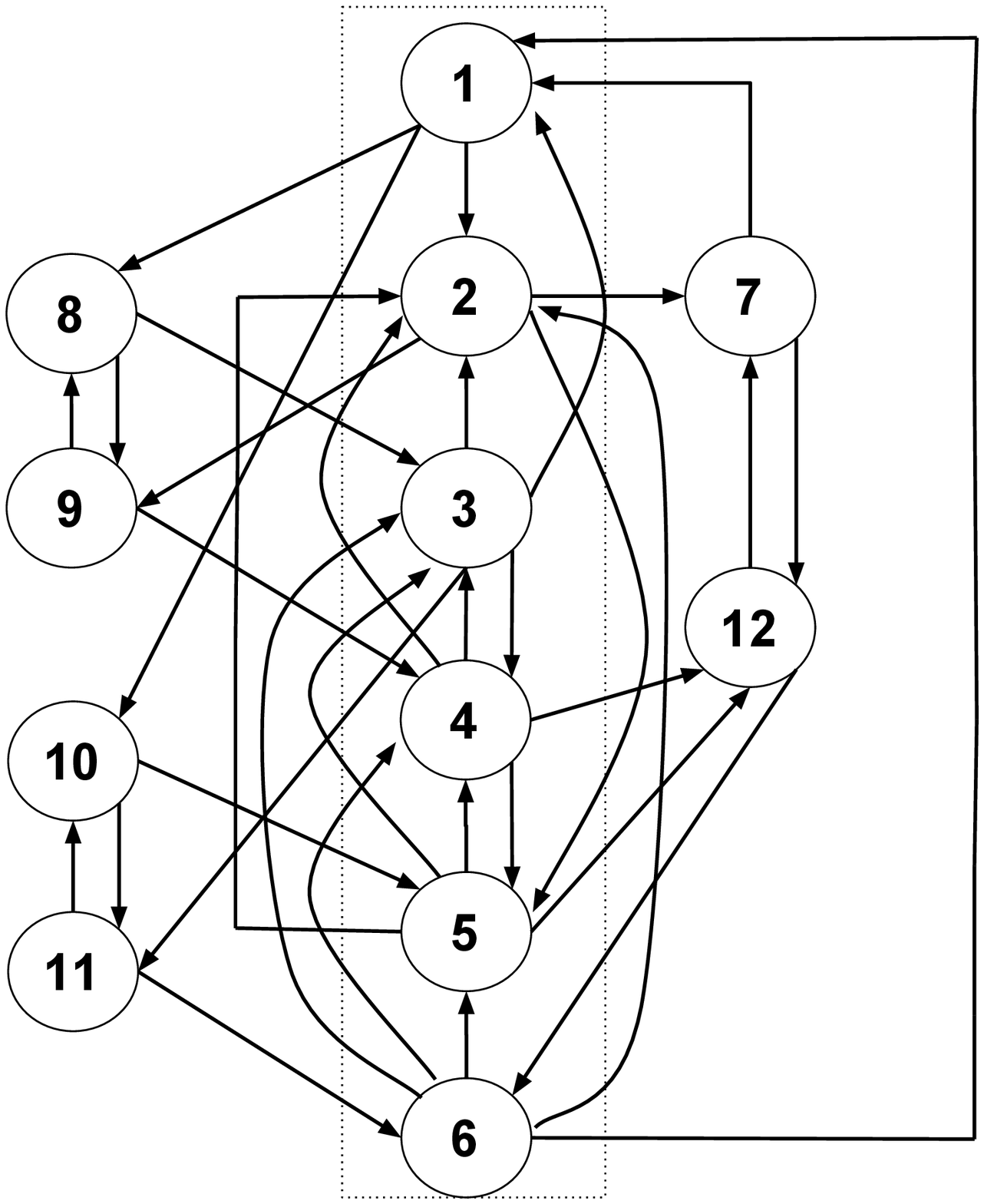}
		\caption{$6$-IC structure $\mathcal{G}_6$ with $V_I=\lbrace1,2,3,4,5,6\rbrace$.}
				\label{f7}
	\end{figure}$\mathcal{G}_6$ is a $ 6 $-IC structure with inner vertex set $V_I=\lbrace1,2,3,4,5,6\rbrace$ because
	\begin{enumerate}
		\item there are no cycles with only one vertex from the set $\lbrace1,2,3,4,5,6\rbrace$ in $\mathcal{G}_6$ (i.e., no I-cycles),
		\item using the rooted trees for each vertex in the set, $\lbrace1,2,3,4,5,6\rbrace$, which are given in Fig. \ref{rt71}, \ref{rt72}, \ref{rt73}, \ref{rt74}, \ref{rt75} and \ref{rt76} respectively, it is verified that there exists a unique path between any two different vertices in $ V_I $ in $ \mathcal{G}_6 $ and does not contain any other vertex in $ V_I $ (i.e., unique I-path between any pair of inner vertices),
		\item $\mathcal{G}_6$ is the union of all the $6$ rooted trees.
	\end{enumerate}
Also, notice that there are three disjoint cycles each of them consisting of only the non-inner vertices.
\begin{figure*}[!t]
		\centering
		\begin{subfigure}{.31\textwidth}
			\hspace{-6mm}
			\includegraphics[width=15pc]{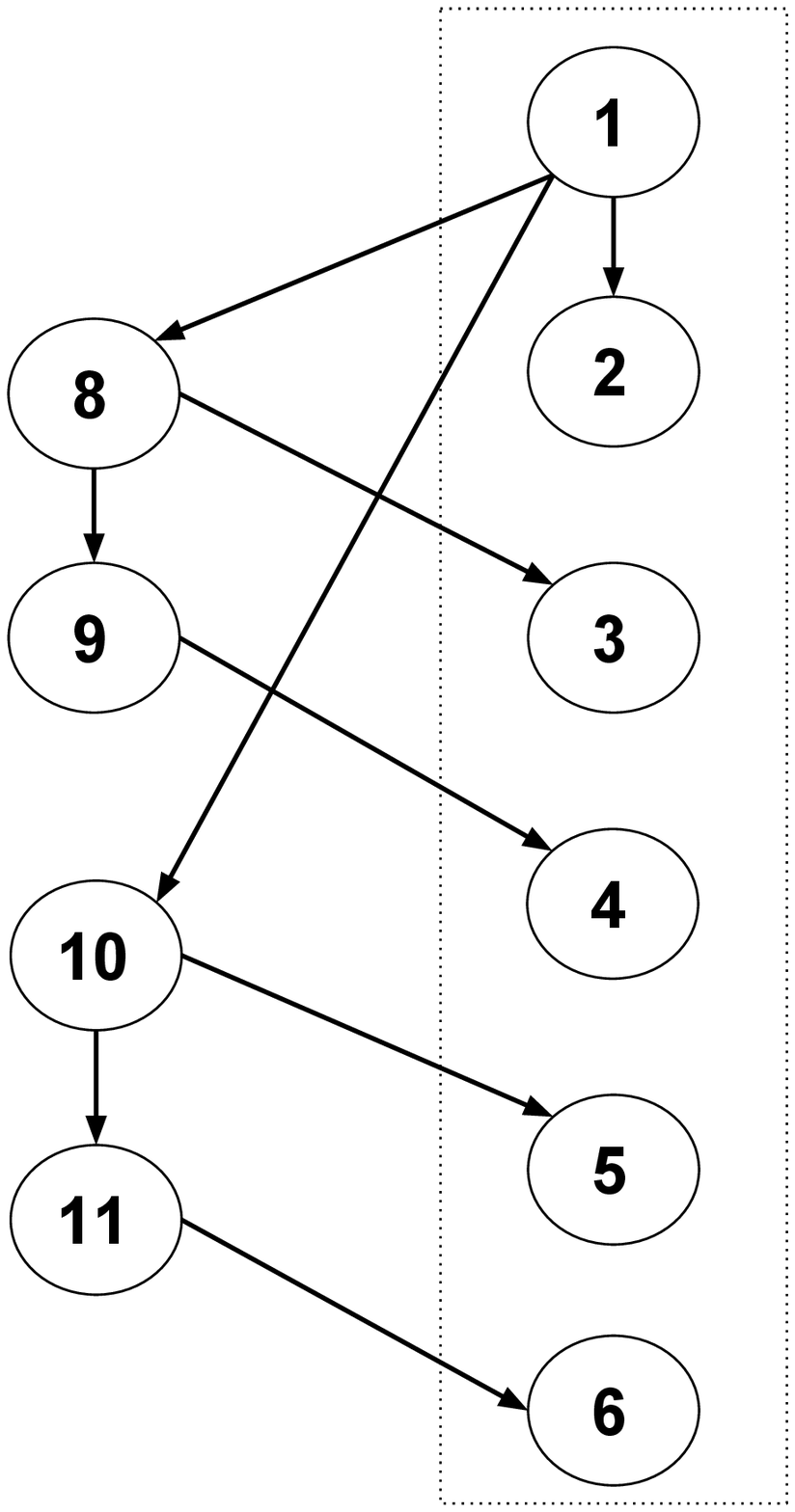}
			\caption{}
			\label{rt71}
		\end{subfigure}%
		\begin{subfigure}{.31\textwidth}
			\hspace{-6mm}
			\includegraphics[width=15pc]{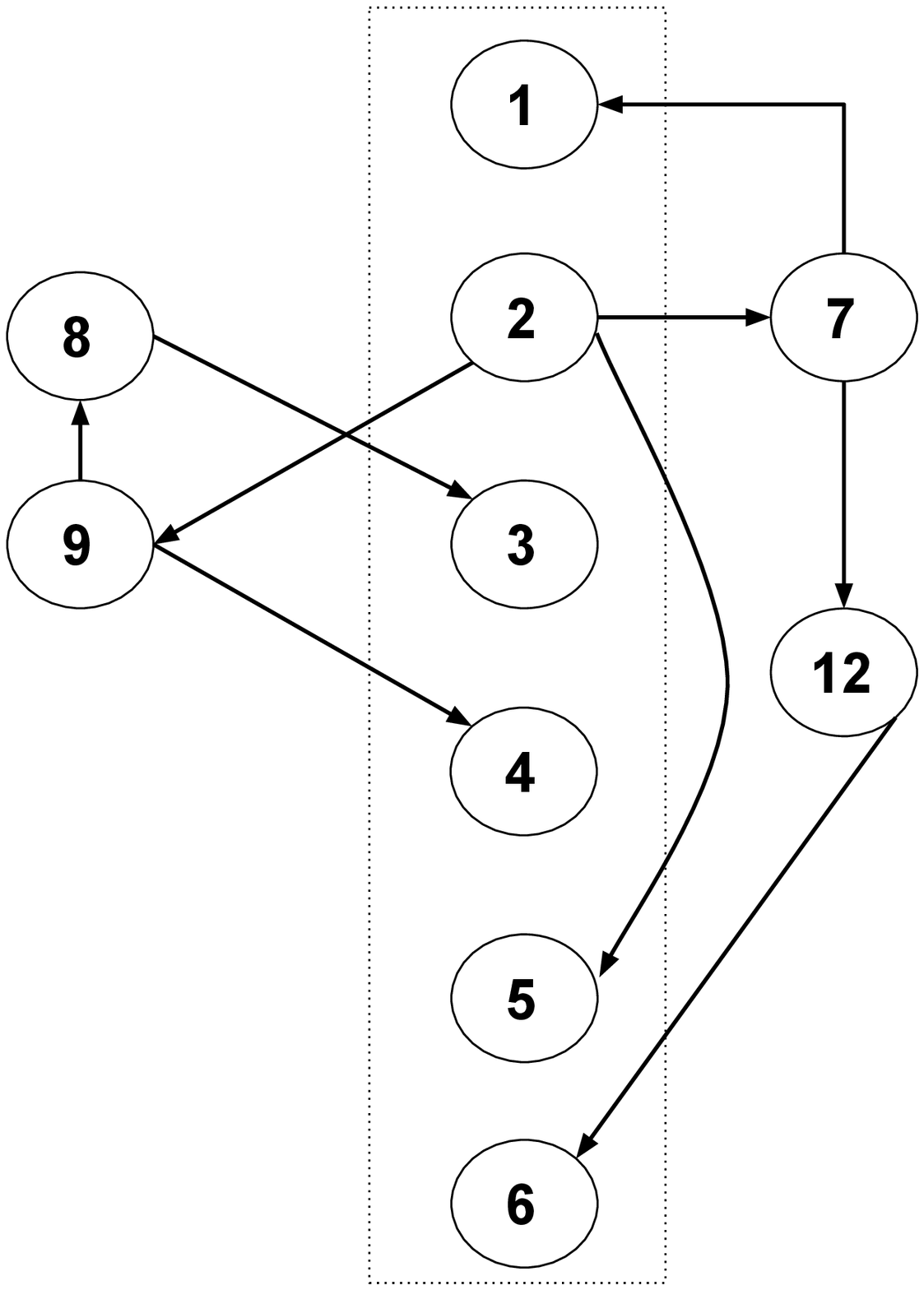}
			\caption{}
			\label{rt72}
		\end{subfigure}
		\begin{subfigure}{.31\textwidth}
			\hspace{-6mm}
			\includegraphics[width=15pc]{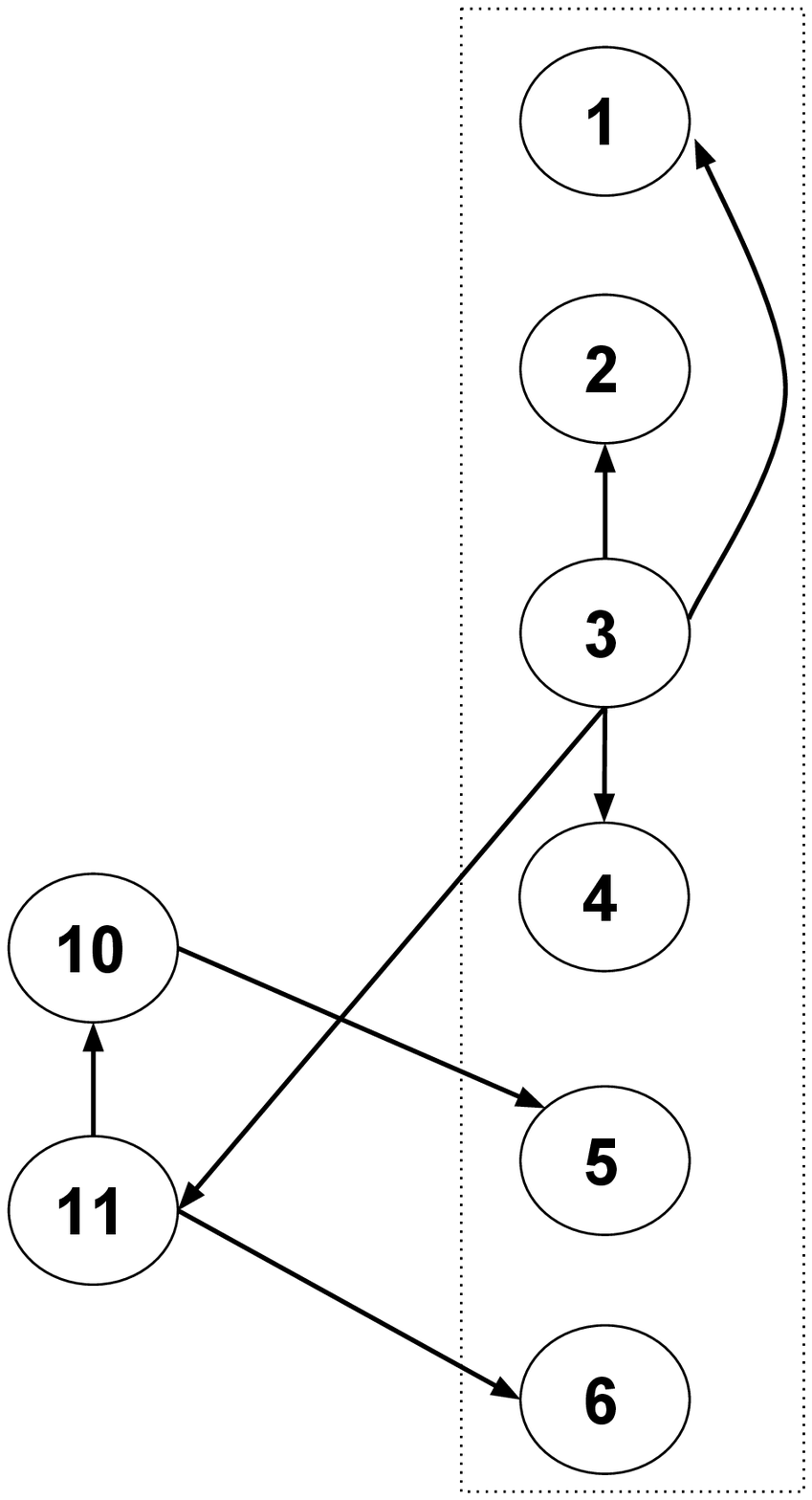}
			\caption{}
			\label{rt73}
		\end{subfigure}
		\centering
		\begin{subfigure}{.31\textwidth}
			\hspace{-6mm}
			\includegraphics[width=15pc]{RT_1_4}
			\caption{}
			\label{rt74}
		\end{subfigure}%
		\begin{subfigure}{.37\textwidth}
			\hspace{-2mm}
			\includegraphics[width=15pc]{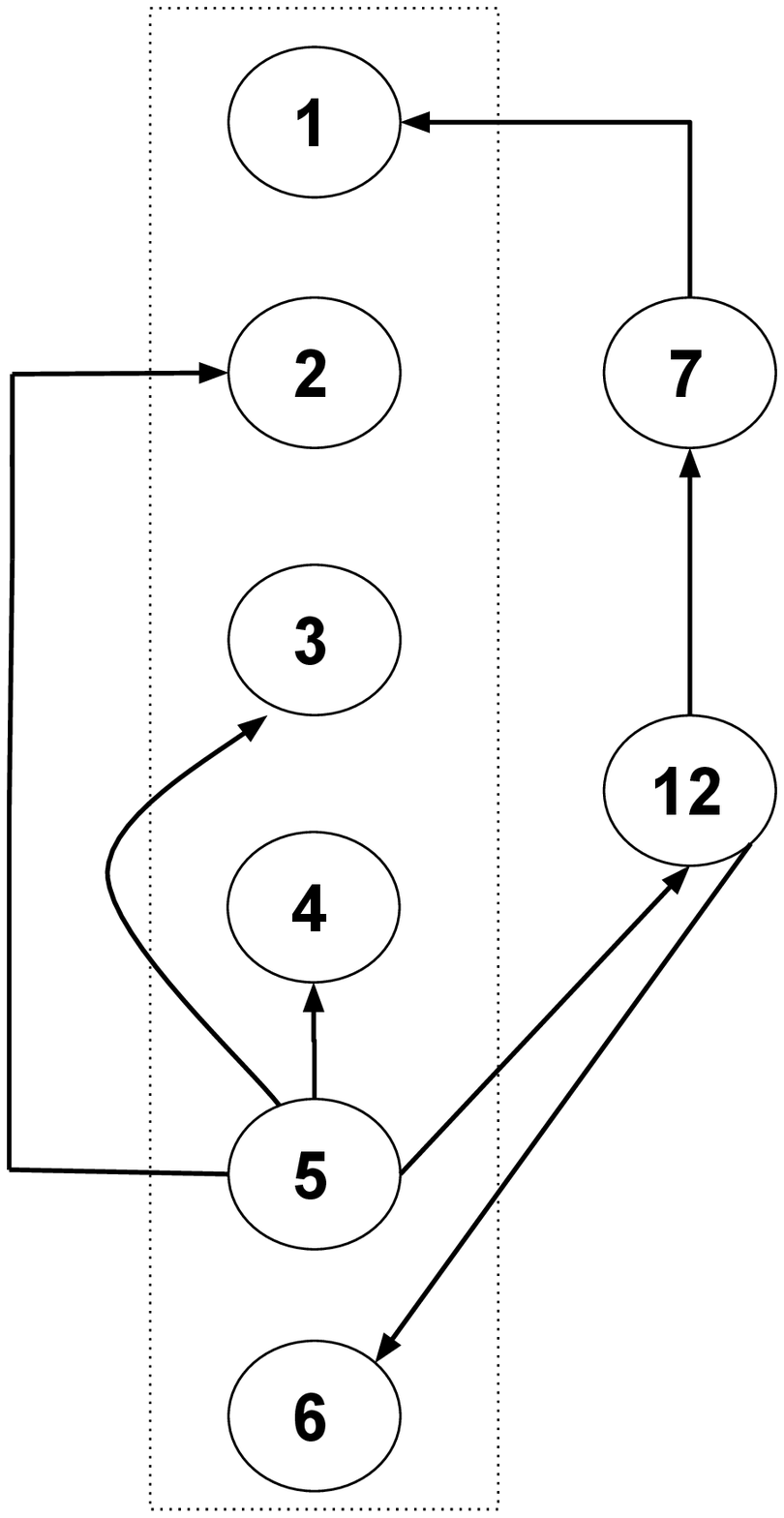}
			\caption{}
			\label{rt75}
		\end{subfigure}
		\begin{subfigure}{.31\textwidth}
			\hspace{-6mm}
			\includegraphics[width=15pc]{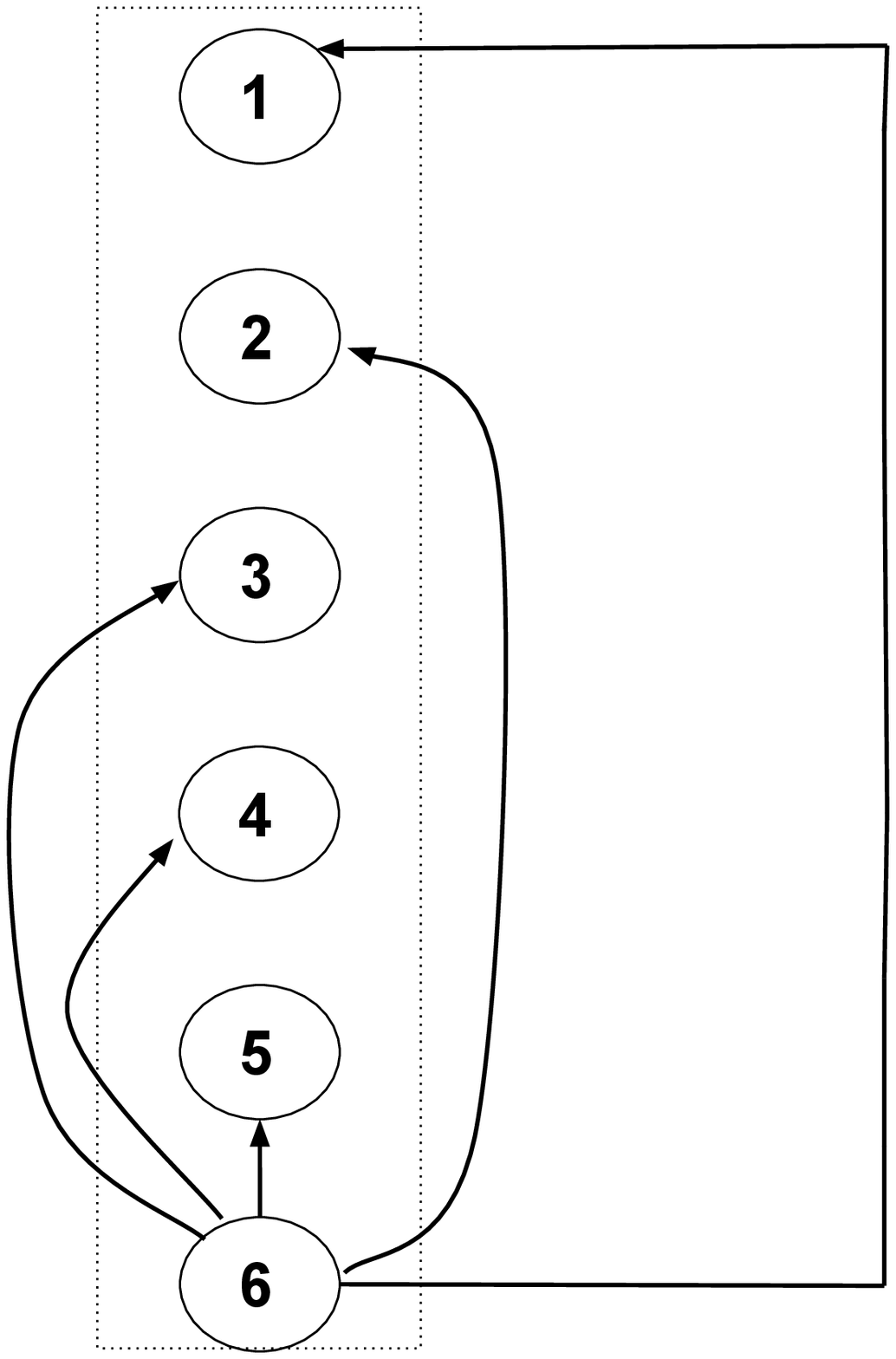}
			\caption{}
			\label{rt76}
		\end{subfigure}
		\caption{Figures showing rooted trees of inner vertices $ 1,2,3,4,5,6 $ of $ \mathcal{G}_6 $, respectively.}
	\end{figure*}
		\begin{table}
	\centering	
		\begin{tabular}{|c|c|c|c|}
			\hline
		 $ T_i $ & $ V_{NI}(i) $ & $ j \in V(T_i)\backslash\lbrace V_I \cup N^+_{T_i}(i) \rbrace $ & $ a_{i,j} $ \\
			\hline\hline
			$ T_1 $ & $ \lbrace8,9,10,11\rbrace $ & $ \lbrace9,11\rbrace $ & $ 1 $, $ 1 $ \\ 	
			\hline
			$ T_2 $ & $ \lbrace7,8,9,12\rbrace $ & $ \lbrace8,12\rbrace $ & $ 1 $, $ 1 $ \\
			\hline
			$ T_3 $ & $ \lbrace10,11\rbrace $ & $ \lbrace10\rbrace $ & $ 1 $ \\ 	 
			\hline
			$ T_4 $ & $ \lbrace7,12\rbrace $ & $ \lbrace7\rbrace $ & $ 1 $ \\ 		
			\hline
			$ T_5 $ & $ \lbrace7,12\rbrace $ & $ \lbrace7\rbrace $ & $ 1 $ \\
			\hline
			$ T_6 $ & $\phi$ & $ \phi $ & $ - $ \\
			\hline
		\end{tabular}\\
		\caption{Table that verifies \textit{c}$ 1 $ for $ \mathcal{G}_6 $.}
		\label{table11}
	\end{table}
	\begin{table}[h!]
\centering
		\begin{tabular}{|c|c|c|c|}
			\hline
			$ T_i $ & $ V_{NI}(i) $ & $ j \in V(\mathcal{G}_6)\backslash V(T_i) $ & $ b_{i,j} $ \\
			\hline\hline
			$ T_1 $ & $ \lbrace8,9,10,11\rbrace $ & $ \lbrace7,12\rbrace $ & $0$, $ 0 $\\ 	
			\hline
			$ T_2 $ & $ \lbrace7,8,9,12\rbrace $ & $ \lbrace10,11\rbrace $ & $ 0 $, $ 0 $ \\ 	
			\hline
			$ T_3$ & $ \lbrace10,11\rbrace $ & $ \lbrace7,8,9,12\rbrace $ & $ 0 $, $ 0 $, $ 0 $, $ 0 $\\ 	
			\hline
			$ T_4 $ & $ \lbrace7,12\rbrace $ & $ \lbrace8,9,10,11\rbrace $ & $ 0 $, $ 0 $, $ 0 $, $ 0 $\\ 	
			\hline
			$ T_5 $ & $ \lbrace7,12\rbrace $ & $ \lbrace8,9,10,11\rbrace $ & $ 0 $, $ 0 $, $ 0 $, $ 0 $\\ 	
			\hline
			$ T_6 $ & $ \phi $ & $ \lbrace7,8,9,10,11,12\rbrace $ & $ - $\\ 
			\hline
		\end{tabular}\\
		\caption{Table that verifies \textit{c}$ 2 $ for $ \mathcal{G}_6 $.}
		\label{table12}
	\end{table}
	\begin{v1} From Table \ref{table11} and Table \ref{table12}, it is observed that \textit{c}$ 1 $ and \textit{c}$ 2 $ are satisfied by $ \mathcal{G}_6 $. As a result, \textit{Algorithm} $ 1 $ can be used to decode an index code obtained by using \textit{Construction} $ 1$ on the IC structure $ \mathcal{G}_6 $.\end{v1}
	The index code obtained is 
$W_I = x_1 \oplus x_2 \oplus x_3 \oplus x_4 \oplus x_5 \oplus x_6; ~~W_7 = x_7 \oplus x_1 \oplus x_{12}; ~~ W_8 = x_8 \oplus x_3 \oplus x_9; ~~ W_9 = x_9 \oplus x_4 \oplus x_8; ~~ W_{10} = x_{10} \oplus x_5 \oplus x_{11}; W_{11} = x_{11} \oplus x_6 \oplus x_{10}; W_{12} = x_{12} \oplus x_6 \oplus x_7.$
Messages $x_7$, $x_8$, $x_9$, $x_{10}$, $ x_{11} $ and $ x_{12} $ are decoded directly using $W_7$, $W_8$, $W_9$, $W_{10}$, $ W_{11} $ and $ W_{12} $ respectively. The computation of $ Z_i $, for $ i=1,2,\dots,6 $ using \textit{Algorithm} $ 1 $ is shown in Table \ref{table13} and the decoding of messages $ x_1 $, $ x_2 $, $ x_3 $, $ x_4 $, $ x_5 $ and $ x_6 $ is shown in Table \ref{table14}.
	\begin{table}[h!]
		\hspace{1cm}
		\begin{tabular}{|c|c|}
			\hline
			Message $ x_i $ &  Computation of $ Z_i$\\
			\hline\hline
			$x_1$ & $ W_I \oplus W_8 \oplus W_9 \oplus W_{10}\oplus W_{11}$\\
			\hline
			$x_2$ & $ W_I \oplus W_7 \oplus W_8 \oplus W_9 \oplus W_{12} $\\
			\hline
			$ x_3 $ & $ W_I \oplus W_{10} \oplus W_{11} $\\
			\hline
			$ x_4 $ & $W_I \oplus W_7 \oplus W_{12}  $\\
			\hline
			$ x_5 $ & $W_I \oplus W_7 \oplus W_{12}  $\\
			\hline
			$ x_6 $	& $ W_I $\\
			\hline	
		\end{tabular}\\
		\caption{Table that shows the working of algorithm $ 1 $ on index code obtained from construction $ 1 $ on $ \mathcal{G}_6 $.}
		\label{table13}
	\end{table}
	\begin{table}[h!]
		\begin{tabular}{|c|c|c|}
			\hline
			Message $ x_i $ &  $ Z_i$ & $ N^+_{\mathcal{G}_6}(i) $\\
			\hline\hline
			$x_1$ & $x_1\oplus x_2 \oplus x_8 \oplus x_{10} $ & $ x_2 $, $x_8 $, $ x_{10} $\\
			\hline
			$x_2$ & $ x_2 \oplus x_5 \oplus x_7 \oplus x_9$ & $ x_5 $, $ x_7 $, $ x_9 $\\
			\hline
			$ x_3 $ & $ x_3 \oplus x_1 \oplus x_2 \oplus x_4\oplus x_{11} $ &$ x_1 $, $ x_2 $, $ x_4 $, $ x_{11} $\\
			\hline
			$ x_4 $ & $ x_4 \oplus x_2 \oplus x_3 \oplus x_5\oplus x_{12} $ &$ x_2 $, $ x_3 $, $ x_5 $, $ x_{12} $\\
			\hline
			$ x_5 $ & $ x_5 \oplus x_2 \oplus x_3 \oplus x_4\oplus x_{12} $ &$ x_2 $, $ x_3 $, $ x_4 $, $ x_{12} $\\
			\hline
			$ x_6 $	& $ x_1 \oplus x_2 \oplus x_3 \oplus x_4 \oplus x_5 \oplus x_6$ & $ x_1 $, $ x_2 $, $ x_3 $, $ x_4 $, $ x_5 $\\
			\hline
		\end{tabular}\\
		\caption{Table showing the decoding of messages using algorithm $ 1 $ on index code obtained from construction $ 1 $ on $ \mathcal{G}_6$.}
		\label{table14}
		
	\end{table}Thus, \textit{Algorithm} $ 1 $ is used to decode the index code obtained by using \textit{Construction}$ 1 $ on $ \mathcal{G}_6 $.
	
\end{ex}
\begin{ex}
\label{exam7}
	Consider $\mathcal{G}_7$, a side-information graph which is a $5$-IC structure, shown in Fig. \ref{f8}.
	\begin{figure}[h!]

		\includegraphics[height=\columnwidth,width=\columnwidth,angle=0]{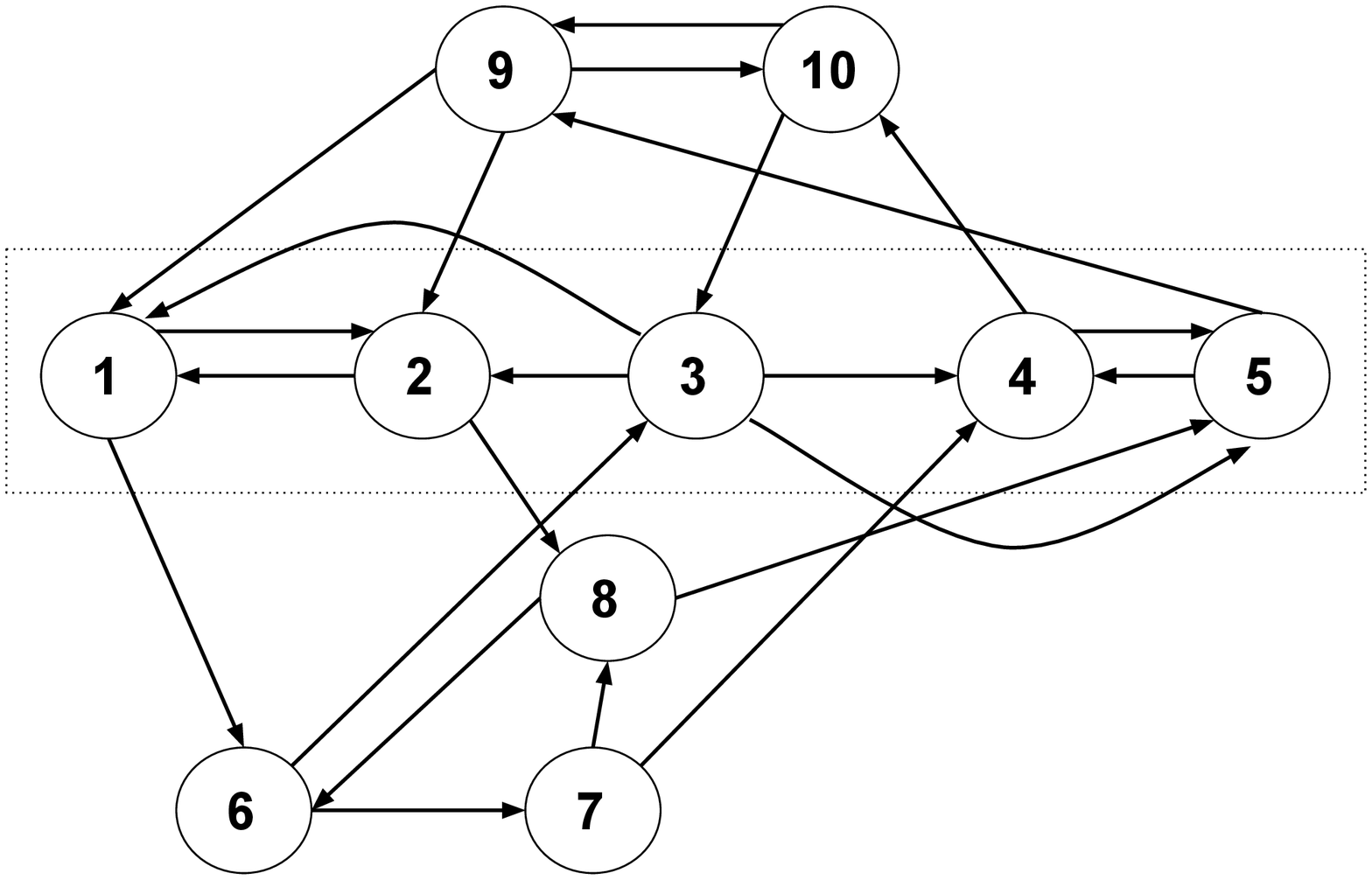}
		\caption{$5$-IC structure $\mathcal{G}_7$ with $V_I=\lbrace1,2,3,4,5\rbrace$.}
				\label{f8}
	\end{figure}$\mathcal{G}_7$ is a $ 5$-IC structure with inner vertex set $V_I=\lbrace1,2,3,4,5\rbrace$ since 
	\begin{enumerate}
		\item there are no cycles with only one vertex from the set$\lbrace1,2,3,4,5\rbrace $ in $\mathcal{G}_7$ (i.e., no I-cycles),
		\item using the rooted trees for each vertex in the set, $\lbrace1,2,3,4,5\rbrace$, which are given in Fig. \ref{rt81}, \ref{rt82}, \ref{rt83}, \ref{rt84} and \ref{rt85} respectively, it is verified that there exists a unique path between any two different vertices in $ V_I $ in $ \mathcal{G}_7 $ and does not contain any other vertex in $ V_I $ (i.e., unique I-path between any pair of inner vertices),
		\item $\mathcal{G}_7$ is the union of all the $5$ rooted trees.
	\end{enumerate}
There are two cycles consisting of only the non-inner vertices.
\begin{figure*}[!t]
		\centering
		\begin{subfigure}{.31\textwidth}
			\centering
			\includegraphics[width=15pc]{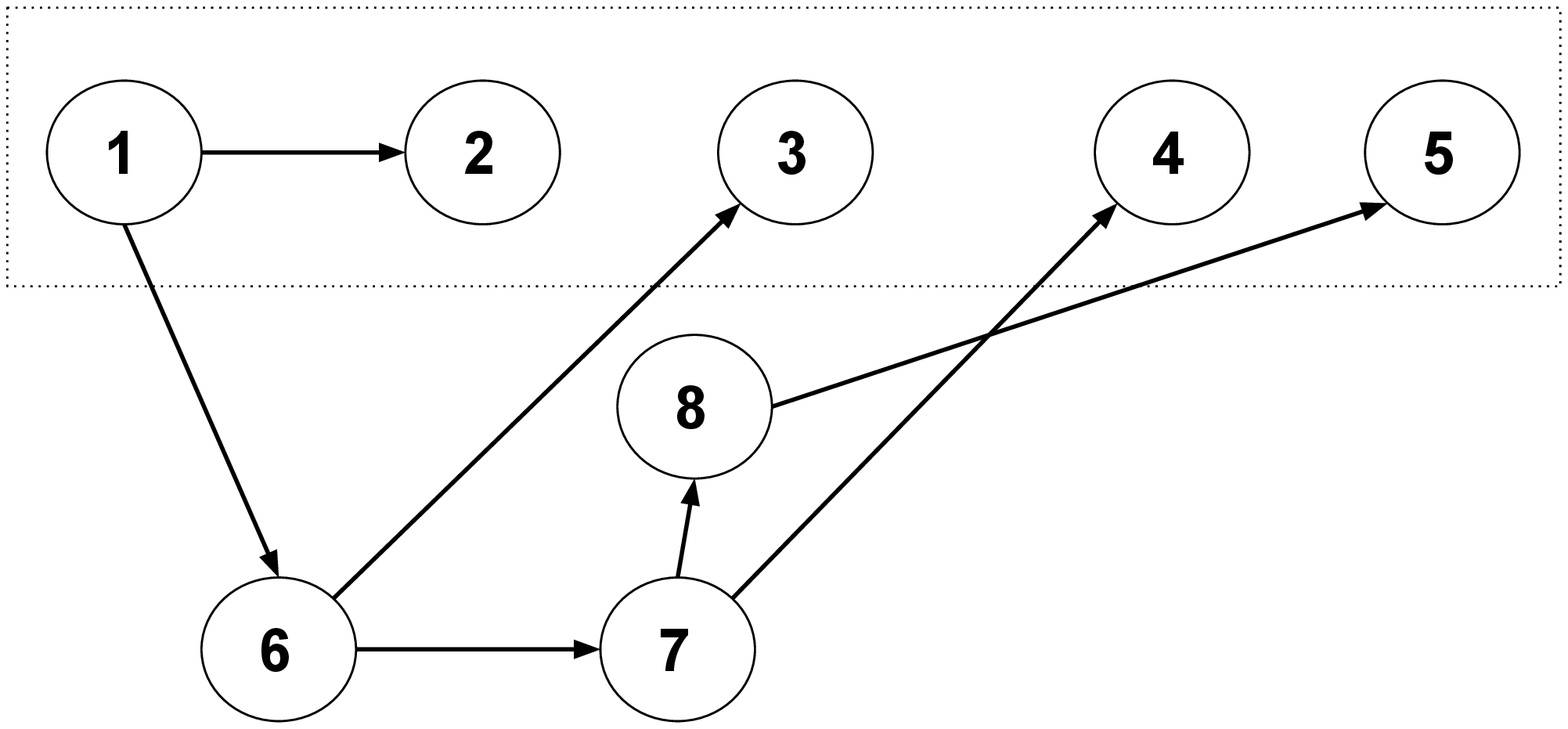}
			\caption{}
			\label{rt81}
		\end{subfigure}%
		\begin{subfigure}{.31\textwidth}
			\centering
			\includegraphics[width=15pc]{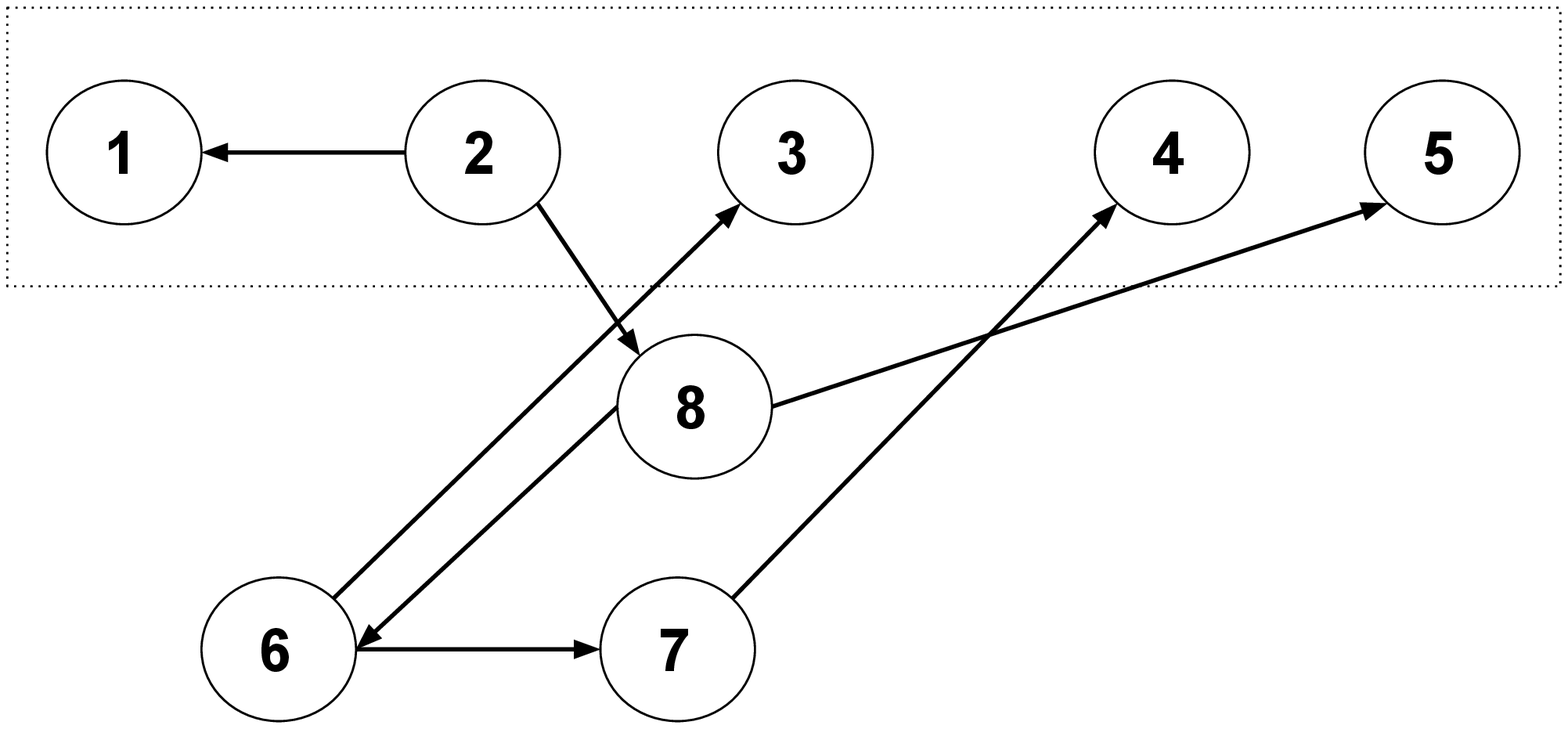}
			\caption{}
			\label{rt82}
		\end{subfigure}
		\begin{subfigure}{.31\textwidth}
			\centering
			\includegraphics[width=15pc]{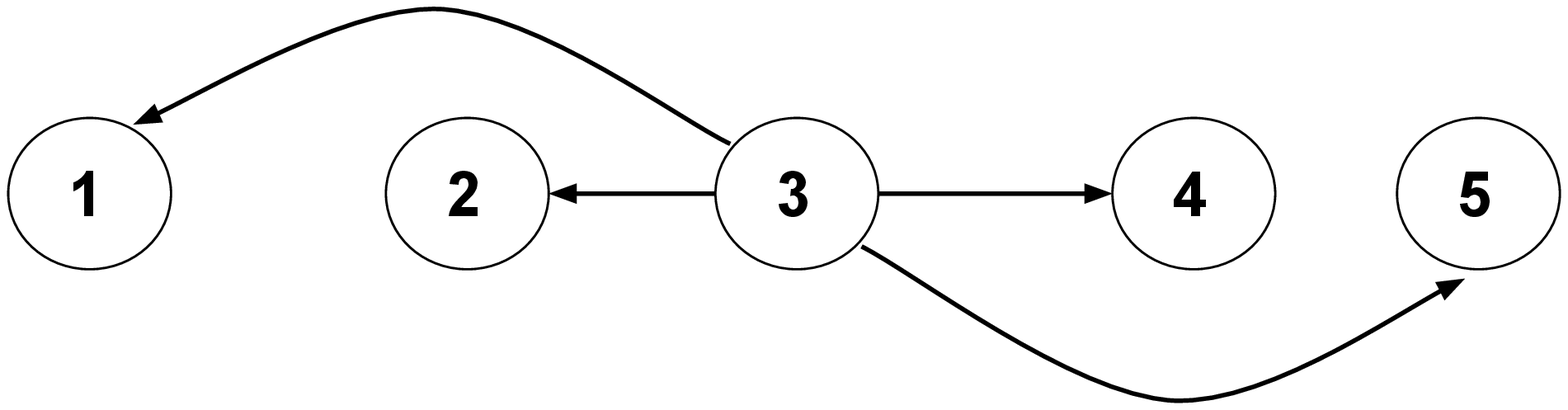}
			\caption{}
			\label{rt83}
		\end{subfigure}
		\centering
		\begin{subfigure}{.31\textwidth}
			\centering
			\includegraphics[width=15pc]{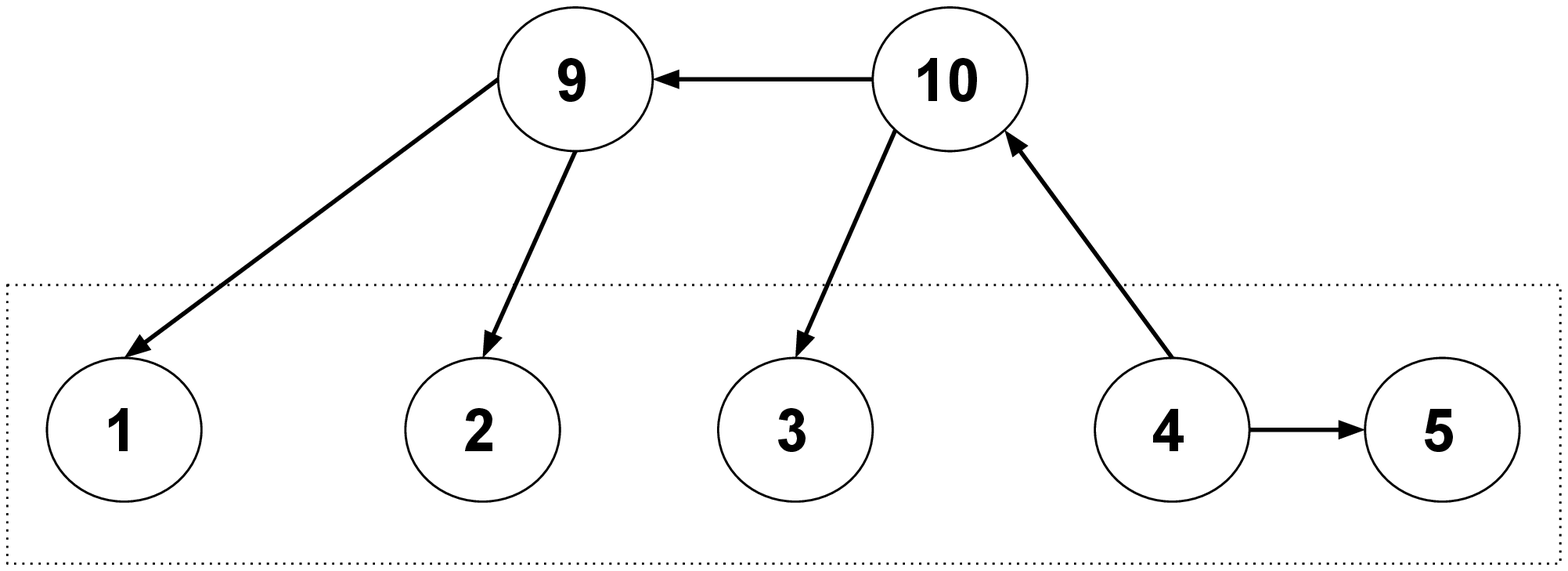}
			\caption{}
			\label{rt84}
		\end{subfigure}%
		\begin{subfigure}{.31\textwidth}
			\centering
			\includegraphics[width=15pc]{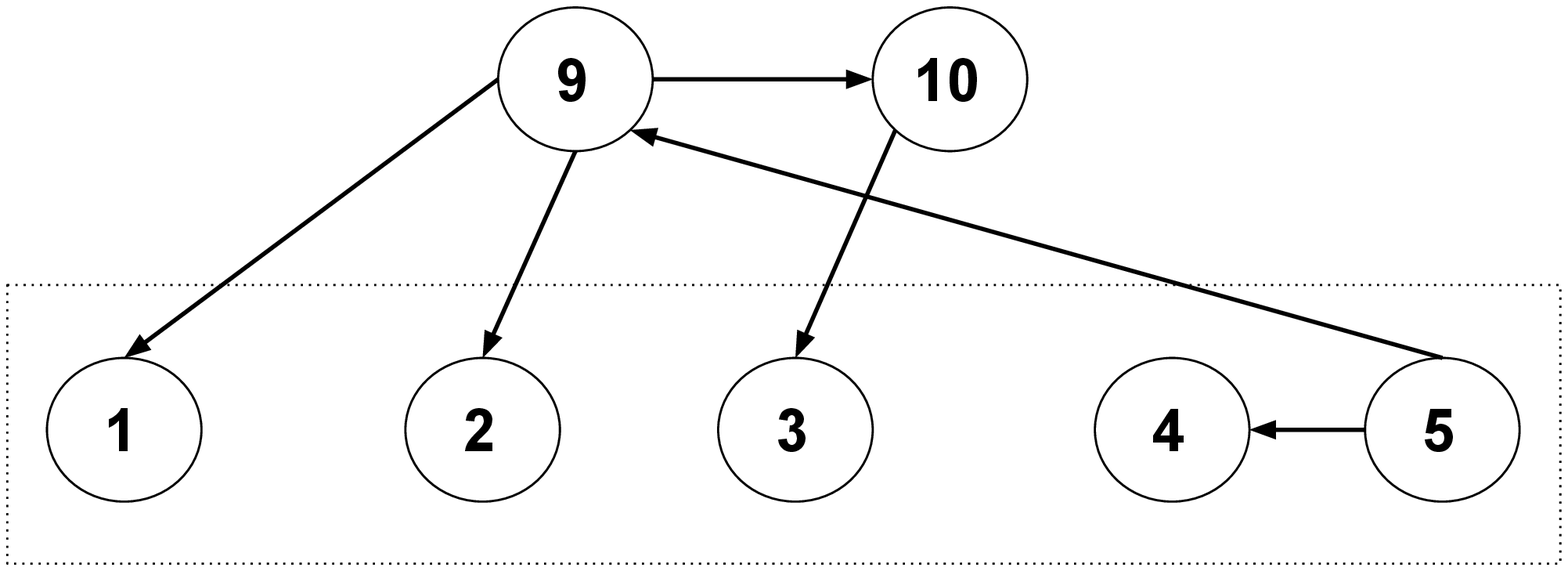}
			\caption{}
			\label{rt85}
		\end{subfigure}
		\caption{Figures showing rooted trees of inner vertices $ 1,2,3,4,5 $ of $ \mathcal{G}_7 $, respectively.}
	\end{figure*}
	\begin{table}
\centering
		\hspace{2mm}
		\begin{tabular}{|c|c|c|c|}
			\hline
			$ T_i $ & $ V_{NI}(i) $ & $ j \in V_{NI}(T_i)\backslash N^+_{T_i}(i) $ & $ a_{i,j} $ \\
			\hline\hline
			$ T_1 $ & $ \lbrace6,7,8\rbrace $ & $ \lbrace7,8\rbrace $ & $ 1 $, $ 1 $ \\ 	
			\hline
			$ T_2 $ & $ \lbrace6,7,8\rbrace $ & $ \lbrace6,7\rbrace $ & $ 1 $, $ 1 $ \\
			\hline
			$ T_3 $ & $ \lbrace \phi \rbrace $ & $ \lbrace\phi\rbrace $ & $ - $ \\ 	 
			\hline
			$ T_4 $ & $ \lbrace9,10\rbrace $ & $ \lbrace9\rbrace $ & $ 1 $ \\ 		
			\hline
			$ T_5 $ & $ \lbrace9,10\rbrace $ & $ \lbrace10\rbrace $ & $ 1 $ \\
			\hline
		\end{tabular}\\
		\caption{Table that verifies \textit{c}$ 1 $ for $ \mathcal{G}_7 $.}
		\label{table15}
	\end{table}
\begin{table}[h!]
\centering
		\begin{tabular}{|c|c|c|c|}
			\hline
			$ T_i $ & $ V_{NI}(i) $ & $ j \in V(\mathcal{G}_7)\backslash V(T_i) $ & $ b_{i,j} $ \\
			\hline\hline
			$ T_1 $ & $ \lbrace6,7,8\rbrace $ & $ \lbrace9,10\rbrace $ & $0$, $ 0 $\\ 	
			\hline
			$ T_2 $ & $ \lbrace6,7,8\rbrace $ & $ \lbrace9,10\rbrace $ & $ 0 $, $ 0 $ \\ 	
			\hline
			$ T_3$ & $ \lbrace\phi\rbrace $ & $ \lbrace6,7,8,9,10\rbrace $ & $ - $\\ 	
			\hline
			$ T_4 $ & $ \lbrace9,10\rbrace $ & $ \lbrace6,7,8\rbrace $ & $ 0 $, $ 0 $, $ 0 $\\ 	
			\hline
			$ T_5 $ & $ \lbrace9,10\rbrace $ & $ \lbrace6,7,8\rbrace $ & $ 0 $, $ 0 $, $ 0 $\\ 	
			\hline
		\end{tabular}\\
		\caption{Table that verifies \textit{c}$ 2 $ for $ \mathcal{G}_7 $.}
		\label{table16}
	\end{table}
	\begin{v1} From Table \ref{table15} and Table \ref{table16}, it is observed that \textit{c}$ 1 $ and \textit{c}$ 2 $ are satisfied by $ \mathcal{G}_7 $. As a result, \textit{Algorithm} $ 1 $ can be used to decode an index code obtained by using \textit{Construction} $ 1$ on the IC structure $ \mathcal{G}_7 $.\end{v1}

	The index code obtained is 
$W_I = x_1 \oplus x_2 \oplus x_3 \oplus x_4 \oplus x_5; ~~W_6 = x_6 \oplus x_3 \oplus x_7; ~~W_7 = x_7 \oplus x_8 \oplus x_4; ~~ W_8 = x_8 \oplus x_5 \oplus x_6; ~~ W_9 = x_9 \oplus x_1 \oplus x_2 \oplus x_{10}; ~~		W_{10} = x_{10} \oplus x_3 \oplus x_9.$
Messages $ x_6 $, $x_7$, $x_8$, $x_9$ and $x_{10}$ are decoded directly using $ W_6 $, $W_7$, $W_8$, $W_9$ and $W_{10}$ respectively. The computation of $ Z_i $, for $ i=1,2,\dots,5 $ using \textit{Algorithm} $ 1 $ is shown in Table \ref{table17} and the decoding of messages $ x_1 $, $ x_2 $, $ x_3 $, $ x_4 $ and $ x_5 $ is shown in Table \ref{table18}.
\begin{table}[h!]
\centering
		\hspace{1.3cm}
		\begin{tabular}{|c|c|}
			\hline
			Message $ x_i $ &  Computation of $ Z_i$\\
			\hline\hline
			$x_1$ & $ W_I \oplus W_6 \oplus W_7 \oplus W_8$\\
			\hline
			$x_2$ & $ W_I \oplus W_6 \oplus W_7 \oplus W_8$\\
			\hline
			$ x_3 $ & $ W_I$ \\
			\hline
			$ x_4 $ & $W_I \oplus W_9 \oplus W_{10}  $\\
			\hline
			$ x_5 $ & $W_I \oplus W_9 \oplus W_{10}  $\\
			\hline
		\end{tabular}\\
		\caption{Table that shows the working of algorithm $ 1 $ on index code obtained from construction $ 1 $ on $ \mathcal{G}_7 $.}
		\label{table17}
	\end{table}
	\begin{table}[h!]
\centering
		\begin{tabular}{|c|c|c|}
			\hline
			Message $ x_i $ &  $ Z_i$ & $ N^+_{\mathcal{G}_7}(i) $\\
			\hline\hline
			$x_1$ & $x_1\oplus x_2  $ & $ x_2 $, $x_6 $\\
			\hline
			$x_2$ & $ x_2 \oplus x_1 $ & $ x_1 $, $ x_8 $\\
			\hline
			$ x_3 $ & $ x_3 \oplus x_1 \oplus x_2 \oplus x_4\oplus x_5 $ &$ x_1 $, $ x_2 $, $ x_4 $, $ x_5 $\\
			\hline
			$ x_4 $ & $ x_4 \oplus x_5 $ &$ x_5 $, $ x_{10} $\\
			\hline
			$ x_5 $ & $ x_5 \oplus x_4  $ &$ x_4 $, $ x_{10} $\\
			\hline
		\end{tabular}\\
		\caption{Table that shows the decoding of messages using algorithm $ 1 $ on index code obtained from construction $ 1 $ on $ \mathcal{G}_7$.}
		\label{table18}
		
	\end{table}Thus, \textit{Algorithm} $ 1 $ is used to decode the index code obtained by using \textit{Construction}$ 1 $ on $ \mathcal{G}_7$.  
\end{ex}
In the following last example it is shown that for some IC structures the code constructed using \textit{Construction 1} is not decodable using any algorithm using only linear combinations of the index code symbols.

\begin{ex}
	\label{exam8}
	Consider $\mathcal{G}_8$, a side-information graph which is a $5$-IC structure, shown in Fig. \ref{f6}.
	\begin{figure}[h!]
		
		\includegraphics[height=\columnwidth,width=\columnwidth,angle=0]{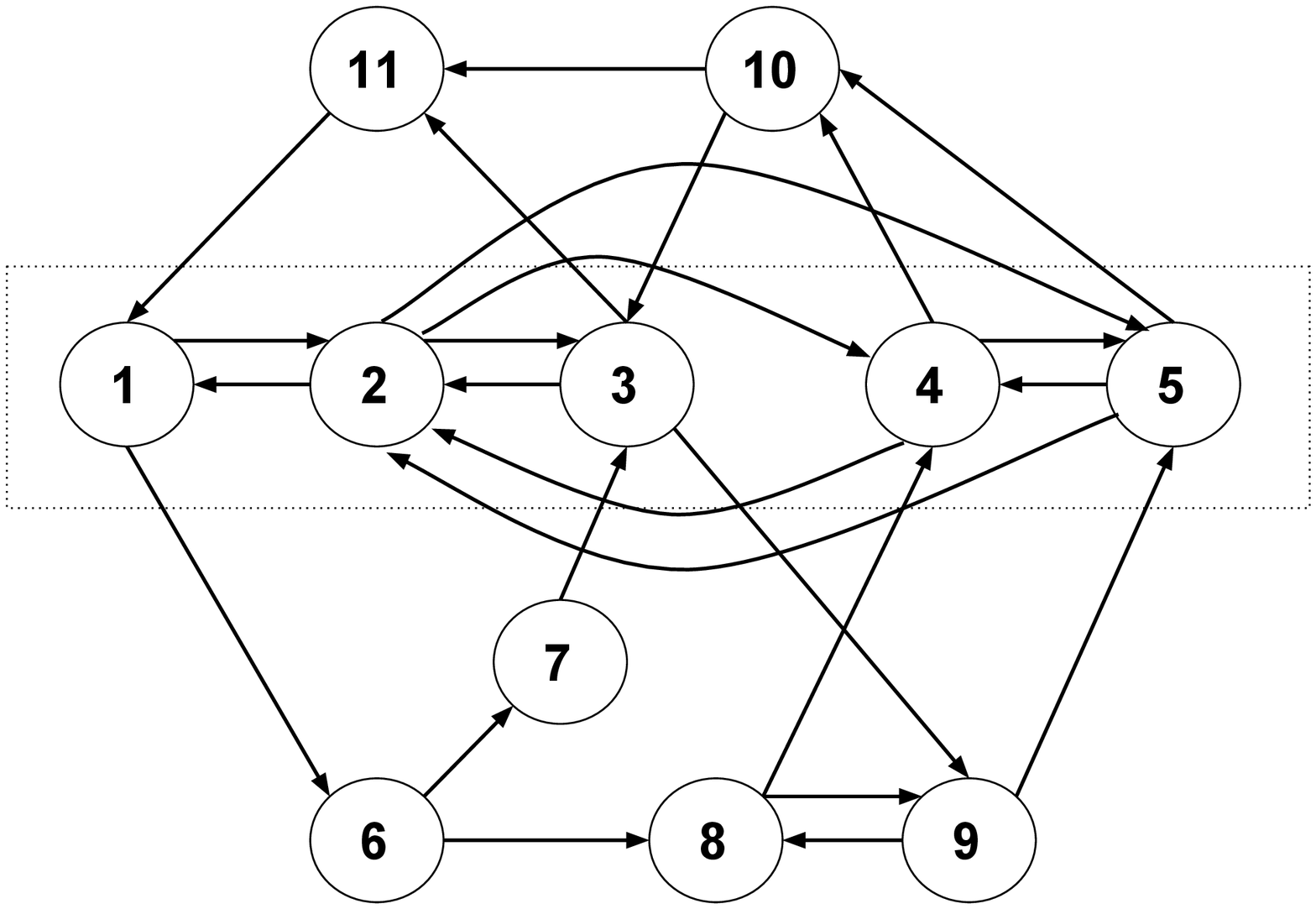}
		\caption{$5$-IC structure $\mathcal{G}_8$ with $V_I=\lbrace1,2,3,4,5\rbrace$.}
		\label{f6}
	\end{figure}$
\mathcal{G}_8$ is a $ 6 $-IC structure with inner vertex set $ V_I=\lbrace1,2,3,4,5\rbrace$ since 
	\begin{enumerate}
		\item there are no cycles with only one vertex from the set $\lbrace1,2,3,4,5\rbrace $ in $\mathcal{G}_8$ (i.e., no I-cycles),
		\item using the rooted trees for each vertex in the set, $\lbrace1,2,3,4,5\rbrace$, which are given in Fig. \ref{rt61}, \ref{rt62}, \ref{rt63}, \ref{rt64} and \ref{rt65} respectively, it is verified that there exists a unique path between any two different vertices in $ V_I $ in $ \mathcal{G}_8 $ and does not contain any other vertex in $ V_I $ (i.e., unique I-path between any pair of inner vertices),
		\item $\mathcal{G}_8$ is the union of all the $5$ rooted trees.
	\end{enumerate}
	\begin{figure*}[!t]
		\centering
		\begin{subfigure}{.31\textwidth}
			\hspace{-6mm}
			\includegraphics[width=15pc]{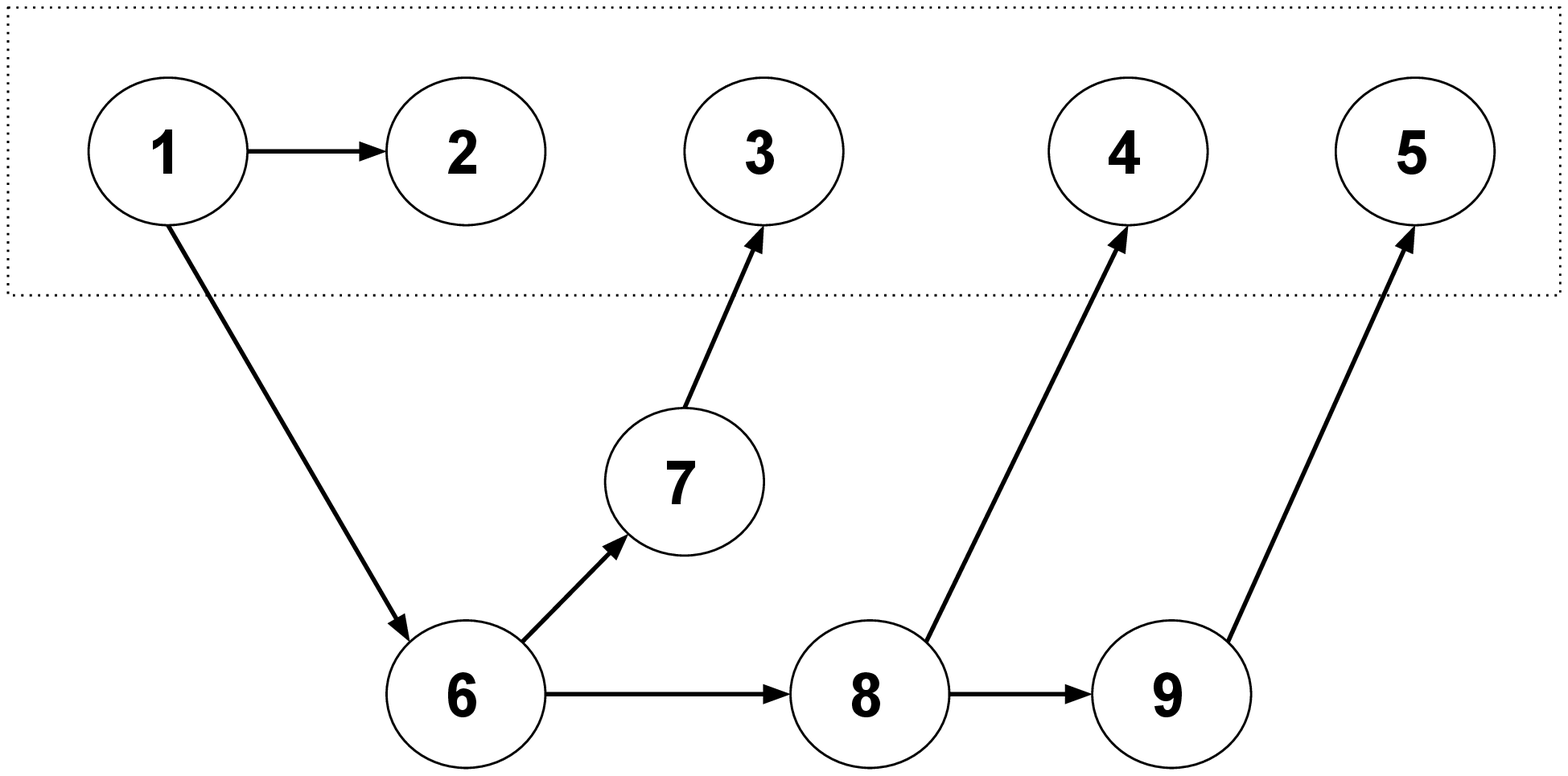}
			\caption{}
			\label{rt61}
		\end{subfigure}%
		\begin{subfigure}{.31\textwidth}
			\hspace{-6mm}
			\includegraphics[width=15pc]{RT_2_2}
			\caption{}
			\label{rt62}
		\end{subfigure}
		\begin{subfigure}{.31\textwidth}
			\hspace{-6mm}
			\includegraphics[width=15pc]{RT_2_3}
			\caption{}
			\label{rt63}
		\end{subfigure}
		\centering
		\begin{subfigure}{.31\textwidth}
			\hspace{-6mm}
			\includegraphics[width=15pc]{RT_2_4}
			\caption{}
			\label{rt64}
		\end{subfigure}%
		\begin{subfigure}{.31\textwidth}
			\hspace{-6mm}
			\includegraphics[width=15pc]{RT_2_5}
			\caption{}
			\label{rt65}
		\end{subfigure}
		\caption{Figures showing rooted trees of inner vertices $ 1,2,3,4,5$ of $ \mathcal{G}_8 $, respectively.}
	\end{figure*}
	\begin{table}[h!] \centering
		\begin{tabular}{|c|c|c|c|}
			\hline
			$ T_i $ & $ V_{NI}(i) $ & $ j \in V_{NI}(T_i)\backslash N^+_{T_i}(i) $ & $ a_{i,j} $ \\
			\hline\hline
			$ T_1 $ & $ \lbrace6,7,8,9\rbrace $ & $ \lbrace7,8,9\rbrace $ & $ 1 $, $ 2 $, $ 1 $ \\	
			\hline
		\end{tabular}\\
		\caption{Table that verifies \textit{c}$ 1 $ for $ \mathcal{G}_8 $.}
		\label{table17a}
	\end{table}
	\begin{table}[h!]
\centering
	\begin{tabular}{|c|c|c|c|}
			\hline
			$ T_i $ & $ V_{NI}(i) $ & $ j \in V(\mathcal{G}_8)\backslash V(T_i) $ & $ b_{i,j} $ \\
			\hline\hline
			$ T_1 $ & $ \lbrace6,7,8,9\rbrace $ & $ \lbrace10,11\rbrace $ & $0$, $ 0 $\\ 	
			\hline
			$ T_2 $ & $ \lbrace10,11\rbrace $ & $ \lbrace6,7,8,9\rbrace $ & $ 0 $, $ 0 $, $0$, $ 0 $\\ 	
			\hline
			$ T_3$ & $ \lbrace8,9,11\rbrace $ & $ \lbrace6,7,10\rbrace $ & $ 0 $, $ 0 $, $ 0 $\\ 	
			\hline
			$ T_4 $ & $ \lbrace10,11\rbrace $ & $ \lbrace6,7,8,9\rbrace $ & $ 0 $, $ 0 $, $ 0 $, $ 0 $\\ 	
			\hline
			$ T_5 $ & $ \lbrace10,11\rbrace $ & $ \lbrace6,7,8,9\rbrace $ & $ 0 $, $0$, $ 0 $, $ 0 $\\ 	
			\hline
		\end{tabular}\\
		\caption{Table that verifies \textit{c}$ 2 $ for $ \mathcal{G}_8 $.}
		\label{table8}
	\end{table}
	\begin{v1} From Table \ref{table17a} and Table \ref{table8}, it is observed that \textit{c}$ 1 $ is not satisfied ($ a_{1,8}=2 $, an even number) and \textit{c}$ 2 $ is satisfied by $ \mathcal{G}_8 $. As a result, \textit{Algorithm} $ 1 $ fails to decode an index code obtained by using \textit{Construction} $ 1$ on the IC structure $ \mathcal{G}_8 $. It is verified as follows.
	\end{v1}

	The index code obtained is 
$W_I = x_1 \oplus x_2 \oplus x_3 \oplus x_4 \oplus x_5: ~~ W_6 = x_6 \oplus x_7 \oplus x_8; ~~W_7 = x_7 \oplus x_3; W_8 = x_8 \oplus x_4 \oplus x_9; ~~  W_9 = x_9 \oplus x_5 \oplus x_8; ~~ W_{10} = x_{10} \oplus x_3 \oplus x_{11}; ~~ W_{11} = x_{11} \oplus x_1.$
	Messages $ x_6 $, $x_7$, $x_8$, $x_9$, $x_{10}$ and $ x_{11} $ are decoded directly using $ W_6 $, $W_7$, $W_8$, $W_9$, $W_{10}$ and $ W_{11} $ respectively. The computation of $ Z_1 $, for $ i=1,2,\dots,5 $ using \textit{Algorithm} $ 1 $ is shown in Table \ref{table9} and Table \ref{table10} illustrates the inability of the \textit{Algorithm} $ 1 $ to decode $ x_1 $.
\begin{table}[h!]
\centering
		\hspace{1cm}
		\begin{tabular}{|c|c|}
			\hline
			Message $ x_i $ &  Computation of $ Z_i$\\
			\hline\hline
			$x_1$ & $ W_I \oplus W_6 \oplus W_7 \oplus W_8 \oplus W_9$\\
			\hline
			$x_2$ & $ W_I $\\
			\hline
			$ x_3 $ & $ W_I \oplus W_8 \oplus W_9 \oplus W_{11} $\\
			\hline
			$ x_4 $ &$ W_I \oplus W_{10} \oplus W_{11} $\\
			\hline
			$ x_5 $ & $ W_I \oplus W_{10} \oplus W_{11} $\\
			\hline
		\end{tabular}\\
		\caption{Table that shows the working of algorithm $ 1 $ on index code obtained from construction $ 1 $ on $ \mathcal{G}_8 $.}
		\label{table9}
	\end{table}
\begin{table}[h!]
\centering
		\hspace{5mm}
		\begin{tabular}{|c|c|c|}
			\hline
			Message $ x_i $ &  $ Z_i$ & $ N^+_{\mathcal{G}_8}(i) $\\
			\hline\hline
			$x_1$ & $x_1\oplus x_2 \oplus x_6\oplus x_8$ & $ x_2 $, $x_6 $\\
			\hline
		\end{tabular}\\
		\caption{Table that shows failure of algorithm $ 1 $ on index code obtained from construction $ 1 $ on $ \mathcal{G}_8$.}
		\label{table10}

	\end{table}
Thus, \textit{Algorithm} $ 1 $ is fails to decode the index code obtained using \textit{Construction}$ 1 $ on $ \mathcal{G}_8 $ since user requesting message $ x_1 $ does not have $ x_8 $ in its side-information.

It turns out that $ x_1 $ cannot be decoded using any linear combination of the index code symbols. All the possible linear combinations of the index code symbols are listed in the Table \ref{tablef} along with the reason for $ x_1 $ not being decodable using that linear combination.
\begin{table*}
\centering
\begin{tabular}{|c|c|c|c|}
\hline
S.no & Linear combination & Obtained sum & Reason\\
\hline\hline
$ 1 $ & $ 0 $ & $ 0 $ & $ - $ \\
\hline
$ 2 $ & $ W_I $ & $ x_1\oplus x_2 \oplus x_3 \oplus x_4 \oplus x_5 $ & $ x_3 $ not in side-information\\
\hline
$ 3 $ & $ W_6 $ & $ x_6\oplus x_7\oplus x_8  $ & $ x_1 $ is absent\\
\hline
$ 4 $ & $ W_I \oplus W_6 $ & $ x_1 \oplus x_2 \oplus x_3 \oplus x_4 \oplus x_5 \oplus x_6 \oplus x_7\oplus x_8  $ & $ x_3 $ not in side-information\\
\hline
$ 5 $ & $ W_7  $ & $x_3\oplus x_7  $ & $ x_1 $ is absent\\
\hline
$ 6 $ & $ W_I \oplus W_7 $ & $ x_1 \oplus x_2 \oplus x_4 \oplus x_5 \oplus x_7 $ & $ x_4 $ not in side-information\\
\hline
$ 7 $ & $ W_6 \oplus W_7  $ & $ x_3\oplus x_6 \oplus x_8 $ & $ x_1 $ is absent\\
\hline
$ 8$ & $ W_I\oplus W_6\oplus W_7 $ & $ x_1 \oplus x_2 \oplus x_4 \oplus x_5 \oplus x_6 \oplus x_8 $ &$ x_4 $ not in side-information\\
\hline
$ 9 $ & $ W_8 $ & $ x_4\oplus x_8 \oplus x_9 $ & $ x_1 $ is absent \\
\hline
$ 10 $ & $ W_I \oplus W_8 $ & $ x_1 \oplus x_2 \oplus x_3 \oplus x_5 \oplus x_8\oplus x_9 $ & $ x_3 $ not in side-information \\
\hline
$ 11 $ & $ W_6\oplus W_8 $ & $ x_4 \oplus x_6 \oplus x_7 \oplus x_9 $ & $ x_1 $ is absent\\
\hline
$ 12 $ & $ W_I \oplus W_6\oplus W_8 $ & $ x_1 \oplus x_2 \oplus x_3 \oplus x_5 \oplus x_6 \oplus x_7\oplus x_9 $ & $ x_3 $ not in side-information \\
\hline
$ 13 $ &  $ W_7 \oplus W_8 $ & $ x_3\oplus x_4 \oplus x_7\oplus x_8\oplus x_9  $ & $ x_1 $ is absent \\
\hline
$ 14 $ &  $ W_I \oplus W_7 \oplus W_8 $ & $ x_1\oplus x_2 \oplus x_5\oplus x_7\oplus x_8\oplus x_9  $ & $ x_5 $ is not in side-information\\
\hline
$ 15 $ &  $ W_6 \oplus W_7 \oplus W_8 $ & $ x_3\oplus x_4 \oplus x_6\oplus x_9 $ & $ x_1 $ is absent \\
\hline
$ 16 $ &  $ W_I \oplus W_6 \oplus W_7 \oplus W_8 $ & $ x_1 \oplus x_2 \oplus x_5 \oplus x_6 \oplus x_9 $ & $ x_5 $ not in side-information \\
\hline
$ 17 $ &  $ W_9 $ & $ x_5\oplus x_{8}\oplus x_{9} $ & $ x_1 $ is absent \\
\hline
$ 18 $ &  $ W_I \oplus W_9 $ & $ x_1 \oplus x_2 \oplus x_3 \oplus x_4 \oplus x_8\oplus x_{9} $ & $x_3$ not in side-information \\
\hline
$ 19 $ &  $ W_6\oplus W_9 $ & $x_5 \oplus x_6 \oplus x_7 \oplus x_{9} $ & $ x_1 $ absent \\
\hline 
$ 20 $ &  $ W_I \oplus W_6\oplus W_9 $ & $ x_1 \oplus x_2 \oplus x_3 \oplus x_4 \oplus x_6 \oplus x_7 \oplus x_9 $ & $ x_3 $ not in side-information\\
\hline
$ 21 $ &  $ W_7 \oplus W_9 $ & $ x_3 \oplus x_5 \oplus x_7 \oplus x_8 \oplus x_{9} $ & $ x_1 $ is absent\\
\hline
$ 22 $ &  $ W_I \oplus W_7 \oplus W_9 $ & $ x_1 \oplus x_2 \oplus x_4 \oplus x_7 \oplus x_8 \oplus x_9 $ & $ x_4 $ not in side-information\\
\hline
$ 23 $ &  $ W_6\oplus W_7\oplus W_9 $ & $  x_3 \oplus x_5 \oplus x_6 \oplus x_9  $ & $ x_1 $ is absent \\
\hline
$ 24 $ &  $ W_I \oplus W_6\oplus W_7\oplus W_9 $ & $ x_1 \oplus x_2 \oplus x_4 \oplus x_6 \oplus x_9 $ & $x_4 $ not in side-information\\
\hline
$ 25 $ &  $ W_8 \oplus W_9 $ & $x_4 \oplus x_5 $ & $ x_1 $ is absent \\
\hline
$26  $ &  $ W_I \oplus W_8 \oplus W_9 $ & $ x_1 \oplus x_2 \oplus x_3 $ & $ x_3 $ not in side-information \\
\hline
$ 27 $ &  $ W_6 \oplus W_8 \oplus W_9 $ & $x_4\oplus x_5 \oplus x_6 \oplus x_7\oplus x_8 $ & $ x_1 $ is absent \\
\hline
$ 28 $ &  $ W_I \oplus W_6 \oplus W_8 \oplus W_9  $ & $ x_1 \oplus x_2 \oplus x_3 \oplus x_6 \oplus x_7 \oplus x_8 $ &$ x_3 $ not in side-information \\
\hline
$ 29 $ &  $ W_7\oplus W_8\oplus W_9 $ & $ x_3\oplus x_4\oplus x_5\oplus x_7 $ & $ x_1 $ is absent\\
\hline
$ 30 $ &  $ W_I\oplus W_7\oplus W_8\oplus W_9  $ & $ x_1\oplus x_2\oplus x_7 $ & $ x_7 $ not in side-information \\
\hline
$ 31 $ &  $ W_6\oplus W_7\oplus W_8\oplus W_9 $ & $  x_3\oplus x_4 \oplus x_5 \oplus x_6\oplus x_8 $ & $ x_1 $ is absent\\
\hline
$ 32 $ &  $ W_I\oplus W_6\oplus W_7\oplus W_8\oplus W_9 $ & $ x_1 \oplus x_2 \oplus x_6 \oplus x_8 $ & $ x_8 $ not in side-information\\
\hline
$ 33 $ &  $ W_{10} $ & $ x_3\oplus x_{10}\oplus x_{11} $ & $ x_1 $ is absent\\
\hline
$ 34 $ &  $ W_I\oplus W_{10} $ & $ x_1 \oplus x_2 \oplus x_4 \oplus x_5\oplus x_{11} $ & $ x_4 $ not in side information\\
\hline
$ 35 $ &  $ W_6\oplus W_{10} $ & $ x_3\oplus x_6\oplus x_7\oplus x_{8}\oplus x_{10}\oplus x_{11} $ & $ x_1 $ is absent\\
\hline
$ 36 $ &  $ W_I\oplus W_6\oplus W_{10} $ & $ x_1 \oplus x_2 \oplus x_4 \oplus x_5 \oplus x_6\oplus x_7 \oplus x_8 \oplus x_{10} \oplus x_{11} $ & $ x_4 $ not in side-information\\
\hline
$ 37 $ &  $ W_7\oplus W_{10} $ & $ x_7\oplus x_{10}\oplus x_{11} $ & $ x_1 $ is absent\\
\hline
$ 38 $ &  $ W_I \oplus W_7\oplus W_{10} $ & $ x_1\oplus x_2 \oplus x_3 \oplus x_4 \oplus x_5 \oplus x_7 \oplus x_{10} \oplus x_{11} $ & $ x_3 $ not in side-information\\
\hline
$ 39 $ &  $ W_6\oplus W_7\oplus W_{10} $ & $ x_6 \oplus x_8 \oplus x_{10} \oplus x_{11} $ & $ x_1 $ is absent \\
\hline
$ 40 $ &  $ W_I\oplus W_6\oplus W_7\oplus W_{10} $ & $x_1\oplus x_2 \oplus x_3 \oplus x_4 \oplus x_5 \oplus x_6 \oplus x_8 \oplus x_{10} \oplus x_{11} $ & $ x_3 $ not in side-information\\
\hline
$ 41 $ &  $ W_8\oplus W_{10} $ & $ x_3\oplus x_4\oplus x_8\oplus x_9\oplus x_{10}\oplus x_{11} $ & $ x_1 $ is absent\\
\hline
$ 42 $ &  $ W_I\oplus W_8\oplus W_{10} $ & $ x_1 \oplus x_2 \oplus x_5 \oplus x_8\oplus x_9 \oplus x_{10} \oplus x_{11} $ & $ x_5 $ not in -side-information\\
\hline
$ 43 $ &  $ W_6\oplus W_8\oplus W_{10} $ & $ x_3 \oplus x_4 \oplus x_6 \oplus x_7 \oplus x_{9} \oplus x_{10} \oplus x_{11} $ & $ x_1 $ is absent\\
\hline
$ 44 $ &  $ W_I\oplus W_6\oplus W_8\oplus W_{10} $ & $x_1 \oplus x_2 \oplus x_5 \oplus x_6 \oplus x_7 \oplus x_9 \oplus x_{10} \oplus x_{11} $ & $ x_5 $ not in side-information\\
\hline
$ 45 $ &  $ W_7\oplus W_8\oplus W_{10} $ & $ x_4\oplus x_7\oplus x_8\oplus x_9\oplus x_{10}\oplus x_{11} $ & $ x_1 $ is absent \\
\hline
$ 46 $ & $ W_I\oplus W_7\oplus W_8\oplus W_{10} $ & $ x_1\oplus x_2\oplus x_3\oplus x_5\oplus x_7\oplus x_8\oplus x_9\oplus x_{10}\oplus x_{11} $ & $ x_3 $ not in side-information\\
\hline
$ 47 $ & $ W_6\oplus W_7\oplus W_8\oplus W_{10} $ & $ x_4 \oplus x_6 \oplus x_9 \oplus x_{10} \oplus x_{11} $ & $ x_1 $ is absent \\
\hline
$ 48 $ & $ W_I\oplus W_6\oplus W_7\oplus W_8\oplus W_{10}  $ & $ x_1\oplus x_2 \oplus x_3 \oplus x_5 \oplus x_6 \oplus x_9 \oplus x_{10} \oplus x_{11} $ & $ x_3 $ not in side-information\\
\hline
$ 49 $ & $ W_9\oplus W_{10} $ & $ x_3\oplus x_5\oplus x_8\oplus x_9\oplus x_{10}\oplus x_{11} $ & $ x_1 $ is absent\\
\hline
$ 50 $ & $ W_I\oplus W_9\oplus W_{10} $ & $ x_1\oplus x_2\oplus x_4\oplus x_8\oplus x_9\oplus x_{10}\oplus x_{11}$ & $ x_4 $ not in side-information \\
\hline
$ 51 $ & $ W_6\oplus W_9\oplus W_{10} $ & $ x_3\oplus x_5\oplus x_6\oplus x_7\oplus x_9\oplus x_{10}\oplus x_{11} $ & $ x_1 $ is absent\\
\hline
$ 52 $ & $ W_I\oplus W_6\oplus W_9\oplus W_{10} $ & $ x_1\oplus x_2\oplus x_4\oplus x_6\oplus x_9\oplus x_{10}\oplus x_{11} $ & $ x_4 $ not in side-information\\
\hline
$ 53 $ & $ W_7 \oplus W_9\oplus W_{10} $ & $ x_5\oplus x_7\oplus x_8\oplus x_9\oplus x_{10}\oplus x_{11} $ & $ x_1 $ is absent \\
\hline
$ 54 $ & $ W_I\oplus W_7 \oplus W_9\oplus W_{10} $ & $ x_1\oplus x_2\oplus x_3\oplus x_4\oplus x_7\oplus x_8\oplus x_9\oplus x_{10}\oplus x_{11} $ & $ x_3 $ not in side-information\\
\hline
$ 55 $ & $ W_6\oplus W_7 \oplus W_9\oplus W_{10} $ & $  x_5\oplus x_6\oplus x_9\oplus x_{10}\oplus x_{11} $ & $ x_1 $ is absent \\
\hline
$ 56 $ & $ W_I\oplus W_6\oplus W_7 \oplus W_9\oplus W_{10} $ & $x_1\oplus x_2\oplus x_3\oplus x_4\oplus x_6\oplus x_9\oplus x_{10}\oplus x_{11} $ & $ x_3 $ not in side-information \\
\hline
$ 57 $ & $ W_8 \oplus W_9\oplus W_{10}  $ & $ x_3\oplus x_4\oplus x_5\oplus x_{10}\oplus x_{11} $ & $ x_1 $ is absent\\
\hline
$ 58 $ & $ W_I \oplus W_8 \oplus W_9\oplus W_{10}  $ & $  x_1\oplus x_2\oplus x_{10}\oplus x_{11} $ & $ x_{10} $ not in side-information \\
\hline
$ 59  $ & $ W_6 \oplus W_8 \oplus W_9\oplus W_{10} $ & $ x_3\oplus x_4\oplus x_5\oplus x_6\oplus x_7\oplus x_8\oplus x_{10}\oplus x_{11} $ &$ x_1 $ is absent \\
\hline
$ 60 $ & $ W_I\oplus W_6 \oplus W_8 \oplus W_9\oplus W_{10}  $ & $ x_1\oplus x_2\oplus x_6\oplus x_7\oplus x_8\oplus x_{10}\oplus x_{11} $ &  $ x_7 $ not in side-information\\
\hline
$ 61 $ & $ W_7 \oplus W_8 \oplus W_9\oplus W_{10} $ & $ x_4\oplus x_5\oplus x_7\oplus x_{10}\oplus x_{11} $ & $ x_1 $ is absent\\
\hline
$ 62 $ & $ W_I\oplus W_7 \oplus W_8 \oplus W_9\oplus W_{10} $ & $ x_1\oplus x_2\oplus x_3\oplus x_7\oplus x_{10}\oplus x_{11} $ & $ x_3 $ not in side-information \\
\hline
$ 63 $ & $ W_6\oplus W_7 \oplus W_8 \oplus W_9\oplus W_{10} $ & $ x_4\oplus x_5\oplus x_6\oplus x_8\oplus x_{10}\oplus x_{11} $ & $ x_1 $ is absent\\
\hline
$ 64 $ & $ W_I\oplus W_6\oplus W_7 \oplus W_8 \oplus W_9\oplus W_{10} $ & $ x_1\oplus x_2\oplus x_3\oplus x_6\oplus x_8\oplus x_{10}\oplus x_{11} $ &  $ x_3 $ not in side-information\\
\hline
$ 65 $ & $ W_{11} $ & $ x_1\oplus x_{11} $ &  $ x_{11} $ not in side-information\\
\hline
$ 66 $ & $ W_I\oplus W_{11} $ & $ x_2\oplus x_3\oplus x_4\oplus x_5\oplus x_{11} $ & $ x_1 $ is absent\\
\hline
$ 67 $ & $  W_6\oplus W_{11} $ & $ x_1\oplus\ x_6\oplus x_7\oplus x_8\oplus x_{11} $ &  $ x_7 $ not in side-information\\
\hline
$ 68 $ & $ W_I\oplus W_6\oplus W_{11} $ & $ x_2\oplus x_3\oplus x_4\oplus x_5\oplus x_6\oplus x_7\oplus x_8\oplus x_{11} $ & $ x_1 $ is absent\\
\hline
$ 69 $ & $ W_7\oplus W_{11} $ & $ x_1\oplus x_3\oplus x_7\oplus x_{11} $ &  $ x_3 $ not in side-information\\
\hline

\end{tabular}
\end{table*}
\begin{table*}
	\centering
\begin{tabular}{|c|c|c|c|}
	\hline
	$ 70 $ & $W_I\oplus W_7\oplus W_{11} $ & $ x_2\oplus x_4\oplus x_5\oplus x_7\oplus x_{11} $ &$ x_1 $ is absent\\
	\hline
	$ 71 $ & $ W_6\oplus W_7\oplus W_{11} $ & $x_1\oplus x_3\oplus x_6\oplus x_8\oplus x_{11} $ &  $ x_3 $ not in side-information\\
	\hline
	$ 72 $ & $ W_I\oplus W_6\oplus W_7\oplus W_{11} $ & $ x_2\oplus x_4\oplus x_5\oplus x_6\oplus x_8\oplus x_{11} $ & $ x_1 $ is absent\\
	\hline
	$ 73 $ & $ W_8\oplus W_{11} $ & $ x_1\oplus x_4\oplus x_8\oplus x_9\oplus x_{11} $ &  $ x_4 $ not in side-information\\
	\hline
	$ 74 $ & $ W_I\oplus W_8\oplus W_{11} $ & $ x_2\oplus x_3\oplus x_5\oplus x_8\oplus x_9\oplus x_{11} $ & $ x_1 $ is absent \\
	\hline
	
	$ 75 $ & $ W_6\oplus W_8\oplus W_{11} $ & $  x_{1} \oplus x_4 \oplus x_{6}\oplus x_7\oplus x_9\oplus x_{11} $ & $ x_4 $ not in side-information\\
	\hline
	$ 76 $ & $ W_I\oplus W_6\oplus W_8\oplus W_{11} $ & $x_2\oplus x_3\oplus x_5\oplus x_6\oplus x_7\oplus x_9\oplus x_{11}$ & $ x_1 $ is absent \\
	\hline
$ 77 $ & $ W_7\oplus W_8\oplus W_{11} $ & $ x_1\oplus x_3\oplus x_4\oplus x_7\oplus x_8\oplus x_9\oplus x_{11} $ &  $ x_3 $ not in side-information\\
\hline
$ 78 $ & $W_I\oplus W_7\oplus W_8\oplus W_{11} $ & $ x_2\oplus x_5\oplus x_7\oplus x_8\oplus x_9\oplus x_{11} $ & $ x_1 $ is absent\\
\hline
$ 79 $ & $ W_6\oplus W_7\oplus W_8\oplus W_{11} $ & $x_1\oplus x_3\oplus x_4\oplus x_6\oplus x_9\oplus x_{11} $ &  $ x_3 $ not in side-information\\
\hline
$ 80 $ & $W_I \oplus W_6\oplus W_7\oplus W_8\oplus W_{11}  $ & $ x_2\oplus x_5\oplus x_6\oplus x_9\oplus x_{11} $ & $ x_1 $ is absent\\
\hline
$ 81 $ & $ W_9\oplus W_{11} $ & $ x_1\oplus x_5\oplus x_8\oplus x_{9}\oplus x_{11} $ &  $ x_5 $ not in side-information\\
\hline
$ 82 $ & $ W_I\oplus W_9\oplus W_{11} $ & $  x_2\oplus x_3\oplus x_4\oplus x_8\oplus x_{9}\oplus x_{11} $ & $ x_1 $ is absent\\
\hline
$ 83 $ & $ W_6\oplus W_9\oplus W_{11} $ & $ x_1\oplus x_5\oplus x_6\oplus x_7 \oplus x_9\oplus x_{11} $ &  $ x_5 $ not in side-information\\
\hline
$ 84 $ & $ W_I\oplus W_6\oplus W_9\oplus W_{11} $ & $x_2\oplus x_3\oplus x_4\oplus x_6\oplus x_7 \oplus x_9\oplus x_{11} $ & $ x_1 $ is absent\\
\hline
$ 85 $ & $ W_7\oplus W_9\oplus W_{11} $ & $ x_1\oplus x_3\oplus x_5\oplus x_7\oplus x_8\oplus x_{9}\oplus x_{11} $ &  $ x_3 $ not in side-information\\
\hline
$ 86 $ & $ W_I\oplus W_7\oplus W_9\oplus W_{11} $ & $ x_2\oplus x_4\oplus  x_7\oplus x_8\oplus x_{9}\oplus x_{11} $ & $ x_1 $ is absent\\
\hline
$ 87 $ & $ W_6\oplus W_7\oplus W_9\oplus W_{11} $ & $ x_1\oplus x_3\oplus x_5\oplus x_6\oplus x_9\oplus x_{11} $ &  $ x_3 $ not in side-information\\
\hline
$ 88 $ & $ W_I\oplus W_6\oplus W_7\oplus W_9\oplus W_{11} $ & $ x_2\oplus x_4\oplus x_6\oplus x_9\oplus x_{11} $ & $ x_1 $ is absent\\
\hline
$ 89 $ & $ W_8\oplus W_9\oplus W_{11} $ & $ x_1\oplus x_4\oplus x_5\oplus x_{11} $ &  $ x_4 $ not in side-information\\
\hline
$ 90 $ & $ W_I\oplus W_8\oplus W_9\oplus W_{11} $ & $ x_2\oplus x_3\oplus x_{11} $ & $ x_1 $ is absent\\
\hline
$ 91 $ & $ W_6\oplus W_8\oplus W_9\oplus W_{11} $ & $x_1\oplus x_4\oplus x_5\oplus x_6 \oplus x_7\oplus x_8\oplus x_{11} $ &  $ x_4 $ not in side-information\\
\hline
$ 92 $ & $ W_I\oplus W_6\oplus W_8\oplus W_9\oplus W_{11} $ & $ x_2\oplus x_3\oplus x_6\oplus x_7\oplus x_8\oplus x_{11} $ & $ x_1 $ is absent\\
\hline
$ 93 $ & $ W_7\oplus W_8\oplus W_9\oplus W_{11} $ &  $ x_1\oplus x_3\oplus x_4\oplus x_5\oplus x_7\oplus x_{11} $ & $ x_3 $ not in side-information\\
\hline
$ 94 $ & $ W_I\oplus W_7\oplus W_8\oplus W_9\oplus W_{11}  $ & $ x_2\oplus x_7\oplus x_{11} $ & $ x_1 $ is absent\\
\hline
$ 95 $ & $ W_6\oplus W_7\oplus W_8\oplus W_9\oplus W_{11} $ & $x_1\oplus x_3\oplus x_4\oplus x_5\oplus x_6\oplus x_8\oplus x_{11} $ &  $ x_3 $ not in side-information\\
\hline
$ 96 $ & $ W_I\oplus W_6\oplus W_7\oplus W_8\oplus W_9\oplus W_{11} $ & $ x_2\oplus x_6\oplus x_8\oplus x_{11} $ & $ x_1 $ is absent \\
\hline
$ 97 $ & $ W_{10}\oplus W_{11} $ & $ x_1\oplus x_3\oplus x_{10} $ & $ x_3 $ not in side-information\\
\hline
$ 98 $ & $ W_I\oplus W_{10}\oplus W_{11}  $ & $ x_2\oplus x_4\oplus x_5\oplus x_{10} $ &  $ x_1 $ is absent\\
\hline
$ 99 $ & $ W_6\oplus W_{10}\oplus W_{11} $ & $ x_1\oplus x_3\oplus x_6\oplus x_7\oplus x_8\oplus x_{10}$ & $ x_3 $ not in side-information\\
\hline
$ 100 $ & $ W_I\oplus W_6\oplus W_{10}\oplus W_{11} $ & $ x_2\oplus x_4\oplus x_5\oplus x_6\oplus x_7\oplus x_8\oplus x_{10} $ &  $ x_1 $ is absent\\
\hline
$ 101 $ & $ W_7\oplus W_{10}\oplus W_{11} $ & $  x_1\oplus x_{7}\oplus x_{10} $ & $ x_7 $ not in side-information\\
\hline
$ 102 $ & $ W_I\oplus W_7\oplus W_{10}\oplus W_{11}  $ & $  x_2\oplus x_3\oplus x_4\oplus x_5\oplus x_7\oplus x_{10} $ &  $ x_1 $ is absent\\
\hline
$ 103 $ & $ W_6\oplus W_7\oplus W_{10}\oplus W_{11} $ & $ x_1\oplus x_6\oplus x_8\oplus x_{10} $ & $ x_8 $ not in side-information\\
\hline
$ 104 $ & $ W_I\oplus W_6\oplus W_7\oplus W_{10}\oplus W_{11} $ & $  x_2\oplus x_3\oplus x_4\oplus x_5\oplus x_6\oplus x_{8}\oplus x_{11} $ &  $ x_1 $ is absent\\
\hline
$ 105 $ & $ W_8\oplus W_{10}\oplus W_{11} $ & $ x_1\oplus x_3\oplus x_4\oplus x_8\oplus x_9\oplus x_{10} $ & $ x_3 $ not in side-information\\
\hline
$ 106 $ & $ W_I\oplus W_8\oplus W_{10}\oplus W_{11}  $ & $ x_2\oplus x_5\oplus x_8\oplus x_9\oplus x_{10} $ &  $ x_1 $ is absent\\
\hline
$ 107 $ & $ W_6\oplus W_8\oplus W_{10}\oplus W_{11} $ & $ x_1\oplus x_3\oplus x_4 \oplus x_6\oplus x_7\oplus x_9\oplus x_{10} $ &$ x_3 $ not in side-information \\
\hline
$ 108 $ & $ W_I\oplus W_6\oplus W_8\oplus W_{10}\oplus W_{11} $ & $ x_2\oplus x_5\oplus x_6\oplus x_7\oplus x_9\oplus x_{10} $ &  $ x_1 $ is absent\\
\hline
$ 109 $ & $ W_7\oplus W_8\oplus W_{10}\oplus W_{11} $ & $  x_1\oplus x_4\oplus x_7\oplus x_8\oplus x_{9}\oplus x_{10} $ & $ x_4 $ not in side-information\\
\hline
$ 110 $ & $W_I\oplus W_7\oplus W_8\oplus W_{10}\oplus W_{11}  $ & $ x_2\oplus x_3\oplus x_5\oplus x_7\oplus x_8\oplus x_{9}\oplus x_{10} $ &  $ x_1 $ is absent\\
\hline
$ 111 $ & $ W_6\oplus W_7\oplus W_8\oplus W_{10}\oplus W_{11} $ & $ x_1\oplus x_4\oplus x_6\oplus x_9\oplus x_{10} $ & $ x_4 $ not in side-information\\
\hline
$ 112 $ & $W_I\oplus W_6\oplus W_7\oplus W_8\oplus W_{10}\oplus W_{11} $ & $ x_2\oplus x_3\oplus x_5\oplus x_6\oplus x_{9}\oplus x_{10} $ &  $ x_1 $ is absent\\
\hline
$ 113 $ & $ W_9\oplus W_{10}\oplus W_{11} $ & $ x_1\oplus x_3\oplus x_5\oplus x_8\oplus x_9\oplus x_{10} $ &$ x_3 $ not in side-information \\
\hline
$ 114 $ & $ W_I\oplus W_9\oplus W_{10}\oplus W_{11} $ & $ x_2\oplus x_4\oplus x_8\oplus x_9\oplus x_{10} $ &  $ x_1 $ is absent\\
\hline
$ 115 $ & $ W_6\oplus W_9\oplus W_{10}\oplus W_{11} $ & $ x_1\oplus x_3\oplus x_5\oplus x_6\oplus x_7\oplus x_9\oplus x_{10} $ & $ x_3 $ not in side-information\\
\hline
$ 116 $ & $ W_I\oplus W_6\oplus W_9\oplus W_{10}\oplus W_{11} $ & $x_2\oplus x_4\oplus x_6\oplus x_7\oplus x_9\oplus x_{10} $ &  $ x_1 $ is absent\\
\hline
$ 117 $ & $ W_7\oplus W_9\oplus W_{10}\oplus W_{11} $ & $ x_1\oplus x_5\oplus x_7\oplus x_8\oplus x_9\oplus x_{10} $ &$ x_5 $ not in side-information \\
\hline
$ 118 $ & $ W_I\oplus W_7\oplus W_9\oplus W_{10}\oplus W_{11} $ & $ x_2\oplus x_3\oplus x_4\oplus x_7\oplus x_8\oplus x_9\oplus x_{10} $ &  $ x_1 $ is absent\\
\hline
$ 119 $ & $ W_6\oplus W_7\oplus W_9\oplus W_{10}\oplus W_{11} $ & $ x_1\oplus x_5\oplus x_6\oplus x_9\oplus x_{10} $ &$ x_5 $ not in side-information \\
\hline
$ 120 $ & $ W_I\oplus W_6\oplus W_7\oplus W_9\oplus W_{10}\oplus W_{11} $ & $ x_2\oplus x_3\oplus x_4\oplus x_6\oplus x_9\oplus x_{10} $ &  $ x_1 $ is absent\\
\hline
$ 121 $ & $ W_8\oplus W_9\oplus W_{10}\oplus W_{11} $ & $ x_1\oplus x_3\oplus x_4\oplus x_5\oplus x_{10} $ & $ x_3 $ not in side-information\\
\hline
$ 122 $ & $ W_I\oplus W_8\oplus W_9\oplus W_{10}\oplus W_{11} $ & $ x_2\oplus x_{10} $ &  $ x_1 $ is absent\\
\hline
$ 123 $ & $ W_6\oplus W_8\oplus W_9\oplus W_{10}\oplus W_{11} $ & $ x_1\oplus x_3\oplus x_4\oplus x_5\oplus x_6\oplus x_7\oplus x_8\oplus x_{10} $ & $ x_3 $ not in side-information\\
\hline
$ 124 $ & $ W_I\oplus W_6\oplus W_8\oplus W_9\oplus W_{10}\oplus W_{11} $ & $ x_2\oplus x_6\oplus x_7\oplus x_8\oplus x_{10} $ &  $ x_1 $ is absent\\
\hline
$ 125 $ & $ W_7\oplus W_8\oplus W_9\oplus W_{10}\oplus W_{11} $ & $ x_1\oplus x_4\oplus x_5 \oplus x_7\oplus x_{10} $ & $ x_4 $ not in side-information \\
\hline
$ 126 $ & $ W_I\oplus  W_7\oplus W_8\oplus W_9\oplus W_{10}\oplus W_{11} $ & $ x_2\oplus x_3\oplus x_7\oplus x_{10} $ &  $ x_1 $ is absent\\
\hline
$ 127 $ & $ W_6\oplus  W_7\oplus W_8\oplus W_9\oplus W_{10}\oplus W_{11} $ & $  x_1\oplus x_4\oplus x_5\oplus x_6 \oplus x_8\oplus x_{10} $ & $ x_4 $ not in side-information \\
\hline
$ 128 $ & $ W_I\oplus W_6\oplus  W_7\oplus W_8\oplus W_9\oplus W_{10}\oplus W_{11} $ & $ x_2\oplus x_3\oplus x_6\oplus x_8\oplus x_{10} $ &  $ x_1 $ is absent\\
\hline
\end{tabular}\\
\caption{Table showing that decoding $ x_1 $ is not possible using the code obtained by using \textit{Construction} $ 1 $ on $ \mathcal{G}_8 $.}
\label{tablef}
\end{table*}
\end{ex}

\section{Discussion}

For the IC structure $ \mathcal{G}$ given in Fig. $ 2 $ of \cite{VaR} it has been shown that the index code obtained by using \textit{Construction} $ 1 $ is not decodable using \textit{Algorithm} $ 1 $. This is supported by the fact that the conditions \textit{c}$ 1 $ and \textit{c}$ 2 $ are violated as shown in the Table \ref{table_a} and Table \ref{table_b} respectively.
\begin{table}[h!]
        \centering
\begin{tabular}{|c|c|c|c|}
        \hline
        $ T_i $ & $ V_{NI}(i) $ & $ j \in V_{NI}(T_i)\backslash N^+_{T_i}(i) $ & $ a_{i,j} $ \\
        \hline\hline
        $ T_1 $ & $ \lbrace7,9,10,11,12,13,14\rbrace $ & $ \lbrace9,10,11,12,13\rbrace $ & $ 2 $, $ 1 $, $ 1 $, $ 1 $, $ 1 $ \\
        \hline
        $ T_2 $ & $ \lbrace10,11,12,13\rbrace $ & $ \lbrace11,12,13\rbrace $ & $ 1 $, $ 1 $, $ 1 $ \\
        \hline
        \end{tabular}\\
\caption{Table that verifies \textit{c}$ 1 $ for $ \mathcal{G} $.}
\label{table_a}
\end{table}
\begin{table}[h!]
        \centering
\begin{tabular}{|c|c|c|c|}
        \hline
        $ T_i $ & $ V_{NI}(i) $ & $ j \in V(\mathcal{G})\backslash V_{T_i}$ & $ b_{i,j} $ \\
        \hline\hline
        $ T_1 $ & $ \lbrace6,7,9,10,11,12,13,14\rbrace $ & $ \lbrace8\rbrace $ & $0$\\
        \hline
        $ T_2 $ & $ 10,11,12,13 $ & $ \lbrace7,8,9,14\rbrace $ & $ 0 $, $ 0 $, $ 1 $,$ 0 $ \\
        \hline
\end{tabular}\\
\caption{Table that verifies \textit{c}$ 2 $ for $ \mathcal{G} $.}
\label{table_b}
\end{table}
The index code obtained by using \textit{Construction} $ 1 $ on $ \mathcal{G}$ is
$W_I =x_1\oplus x_2\oplus x_3\oplus x_4\oplus x_5\oplus x_6; ~~W_7=x_7 \oplus x_2 \oplus x_6; ~~ 
W_8=x_8 \oplus x_3 \oplus x_5; ~~ W_9=x_9\oplus x_{10}; ~~ W_{10}=x_{10}\oplus x_{11}; ~~ 
W_{11}=x_{11}\oplus x_4 \oplus x_{12}; ~~ W_{12}=x_{12}\oplus x_{13}; ~~ W_{13}=x_{13}\oplus x_{5} \oplus x_{9}; ~~ 
W_{14}=x_{14}\oplus x_{9}.$
Even though the code is not decodable by \textit{Algorithm} $ 1 $, it is possible for the messages to be decoded using some other linear combinations of index code symbols as
\begin{itemize}
\item $ x_1 $ is decoded using
\begin{eqnarray*}
Z'_1 &=& W_I \oplus W_7\oplus W_9 \oplus W_{10}\oplus W_{11}\oplus W_{12}\oplus W_{13}\\
&=& x_1 \oplus x_3 \oplus x_7.
\end{eqnarray*}
\item  $ x_2 $ is decoded using
\begin{eqnarray*}
        Z'_2 &=& W_I \oplus W_9 \oplus W_{10}\oplus W_{11}\oplus W_{12}\oplus W_{13}\\
        &=& x_1 \oplus x_2 \oplus x_3 \oplus x_6.
        \end{eqnarray*}
\item The remaining messages are decodable using \textit{Algorithm} $ 1 $.
\end{itemize}

In \cite{TOJ}, it is claimed that the the index code obtained by \textit{Construction} $1$ for an IC structure is a valid index code by proposing \textit{Algorithm} $ 1 $ for decoding. Since \textit{Algorithm} $ 1 $ works only for a class of IC structures (those that satisfy the conditions \textit{c}$1$ and \textit{c}$2$), the validity of index code obtained by \textit{Construction} $ 1 $ for an arbitrary IC structure is now an open problem. Also from \textit{Theorem} $ 2 $, it is clear that an IC structure which has no cycles containing only non-inner vertices satisfies \textit{c}$ 1 $ and \textit{c}$ 2 $. Hence, along with the proof of \textit{Theorem} $ 3 $ in \cite{TOJ}, the proof of optimality of index codes obtained by using \textit{Construction} $ 1 $ on IC structures which do not contain cycles consisting of only non-inner vertices holds. 

From the example of Fig.$2$ in \cite{VaR} discussed at the beginning of this section and Examples 1,2 and $8$ in this paper the following  directions for further research arise:
\begin{itemize}
\item Characterize the IC structures for which there exists no decoding algorithm that uses only linear combinations of the index code symbols for the codes constructed using \textit{Construction} 1 (like in Example \ref{exam8}). 
\item Characterize the IC structures for which the \textit{Algorithm} 1 does not work for the code constructed using \textit{Construction} 1 but there exists decoding algorithm for the code which use only linear combinations of index code symbols (like in Examples 1 and $2$ in this paper and Fig.$2$ in \cite{VaR}). 
\item Identify the IC structures apart from those in Theorem \ref{thm2} for which the codes obtained by \textit{Construction} 1 is decodable using \textit{Algorithm} 1. 
\end{itemize}

\section*{Acknowledgements}
This work was supported partly by the Science and Engineering Research Board (SERB) of Department of Science and Technology (DST), Government of India, through J. C. Bose National Fellowship to Professor B. Sundar Rajan.


\end{document}